\documentclass[aps, amsmath, amssymb, aps, pre,showkeys, twocolumn,groupedaddress,superscriptaddress,nofootinbib,longbibliography]{revtex4-1}

\usepackage{multirow}

\usepackage{enumitem}
\setlist{noitemsep,topsep=0pt,parsep=0pt,partopsep=0pt}
\usepackage{array}

\usepackage[english]{babel}

\usepackage[colorlinks,linkcolor=black,citecolor=blue,urlcolor=black]{hyperref}
\usepackage{amsmath}
\usepackage{tabularx}
\usepackage{amsthm}

\usepackage{graphics}
\usepackage{graphicx}
\usepackage{xcolor}
\usepackage[config, labelfont={sf,bf}, textfont=sf]{subfig}
\graphicspath{{./FIG/}}
\usepackage[suffix=, outdir=./FIG/]{epstopdf}
\usepackage{hhline}

\usepackage{lineno}

\modulolinenumbers[5]

\newtheorem{theorem}{Theorem}
\newtheorem{lemma}{Lemma}
\newtheorem{corollary}{Corollary}

\newtheorem{remark}{Remark}

\begin{document}

\title{Analytical and numerical study of the hidden boundary of practical stability: \\
complex versus real Lorenz systems}

\author{N.V. Kuznetsov}
\email[]{Corresponding author: nkuznetsov239@gmail.com}
\affiliation{Faculty of Mathematics and Mechanics, St. Petersburg State University,
Peterhof, St. Petersburg, Russia}
\affiliation{Faculty of Information Technology,
University of Jyv\"{a}skyl\"{a}, Jyv\"{a}skyl\"{a}, Finland}
\affiliation{Institute for Problems in Mechanical Engineering RAS, Russia}

\author{T.N. Mokaev}
\affiliation{Faculty of Mathematics and Mechanics, St. Petersburg State University,
Peterhof, St. Petersburg, Russia}

\author{A.A.-H. Shoreh}
\affiliation{Faculty of Mathematics and Mechanics, St. Petersburg State University,
Peterhof, St. Petersburg, Russia}
\affiliation{Department of Mathematics, Faculty of Science, Al-Azhar University, Assiut, Egypt}

\author{A. Prasad}
\affiliation{Department of Physics \& Astrophysics, Delhi University, India}

\author{M.D. Shrimali}
\affiliation{Central University of Rajasthan, Ajmer, India}

\date{\today}


\begin{abstract}
This work presents the continuation of the recent article ''The Lorenz system: hidden boundary
of practical stability and the Lyapunov dimension''~\cite{KuznetsovMKKL-2020},
published in the Nonlinear Dynamics journal.
In this work, in comparison with the results
for classical real-valued Lorenz system (henceforward -- Lorenz system),
the problem of analytical and numerical
identification of the boundary of global stability for the complex-valued Lorenz system
(henceforward -- complex Lorenz system) is studied.
As in the case of the Lorenz system,
to estimate the inner boundary of global stability
the possibility of using the mathematical apparatus of Lyapunov functions
(namely, the Barbashin-Krasovskii and LaSalle theorems) is demonstrated.
For additional analysis of homoclinic bifurcations in complex Lorenz system
a special analytical approach by Vladimirov is utilized.
To outline the outer boundary of global stability and identify the
so-called hidden boundary of global stability,
possible birth of hidden attractors and transient chaotic sets
is analyzed.

\end{abstract}

\maketitle
\section{Introduction}

In the study of the dynamics of applied systems, an important role is played
by the property of \emph{global stability}\footnote{
  We use the term ``global stability'' for simplicity
  of further presentation, while in the literature there
  are used different notions, e.g.
  ``globally asymptotically stable'' \cite[p.~137]{Vidyasagar-1978}, \cite[p.~144]{HaddadC-2011},
  ``gradient-like'' \cite[p.~2]{LeonovRS-1992}, \cite[p.~56]{YakubovichLG-2004},
  ``quasi-gradient-like'' \cite[p.~2]{LeonovRS-1992}, \cite[p. 56]{YakubovichLG-2004}
  and others, reflecting the features of
  the stationary set and the convergence of trajectories to it.
},
when there are no oscillations in the system, and all trajectories are attracted to a stationary set.
In some cases, the global stability property can be established analytically
(e.g. by Lyapunov-type, either frequency-domain methods,
see~\cite{BarbashinK-1952,Lasalle-1960,LeonovPS-1996,KuznetsovLYYKKRA-2020-ECC}),
and the global stability boundary
can be constructed in the parameter space.
Crossing this analytical boundary can lead to the following scenarios.
After crossing the analytical border of global stability,
the system can still exhibits \emph{practical}\footnote{
  Let us remark here a different notion of \emph{practical} stability,
  suggested and developed in \cite{LaSalleL-1961,LakshmikanthamLM-1990}.
  According to this notion, a unique equilibrium of a perturbed system
  is practically stable, if for all bounded perturbations
  and all initial points from a vicinity
  $Q_0$ of equilibrium, corresponding trajectories stay inside a
  larger vicinity $Q \supset Q_0$ for all moments of time;
  at the same time, the equilibrium can be Lyapunov unstable.
} global stability~\cite{KuznetsovMKKL-2020},
i.e. almost all trajectories are attracted to a stationary set,
however, in this case a countable number of orbits may appear,
which nevertheless are not realized in practice and in standard
numerical experiments (e.g., unstable periodic or
homoclinic orbits~\cite{AfraimovicBSh-1977,AuerbachCEGP-1987,Cvitanovic-1991,ShilnikovTCh-2001,Leonov-2013-IJBC,LeonovKM-2015-EPJST,LeonovMKM-2020-IJBC}).
On the other hand, crossing the analytical boundary of global stability
can lead to the loss of stability of the stationary set.
Then, if the system is dissipative in the sense Levinson
(\cite{Levinson-1944}, see also~\cite{LeonovKM-2015-EPJST}),
but all stationary points are unstable, then inside the absorbing set
a self-excited nontrivial attractor can be observed and easily localized.
Finally, after crossing the boundary of global stability,
the stationary set can remain locally stable (or partially locally stable),
but at the same time \emph{hidden oscillations} (periodic or chaotic)~\cite{LeonovKV-2011-PLA,LeonovK-2013-IJBC,LeonovKM-2015-EPJST,Kuznetsov-2016,Kuznetsov-2020-TiSU}
may appear in the phase space, which basins of attraction do not intersect
with arbitrarily small open vicinities of equilibria.
This part of border of global stability, accordingly, can be also called
"hidden"~\cite{Kuznetsov-2020-TiSU,KuznetsovMKKL-2020}.

This work is the continuation of the recent article
''The Lorenz system: hidden boundary
of practical stability and the Lyapunov dimension''~\cite{KuznetsovMKKL-2020},
where the analytical and numerical analysis of the boundary of
global stability was performed for the classical Lorenz system~\cite{Lorenz-1963}:
\begin{equation}\label{sys:lorenz}
\begin{cases}
 \dot{x} =  - \sigma(x-y),\\
 \dot{y} = r x - y - x z, \\
 \dot{z} = - b z + x y,
\end{cases}
\end{equation}
where parameters $r$, $\sigma > 0$, $b \in (0,4]$.
Here in the introduction, we briefly outline these results
to compare them further with the results for the complex Lorenz system.

\medskip
\renewcommand{\labelenumi}{\bf\arabic{enumi})}
\begin{enumerate}[leftmargin=0pt,labelwidth=-10pt,labelsep=0pt]
\setlength\itemsep{1em}
\item {\bf Dissipativity}.
The Lorenz system \eqref{sys:lorenz} is dissipative in the sense
of Levinson, since it has the absorbing set~\cite{LeonovB-1992}
\begin{multline}\label{eq:abs_set}
   \mathfrak{B}\!=\!\!\big\{\!(x,\!y,\!z) \!\in\! \mathbb{R}^{3} \, | \, x^{2}+y^{2}+(z-r-\sigma)^{2} \!\leq\!
   \tfrac{b(\sigma+r)^{2}}{2\min(1,\sigma,\frac{b}{2})}\!\big\}.
\end{multline}
This implies that all solutions of \eqref{sys:lorenz} exist for $t\in[0,+\infty)$ and,
thus, system \eqref{sys:lorenz} generates a dynamical system.
For $0 < r \leq 1$, the stationary set of
system~\eqref{sys:lorenz} consists of a unique stable equilibrium $S_0 = (0,0,0)$;
for $r > 1$ a pair of symmetric equilibria
$S_\pm = (\pm\sqrt{b(r-1)},\pm\sqrt{b(r-1)},r-1)$ is added to the stationary set,
and $S_0$ turns into a saddle.

\item {\bf Inner estimation of global stability}.
For the Lorenz system~\eqref{sys:lorenz}, combining several approaches
based on the construction of Lyapunov functions \cite{BarbashinK-1952,Lasalle-1960,Leonov-2018-UMZh-rus}
it is possible to prove the following
criterion for the absence of self-excited and hidden oscillations
(see Figs.~\ref{fig:lorenz:GS:inner:barbashin},~\ref{fig:lorenz:GS:inner:lasalle}).
\begin{theorem}\label{thm:gs}
If for parameters of system~\eqref{sys:lorenz} one of
the following cases holds:
\begin{equation}\label{eq:gs_condition}
  2 \sigma \leq b, \qquad \text{or} \qquad
  \begin{cases}
    2 \sigma > b, \\
    r < r_{\rm gs} = \frac{(\sigma + b) (b + 1)}{\sigma},
  \end{cases}
\end{equation}
then there are no nontrivial self-excited
and hidden oscillations in the phase space
of system~\eqref{sys:lorenz}, and
any its solution $(x(t),y(t),z(t))$
tends to the stationary set as $t \to +\infty$.
\end{theorem}

Beyond the estimate \eqref{eq:gs_condition} in Theorem~\ref{thm:gs},
the analysis of global stability and the birth of nontrivial attractors can be performed numerically.
It is further known that the separatrix of saddle $S_{0}$ can form a homoclinic loop
from which unstable cycles can arise and violate the global stability
(however, a set of measure zero does not affect the global attraction
on a stationary set from a practical point of view).
Using the Fishing principle \cite{Leonov-2013-IJBC,LeonovKM-2015-EPJST,LeonovMKM-2020-IJBC}
for the Lorenz system~\eqref{sys:lorenz}
it is possible to prove the following:
\begin{theorem}\label{thm:homo}
For $\sigma$ and $b$ fixed, there exists $r = r_{\rm h}\in (1,+\infty)$
corresponding to the homoclinic orbit of the
saddle equilibrium $S_0$ if and only if $3\sigma > 2 b + 1$.
\end{theorem}

For instance, for the classical values of parameters
$\sigma = 10$, $b = 8/3$ of system~\eqref{sys:lorenz}
it is possible to find numerically the approximate value
of such homoclinic bifurcation $r_{\rm h} \approx 13.926$,
when two symmetric homoclinic orbits forming a homoclinic butterfly
appear (see Fig.~\ref{fig:lorenz:GS:inner:homoclinic}). Further increase of
the parameter $r$ leads to the birth of two saddle periodic orbits
from each homoclinic orbit~\cite{ShilnikovTCh-2001}.
\begin{figure*}[!ht]
 \centering
 \subfloat[
 {\scriptsize $r < 1$}
 ] {
 \label{fig:lorenz:GS:inner:barbashin}
 \includegraphics[width=0.33\textwidth]{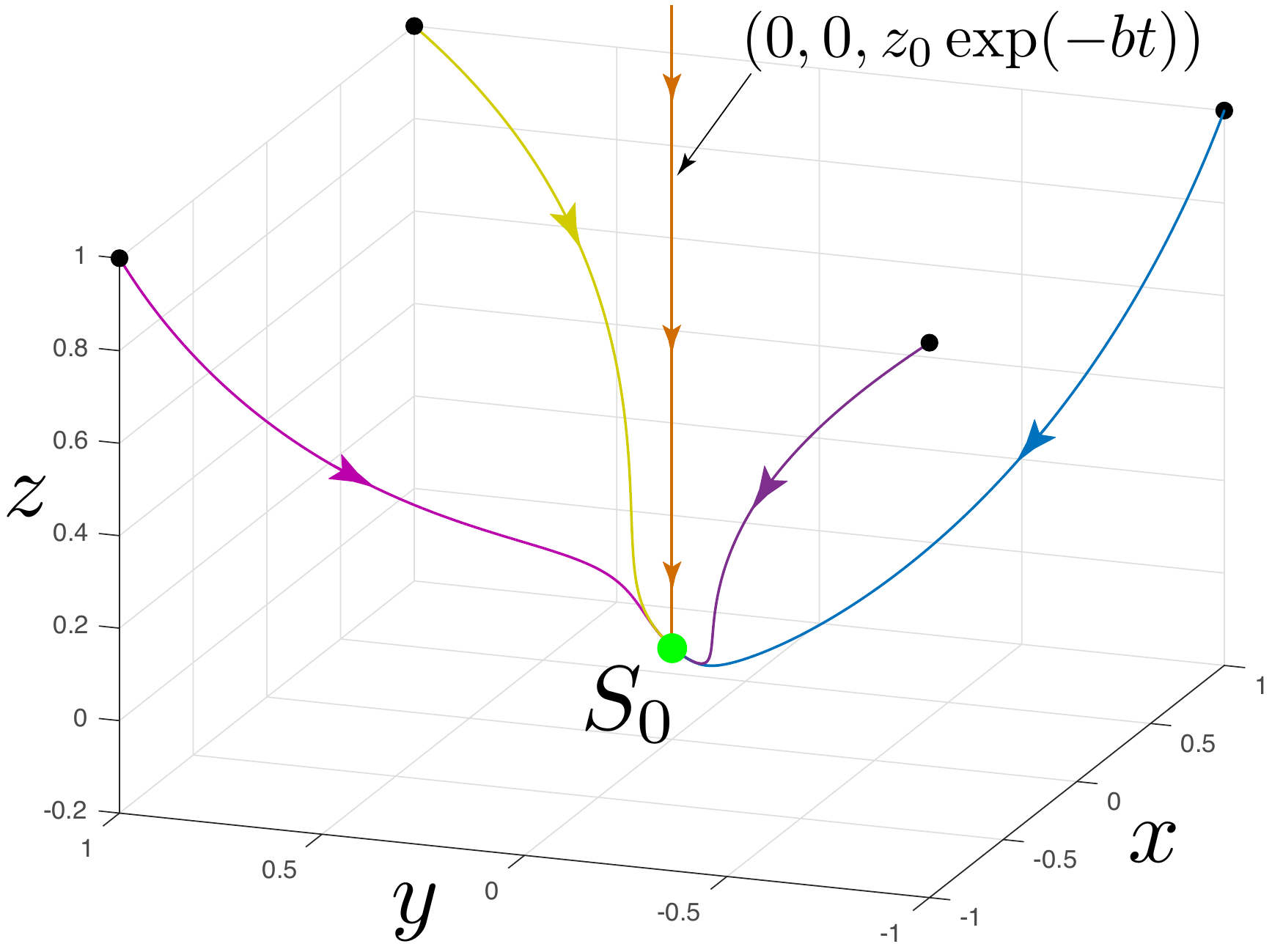}
 }~
 \subfloat[
 {\scriptsize $1 < r < r_{\rm gs} \approx 4.64$}
 ] {
 \label{fig:lorenz:GS:inner:lasalle}
 \includegraphics[width=0.33\textwidth]{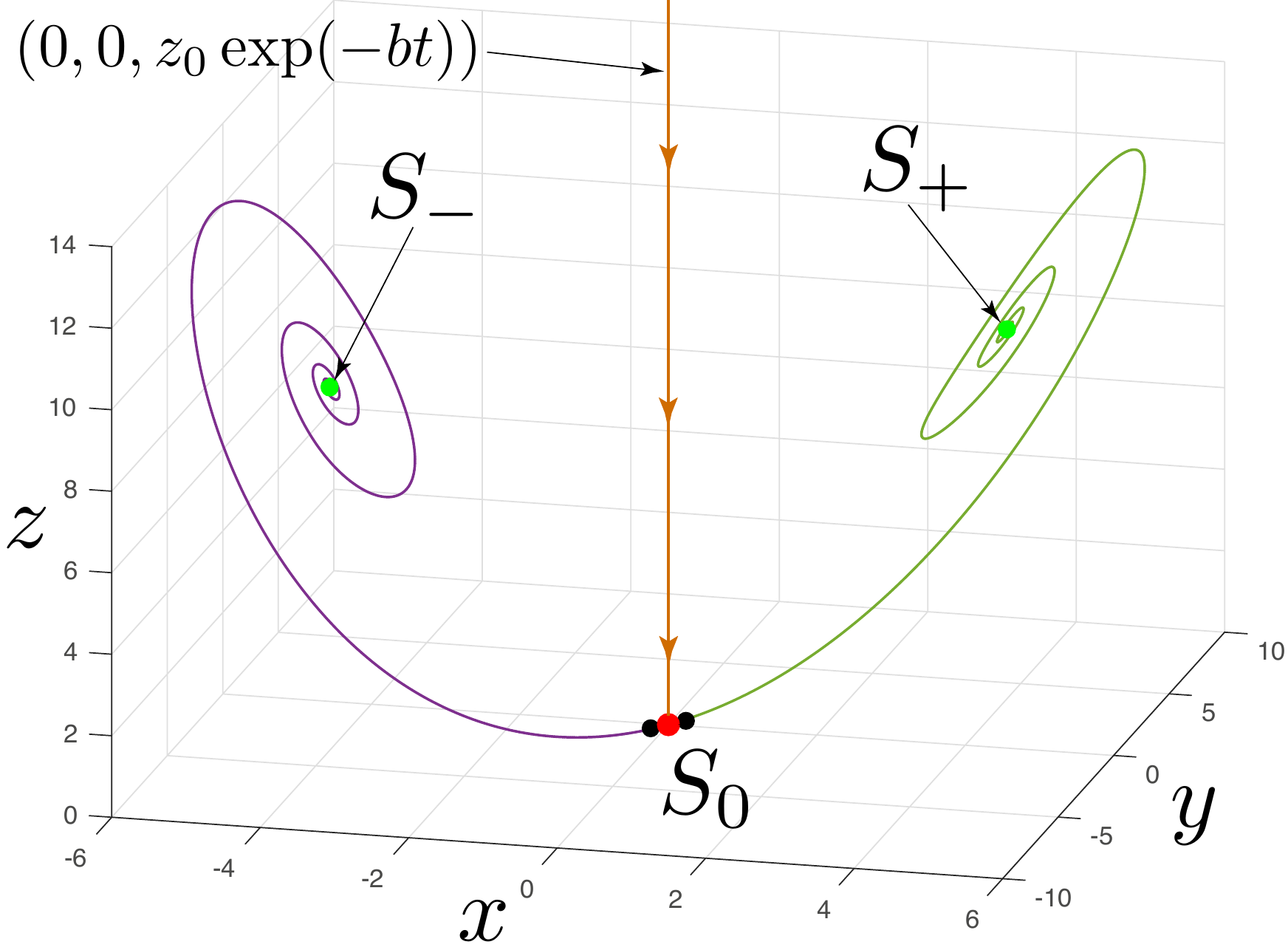}
 }
 \subfloat[{\scriptsize $r = r_{\rm h} \approx 13.926$}
 ] {
 \label{fig:lorenz:GS:inner:homoclinic}
 \includegraphics[width=0.33\textwidth]{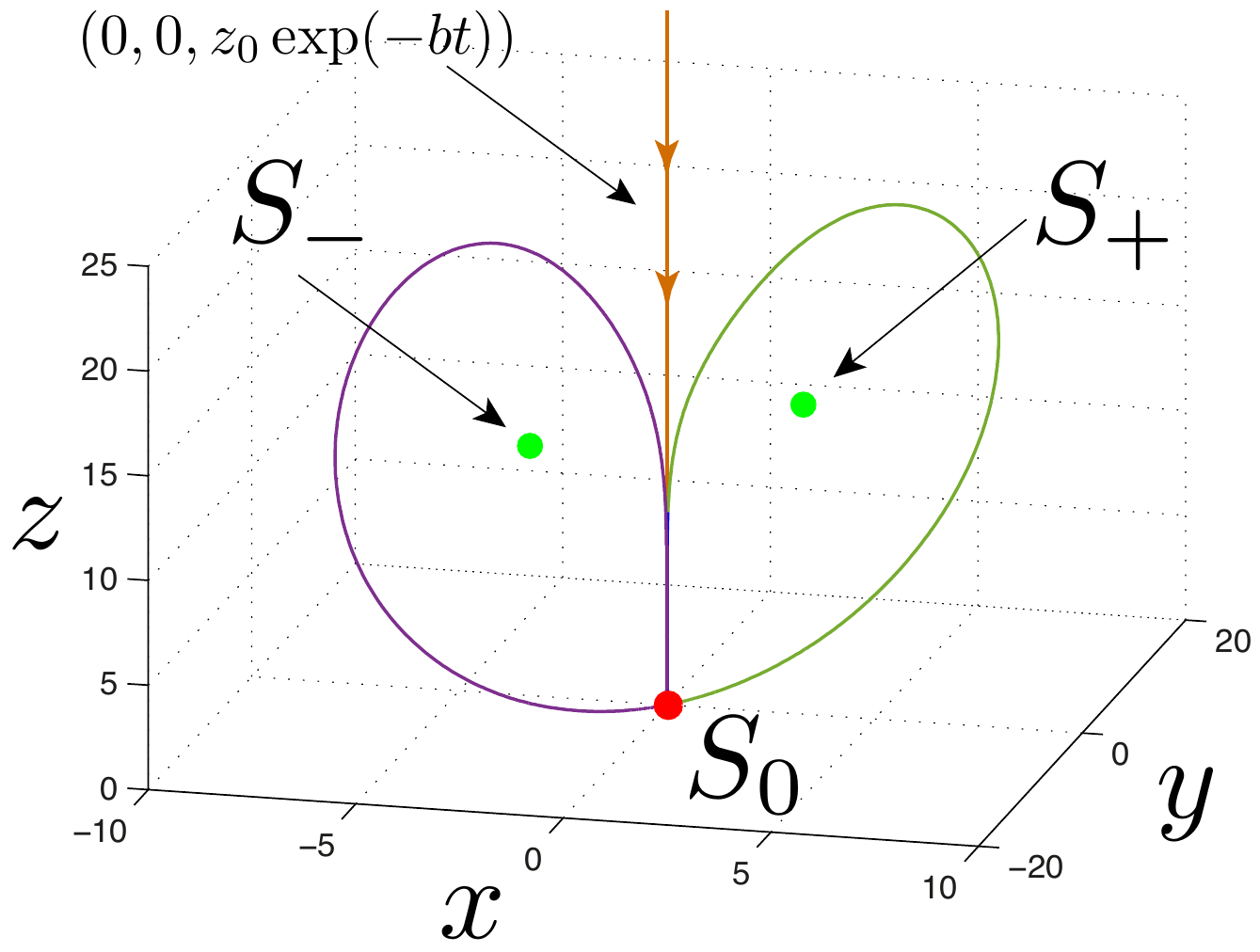}
 }
 \caption{
 Dynamics of the Lorenz system~\eqref{sys:lorenz} with fixed
 parameters $\sigma = 10$, $b = \tfrac{8}{3}$, when parameter $r$
 varying according to the inner estimation of global stability.
 }
 \label{fig:lorenz:GS:inner}
\end{figure*}

\item {\bf Outer estimation of global stability}.
For systems with a global absorbing set
and an unstable stationary set,
the existence of self-excited attractors is obvious.
From a computational perspective
this allows one to use a \emph{standard computational
procedure}, in which after a transient process,
a trajectory, starting from a point of unstable
manifold in a neighborhood of equilibrium, reaches
a state of oscillation, thus one can easily
identify it.

System~\eqref{sys:lorenz} possesses the absorbing set $\mathcal{B}$
(see Eq.~\eqref{eq:abs_set})
and for $\sigma > b~+~1$,
$r > r_{\rm cr} = \sigma (\tfrac{\sigma + b + 3}{\sigma - b - 1})$
all equilibria are unstable.
Thus, in this case, system~\eqref{sys:lorenz} has
a nontrivial self-excited attractor:
if we consider classical values of parameters
$\sigma = 10$, $b = 8/3$, then
for $r > r_{\rm cr}$, e.g. for $r = 28$,
it is possible to observe the self-excited chaotic attractor
with respect to all three equilibria $S_0$, $S_\pm$
(see Fig.~\ref{fig:lorenz:GS:outer:se_attr_wrt3eqs}).
\begin{figure*}[!ht]
 \centering
 \subfloat[
 {\scriptsize $r = 24$, hidden transient chaotic set}
 ] {
 \label{fig:lorenz:GS:outer:hid_trans_set}
 \includegraphics[width=0.32\textwidth]{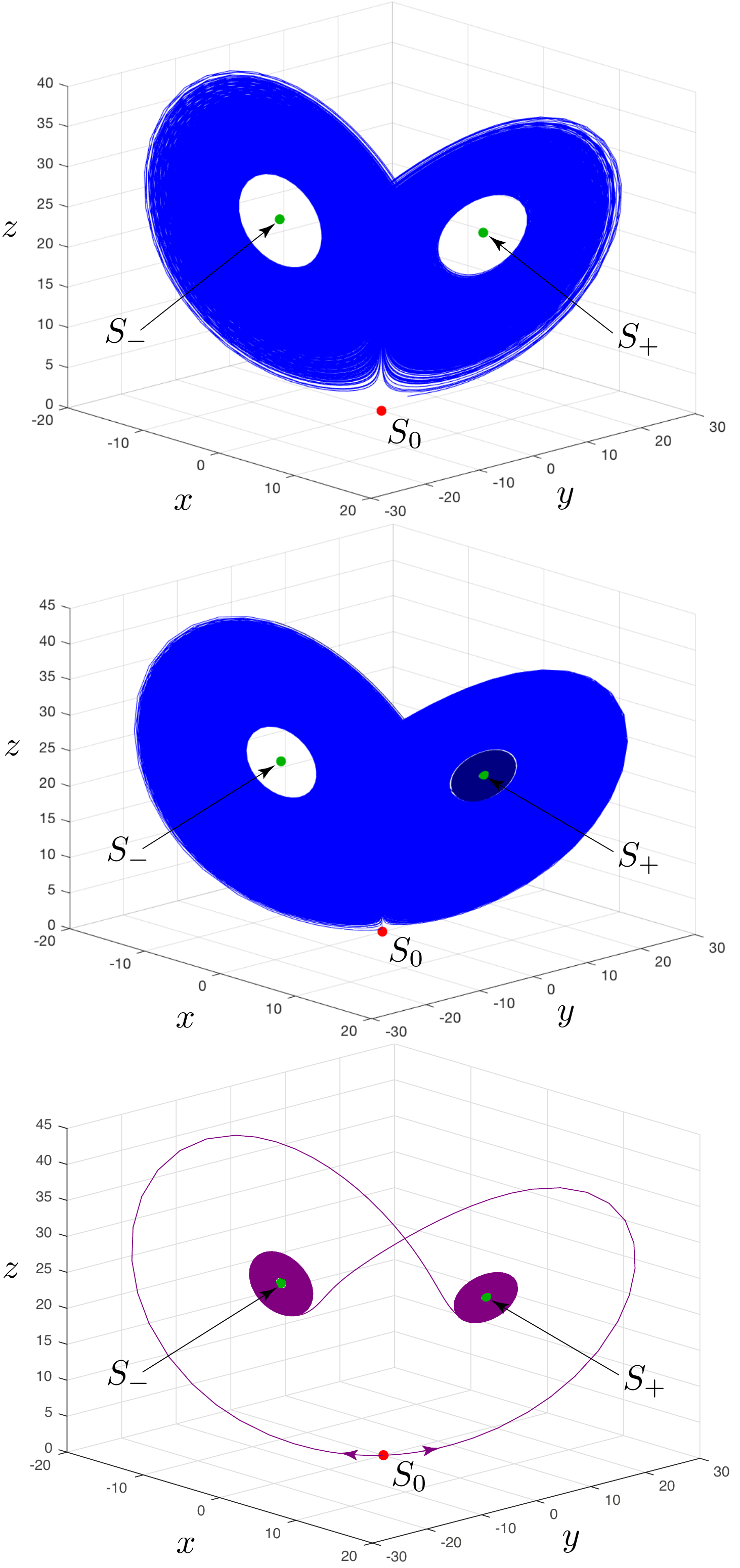}
 }~
 \subfloat[
 {\scriptsize $r = 24.5$, self-excited attractor w.r.t $S_0$}
 ] {
 \label{fig:lorenz:GS:outer:se_attr_wrt1eq}
 \includegraphics[width=0.32\textwidth]{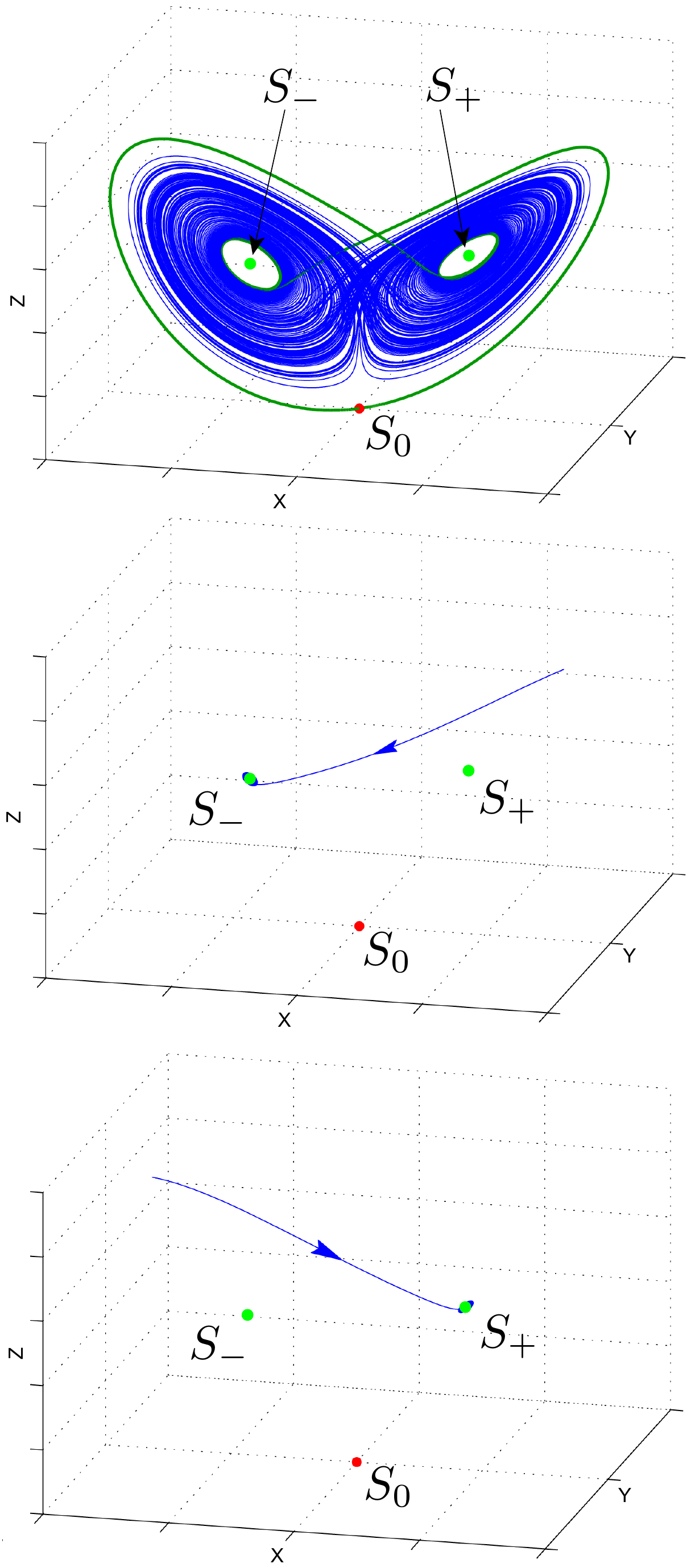}
 }~
 \subfloat[
 {\scriptsize $r = 28$, self-excited attractor w.r.t $S_0$, $S_\pm$}
 ] {
 \label{fig:lorenz:GS:outer:se_attr_wrt3eqs}
 \includegraphics[width=0.32\textwidth]{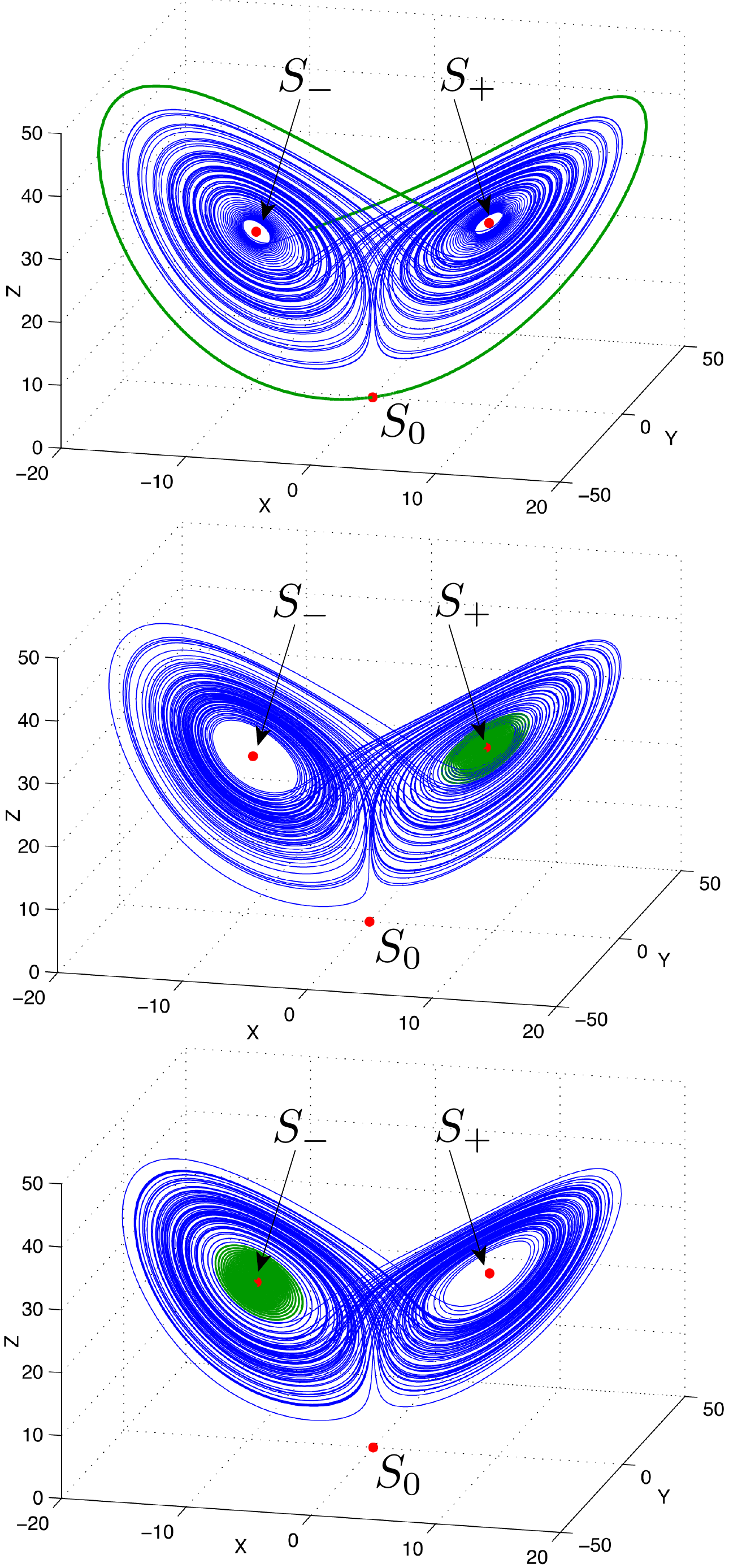}
 }
 \caption{
 Dynamics of the Lorenz system~\eqref{sys:lorenz} with fixed
 parameters $\sigma = 10$, $b = \tfrac{8}{3}$, when parameter $r$
 varying according to the outer estimation of global stability.
 }
 \label{fig:lorenz:GS:outer}
\end{figure*}

\item {\bf Boundary of practical global stability}.
The presence of an absorbing set
implies the existence of a globally attractor $\mathcal{A}_{\rm glob}$,
which contains all local self-excited and hidden attractors, and stationary set.
Thus, inside the set $\mathcal{B}$
it is possible to study numerically
the presence of nontrivial self-excited and hidden attractors
for parameters $r$, $\sigma$, $b$ not satisfying
conditions~\eqref{eq:gs_condition} of global stability, i.e.
by fixing $\sigma$ and $b$, and by decreasing $r$ from $r_{\rm cr}$.
For $\sigma = 10$, $b = 8/3$,
this gives us the following region
$r \in (r_{\rm gs}, r_{\rm cr})$, where
$r_{\rm gs} \approx 4.64$, $r_{\rm cr} \approx 24.74$.

A nontrivial self-exited attractor can be observed
numerically for $24.06 \lessapprox r < r_{\rm cr} \approx 24.74$
(see e.g. \cite{Sparrow-1982}).
In this case of nontrivial multistability,
system~\eqref{sys:lorenz} possesses
a local chaotic attractor $\mathcal{A}$ which is
self-excited with respect to equilibrium $S_0$ and
co-exists with the trivial attractors $S_\pm$
(see Fig.~\ref{fig:lorenz:GS:outer:se_attr_wrt1eq}).

\item {\bf Hidden attractor or transient set?}
For the Lorenz system~\eqref{sys:lorenz}
it is still an open question~\cite[p.~14]{Kuznetsov-2016},
whether for some parameters there exists a hidden chaotic attractor,
i.e. whether it is possible by changing parameters
to disconnect
the basin of attraction from equilibria $S_0$, $S_\pm$
(e.g. for the parameters $\sigma = 10$, $b = \tfrac{8}{3}$:
if $r = 28$, then attractor is connected with $S_0$, $S_\pm$;
if $r = 24.5$, then attractor is connected with only $S_0$).
The current results on the existence of the hidden attractors
in the Lorenz system are the following:
recently reported hidden attractors in the Lorenz system
with $r < r_{\rm cr}$ and locally stable equilibria $S_\pm$
turn out to be a \emph{transient chaotic set}
(a set in the phase space, which can persist for a long time,
but after all collapses), but not a \emph{sustained} hidden chaotic attractor
\cite{YuanYW-2017-HA,MunmuangsaenS-2018-HA}.

In a numerical computation of a trajectory over a finite-time interval
it is difficult to distinguish a \emph{sustained chaos} from
a \emph{transient chaos}
\cite{GrebogiOY-1983,LaiT-2011}, thus it is reasonable to give
a similar classification for transient chaotic sets~\cite{DancaK-2017-CSF,ChenKLM-2017-IJBC}:
\emph{transient chaotic set} is a \emph{hidden} one
if it does not involve and attract trajectories
from a small neighborhood of equilibria;
otherwise, it is \emph{self-excited}.

For the Lorenz system \eqref{sys:lorenz}
with parameters $\sigma = 10$, $b = \tfrac{8}{3}$ fixed,
near the point $r \approx 24.06$ it is possible
to observe a long living transient chaotic set, which is hidden
since it's basin of attraction does not intersect
with the small vicinities of equilibrium $S_0$
(see Fig.~\ref{fig:lorenz:GS:outer:hid_trans_set}).

In our experiments, consider system \eqref{sys:lorenz}
with parameters $r = 24$, $\sigma = 10$, $b = 8/3$.
Using MATLAB's standard procedure \texttt{ode45}
with default parameters
(relative tolerance $10^{-3}$, absolute tolerance $10^{-6}$)
for trajectory of system~\eqref{sys:lorenz}
with initial point $u_0 = (20,20,20)$
a transient chaotic behavior is observed
on the time interval $[0, \, 1.8 \cdot 10^{4}]$,
for initial point $u_0 = (-7,8,22)$ --- on the time interval $[0,\,7.2\!\cdot\!10^{4}]$,
for initial point $u_0 = (2,2,2)$ --- on the time interval $[0,\,2.26\!\cdot\!10^{5}]$,
and for initial point $u_0 = (0, -0.5, 0.5)$ a transient chaotic behavior
continues over a time interval of more than $[0,\,10^{7}]$.
Remark that, if we consider the same initial points, but
use MATLAB's procedure \texttt{ode45}
with relative tolerance $10^{-6}$, for all these initial points
the chaotic transient behavior will last
over a time interval of more than $[0,\, 10^6]$,
and corresponding transient chaotic sets won't collapse.
\end{enumerate}

\section{The complex Lorenz system}

The classical Lorenz system~\eqref{sys:lorenz}
has gained and still attracts much attention of scientists\footnote{
  The original celebrated work by Lorenz~\cite{Lorenz-1963} has gained more than
  24000 cites according to Google Scholar.
},
however, it is a rather simple mathematical model of the fluid dynamics problem,
which E.~Lorenz has initially investigated.
A more complicated model of many real-world physical problems,
such as two-layer baroclinic instability with weak viscosity and beta-effect
\cite{GibbonM-1982,FowlerGM-1982,FowlerGM-1983}
and detuned laser optics \cite{Ning-1990,NeimarkL-1992},
is described by the following complex Lorenz system:
\begin{equation}\label{sys:complex-Lorenz}
  \left\{
    \begin{array}{ll}
      \dot{X} = \sigma (Y-X),  \\
      \dot{Y} = r X -a Y-X Z,  \\
      \dot{Z} = -b Z + \frac{1}{2}( X^* Y + X Y^*),
    \end{array}
  \right.
\end{equation}
where $X = x_{1} + x_{2} i$, $Y = x_{3} + x_{4} i$ are complex variables,
$Z = x_{5}$ is real, and $"*"$ denotes complex conjugation operator.
Parameters $\sigma$, $b > 0$ are real, $r = r_{1}+r_{2}i$ and $a=1-ei$ are complex
with $r_{1}, e \in \mathbb{R}$.
Note that since the transformation
$X \rightarrow X^*, Y \rightarrow Y^*$
changes only the sign of $r_{2}$ in \eqref{sys:complex-Lorenz},
this parameter may be chosen positive~\cite{RauhHA-1996}.

System (\ref{sys:complex-Lorenz}) is also dissipative in the
sense of Levinson, and the bounded convex absorbing set could
be constructed according to the following theorem~\cite{LeonovR-1987}.
\begin{theorem}
\label{dissipativity}
Let for $\lambda_{0} = \min(1,b,\sigma)$
there exist $\lambda \in (0, \lambda_{0})$,
$\gamma > 0$, $\vartheta$ such that
\begin{equation}\label{comp_lorenz:cond:dissipativity}
  \gamma(\sigma-\lambda)(1-\lambda)-\tfrac{1}{4}\mid\sigma+\gamma r-\vartheta \mid^{2} \geq 0.
\end{equation}
Define $\Gamma=\tfrac{\vartheta^{2}(b-2\lambda)^{2}}{8\lambda \gamma(b-\lambda)}$,
$\beta=\frac{\vartheta^{2}}{\gamma(1+\gamma)}\Big(\frac{(b-2\lambda)^{2}}{4\lambda (b-\lambda)}+1\Big)$
and the following function $W: \mathbb{C}\times \mathbb{C}\times \mathbb{R}\rightarrow \mathbb{R}$:
\[
  W(X,Y,Z)= \tfrac{1}{2}\big[|X|^{2}+\gamma(|Y|^{2}+Z^{2})\big]-\vartheta Z.
\]
Then for any solution $(X(t),Y(t),Z(t))$ of \eqref{sys:complex-Lorenz}
and any $\delta>0$ with $T = T\big(\delta, X(0), Y(0), Z(0)\big)$,
the following inequalities hold for all $t\geq T$:
 \begin{equation}\label{inequa1Lemma}
   W(X(t), Y(t), Z(t)) \leq \Gamma+\delta,
 \end{equation}
 and
\begin{equation}\label{inequa2Lemma}
   |X(t)|^{2} < \beta + \tfrac{4\delta}{1+\gamma}.
 \end{equation}
 \end{theorem}

Thus, if condition \eqref{comp_lorenz:cond:dissipativity} holds,
then the solutions of system~\eqref{sys:complex-Lorenz}
exist for $t\in[0,+\infty)$
and, thus, system~\eqref{sys:complex-Lorenz} generates a dynamical system and
possesses a global attractor~\cite{Chueshov-2002}
containing a set of all equilibria.

It is also useful to consider the following equivalent
form of system~\eqref{sys:complex-Lorenz}
in terms of real variables
$x_{1}, x_{2}, x_{3}, x_{4}, x_{5}$
and real parameters $r_{1}, r_{2}, \sigma, b, e$ (see, e.g.~\cite{RauhHA-1996}):
\begin{equation}\label{13}
  \left\{
    \begin{array}{ll}
      \dot{x}_{1} = \sigma(x_{3}-x_{1}),  \\
      \dot{x}_{2} = \sigma(x_{4}-x_{2}),  \\
      \dot{x}_{3} = r_{1} x_{1} -r_{2} x_{2}- x_{3}-e x_{4}-x_{1} x_{5},  \\
      \dot{x}_{4} = r_{2} x_{1}+r_{1} x_{2}+e x_{3} - x_{4}-x_{2} x_{5},  \\
      \dot{x}_{5} = -b x_{5} + x_{1}x_{3} + x_{2} x_{4}.  \\
    \end{array}
  \right.
\end{equation}
For
\begin{equation}\label{comp_lorenz:cond:unique_equil}
  {\rm Im}(r-a) \neq 0, \quad \text{or} \quad
  \left\{\begin{array}{l}
  {\rm Im}(r-a) = 0,\\
  {\rm Re}(r-a)\leq 0
  \end{array}\right.
\end{equation}
system~(\ref{13}) has a unique equilibrium $S_{0}=(0,0,0,0,0)$.
For
\begin{equation}\label{comp_lorenz:cond:circle_equil}
  \left\{\begin{array}{l}
  {\rm Im}(r-a) = 0,\\
  {\rm Re}(r-a) > 0
  \end{array}\right.
\end{equation}
system~(\ref{13}) has a stationary set containing
$S_{0}$ and a whole circle
of equilibria given by the expression:
$$
  x_{1}^{2}+x_{2}^{2}=b(r_{1}-1), \quad x_{1}=x_{3},\quad x_{2}=x_{4}, \quad x_{5}=r_{1}-1.
$$
These equilibria could be parameterized as follows:
\[
  S_{\theta} = (\pm\rho\cos(\theta),
                \pm\rho\sin(\theta),
                \pm\rho\cos(\theta),
                \pm\rho\sin(\theta),
                r_{1} - 1),
\]
where $\rho = \sqrt{b(r_{1}-1)}$ and $\theta\in[0,2\pi)$.
In literature, relation ${\rm Im}(r-a) = 0$
is often called the "laser case", since in this case
system~\eqref{sys:complex-Lorenz} describe the dynamics of some lasers,
e.g. a detuned laser~\cite{Ning-1990,RauhHA-1996}.

Let us outline the following statements about local stability of
the stationary set of system (\ref{13}) (see e.g.~\cite{FowlerGM-1982}).
\begin{lemma}\label{Lemma1}
The equilibrium $S_{0}$ of system (\ref{13})
is stable if and only if the following condition hold:
\[
  r_{1} < r_{1c}, \quad \text{where} \quad
  r_{1c} = 1 + \tfrac{(e +r_{2})(e - \sigma \, r_{2})}{(\sigma + 1)^{2}}.
\]
\end{lemma}

\begin{lemma}\label{Lemma2}
The equilibria $S_{\theta}$ of system~(\ref{13}) are stable\footnote{
  One could easily check that the eigenvalues of the Jacobi matrix at all equilibria $S_{\theta}$ are
  the same for any $\theta \in [0, 2\pi)$.
}
if and only if one of the following
conditions hold:
\[
  \sigma < b+1, \quad \text{or} \quad
  \left\{\begin{array}{l}
    \sigma > b+1, \\
    1<r_{1} < r'_{1c} = 1 + \frac{\sqrt{Q_{2}^{2} + 4 Q_{1} Q_{3}} - Q_{2}}{2 b Q_{1}},
  \end{array}\right.
\]
where
\[
\begin{aligned}
  Q_{1} = & (3\sigma+1)(\sigma-b-1), \\
  Q_{2} = & \gamma_2 (b + 2 \gamma_2) (-2 b \gamma_2 - b \gamma_3 + 2 \gamma_2 \gamma_3) \\
          & - \gamma_1 (b^2 + b \gamma_2 - b \gamma_3 + 2 \gamma_2^2 + 2 \gamma_2 \gamma_3),\\
  Q_{3} = & 2 \gamma_1 \gamma_2 \, b \, (b^2 + 2 b \gamma_2 + \gamma_1),
\end{aligned}
\]
and $\gamma_1 = (\sigma+1)^2+\big(\tfrac{2 \sigma (e+r_2)}{\sigma+1} - e\big)^2$,
$\gamma_2 = \sigma + 1$, $\gamma_3 = 2 \sigma$.
\end{lemma}

\section{Inner estimation: the global stability and trivial attractors}\label{Inner:estimation}
Constructing complex-valued Lyapunov functions,
by analogy with Theorem~\ref{thm:gs},
it is possible to derive the following
criterion for the absence of self-excited and hidden oscillations
in the complex Lorenz system \eqref{sys:complex-Lorenz}.
\begin{theorem}\label{theorem1}
If for parameters of the system~\eqref{sys:complex-Lorenz}
one of the following conditions holds:
\renewcommand{\labelenumi}{\bf\arabic{enumi})}
\begin{enumerate}[leftmargin=1pt,labelwidth=-10pt,labelsep=1pt]
  \item $|r + 1| < 2$ ,
  \item $4 r_1 < - r_2^2$ ,
  \item $2\sigma\!-\!b \!\neq\! 0$ and conditions~\eqref{comp_lorenz:cond:circle_equil} hold.
  If $2\sigma\!-\!b\!>\!0$, then for $\lambda_{0}=\min(1,b,\sigma)$,
  there exist $\lambda\in(0,\lambda_{0})$, $\gamma > 0$, $\vartheta$, such that
  condition on dissipativity \eqref{comp_lorenz:cond:dissipativity} and the following
  inequality hold:
  \begin{equation}\label{cond-glonal}
       \tfrac{\vartheta^{2}}{\gamma(1+\gamma)}
       \big(\tfrac{(b-2\lambda)^{2}}{4\lambda (b-\lambda)}+1\big)<\tfrac{b^{2}(\sigma+1)}{2\sigma-b}.
 \end{equation}
\end{enumerate}
Then in the phase space of system~\eqref{sys:complex-Lorenz},
there are no nontrivial self-excited and hidden oscillations, and any of its solution
$(X(t),Y(t),Z(t))$ tends to the stationary set as $t\rightarrow\infty$.
\end{theorem}
\begin{proof} We consider the following cases:

\noindent {\bf Case 1}. For $|r+1|<2$, one can check that
conditions~\eqref{comp_lorenz:cond:unique_equil} are satisfied and
in this case $S_{0}$ is the only equilibria of system~\eqref{sys:complex-Lorenz}.
So, the absence of self-excited oscillations follows from
the Routh-Hurwitz criterion on local stability for the equilibrium $S_{0}$.
The absence of hidden oscillations can be obtained
by Barbashin-Krasovskii theorem (see e.g. \cite{BarbashinK-1952,KuznetsovLYYKKRA-2020-ECC})
and the Lyapunov function,
if $|r+1|<2$, from Theorem~\ref{dissipativity}
for $\vartheta := 0$ and $\gamma := \sigma$
we get the following Lyapunov function:
$$
  V(X,Y,Z) = \tfrac{1}{2}[|X|^{2}+\sigma(|Y|^{2}+Z^{2})]
$$
with  $\lambda_{0}=\min(1,b,\sigma)$, $\lambda\in(0,\lambda_{0} )$, $\delta>0$
and the inequalities $\sigma(\sigma-\lambda)(1-\lambda)-\frac{1}{4}\sigma^{2} \mid r+1 \mid^{2}\geq0$
and $|X(t)|^{2}<\frac{4\delta}{1+\sigma}$ hold.
The derivative of $V$ with respect to system~\eqref{sys:complex-Lorenz} is as follows~\cite{LeonovR-1987}:
\begin{equation*}
  \dot{V}(X,Y,Z)\leq-\lambda\big[|X|^{2}+\sigma|Y|^{2}\big]<0,\quad \forall X, Y,Z \neq0.
\end{equation*}
Also, $V(X,Y,Z)\geq0$, $V(0,0,0)=0$ and $V(X,Y,Z)\rightarrow\infty$ as $|X,Y,Z|\rightarrow \infty$;
thus, all conditions of Barbashin-Krasovskii theorem are satisfied
implying global stability of the unique equilibrium $S_{0}$.

\noindent {\bf Case 2}. For $4 r_1 < - r_2^2$,
conditions~\eqref{comp_lorenz:cond:unique_equil} are also fulfilled and
for this case we introduce the following Lyapunov function (see,~\cite{RauhHA-1996}):
\begin{equation}\label{Laypunovfanc}
  V(X,Y,Z) = \tfrac{1}{2}\big[ D^2 |X|^{2}+|Y|^{2}+Z^{2} \big],
\end{equation}
where $D = \sqrt{\frac{-r_{1}}{\sigma}}$.
It is clear that Lyapunov function~\eqref{Laypunovfanc} is positive definite.
Rewrite Lyapunov function \eqref{Laypunovfanc}
in terms of real variables:
\begin{equation}\label{Laypunovfanc1}
  V(x_{1},x_{2},x_{3},x_{4},x_{5}) =
  \tfrac{1}{2}\big[D^{2}(x_{1}^{2}\!+\!x_{2}^{2})\!+\!x_{3}^{2}\!+\!x_{4}^{2}\!+\!x_{5}^{2}\big].
\end{equation}

The derivative of $V$ with respect to system \eqref{13} reads:
\begin{multline}\label{deriv-Laypunovfanc1}
   \dot{V} = -\sigma D^{2}(x_{1}^{2} + x_{2}^{2}) - (x_{3}^{2}+x_{4}^{2})
   -  r_{2}(x_{2}x_{3}-x_{1}x_{4}) - b x_{5}^2.
\end{multline}
By using the following scaled variables
$(u_{1},u_{2},u_{3},u_{4},u_{5})$ as:
$x_{1}=\frac{u_{1}}{D\sqrt{\sigma}}$, $x_{2}=\frac{u_{2}}{D\sqrt{\sigma}}$,
$x_{3}=u_{3}$, $x_{4}=u_{4}$, $x_{5}=\frac{u_{5}}{\sqrt{b}}$, we have
\begin{equation}\label{deriv1-Laypunovfanc1}
  \dot{V} = -(u_{1}^{2} + u_{2}^{2} + u_{3}^{2} + u_{4}^{2} + u_{5}^{2})
  - 2 \Delta (u_{2}u_{3} - u_{1}u_{4}),
\end{equation}
where $\Delta=\frac{r_{2}}{2 D\sqrt{\sigma}}$.

The transformation $u \rightarrow v$, $u = A v$ with
\begin{equation*}
  A=\frac{1}{\sqrt{2}}\left(
    \begin{array}{ccccc}
      0 & -1 & 0 & 1 & 0 \\
      0 & 0 & 1 & 0 & -1 \\
      0 & 0 & 1 & 0 & 1 \\
      0 & 1 & 0 & 1 & 0 \\
      \sqrt{2}& 0 & 0 & 0 & 0 \\
    \end{array}
  \right),
\end{equation*}
leads to the following expression:
\begin{equation*}
 \dot{V} = - v_{1}^{2} - (1+\Delta)(v_{2}^{2} + v_{3}^{2}) - (1-\Delta)(v_{4}^{2}+v_{5}^{2}).
\end{equation*}
One can see that initial assumption $4 r_1 < - r_2^2$ implies $\Delta^{2} < 1$ and, therefore, $\dot{V}<0$,
and all conditions of Barbashin-Krasovskii
theorem are satisfied.
Thus, the equilibrium $S_{0}$ is globally stable.

As we discussed earlier,
if conditions~\eqref{comp_lorenz:cond:unique_equil} are not satisfied,
system~\eqref{sys:complex-Lorenz} has a stationary set containing a continuum of equilibria,
i.e the zero equilibrium $S_{0}$ and the equilibria $S_{\theta}$.
In this case, the Barbashin-Krasovskii theorem is not applicable.

\noindent {\bf Case 3}.
Suppose conditions \eqref{comp_lorenz:cond:circle_equil} are satisfied,
then the absence of nontrivial oscillations
(and, thus, the global stability of the stationary set $\{S_{0}, S_{\theta}\}$
can be demonstrated (see, \cite{LeonovR-1987}) by the LaSalle principal \cite{Lasalle-1960}.
For this purpose, consider the following time and coordinate transformations:
\begin{equation}\label{transformation-global}
  \left\{
    \begin{array}{ll}
       t \rightarrow \tau, \quad \psi:  (X, Y, Z) \rightarrow (\chi, \eta, \xi),  \\
     \tau=\frac{\sqrt{\sigma}}{\varepsilon}t,  \quad \chi=\frac{\varepsilon}{\sqrt{2\sigma}}X, \quad \eta=\frac{\varepsilon^{2}}{\sqrt{2}}(Y-X),\\
     \xi=\varepsilon^{2}(Z-\frac{|X|^{2}}{b}), \quad \varepsilon=\frac{1}{\sqrt{r-a}} > 0.
    \end{array}
  \right.
\end{equation}
System \eqref{sys:complex-Lorenz} transform into the following system:
\begin{equation}\label{Transf:sys:complex-Lorenz}
  \left\{
    \begin{array}{ll}
      \vspace{0.3cm}\dot{\chi}=\frac{d\chi}{d\tau}=\eta,  \\
      \vspace{0.3cm}\dot{\eta}=\frac{d\eta}{d\tau}=-\varrho \eta -\xi \chi-\varphi(\chi),  \\
     \vspace{0.3cm} \dot{\xi}=\frac{d\xi}{d\tau}=-\kappa \xi-\frac{\Lambda}{2}(\chi^{\star} \eta+ \chi\eta^{\star}),
    \end{array}
  \right.
\end{equation}
where $\varrho=\frac{\varepsilon(a+\sigma)}{\sqrt{\sigma}}$, $\varphi(\chi)=-\chi+\Theta\chi|\chi|^{2}$, $\Theta=\frac{2\sigma}{b}$, $\kappa=\frac{\varepsilon b}{\sqrt{\sigma}}$, $\Lambda=\frac{2}{b}(2\sigma-b)$. Here, the variables $\chi$ and $\eta$ are complex while $\xi$ is real. The new parameters $\Theta$, $\kappa$, and $\Lambda$ are real with $\Theta>0$, $\kappa>0$, and $\varrho$ is complex. Consider the following Lyapunov function:
\begin{equation}\label{Laypunovfanc2}
 V(\chi,\eta,\xi)=\tfrac{1}{2}\Big[\tfrac{1}{|\Lambda|}\xi^{2}+|\eta|^{2}-|\chi|^{2}+\tfrac{\Theta}{2}|\chi|^{4}\Big].
\end{equation}

Note that, the LaSalle principle requires the compactness of the set,
where the Lyapunov function $V$ is defined to show its boundedness from below.
In our case, one can show that the inequality $V(\chi,\eta,\xi)>-\frac{1}{4\Theta}$
is valid for any $(\chi,\eta,\xi)\in \mathbb{C}\times \mathbb{C}\times\mathbb{R}$.

From the relation $\dot{V}(\chi,\eta,\xi)=0,$ it follows that the largest
invariant set
$$
  M \subset \{(\chi,\eta,\xi)\in \mathbb{C}\times \mathbb{C}\times\mathbb{R} ~|~ \dot{V}(\chi,\eta,\xi)=0\}
$$
consists of the equilibrium points of system \eqref{Transf:sys:complex-Lorenz}.

The derivative of function $V$ with respect to system~\eqref{Transf:sys:complex-Lorenz} is as follows:
\begin{multline}\label{deriv-Lyapunovfanc2}
\dot{V}(\chi,\eta,\xi)= -\tfrac{\kappa}{|\Lambda|}\xi^{2}-\tfrac{1}{2}(sgn(\Lambda)+1)\chi\xi\eta^* \\
 -\tfrac{1}{2}(sgn(\Lambda)+1)\chi^*\xi\eta  - Re \varrho |\eta|^{2}.
\end{multline}

The last thing to check is that $\dot{V}(\chi,\eta,\xi) \leq 0$,
$\forall \chi,\!\eta \in \mathbb{C}$, $\xi \in \mathbb{R}$.
Consider two cases.

\noindent {\bf 3.1}. If $\Lambda = 2\sigma -b > 0$, then
following~\cite{LeonovR-1987} and using Theorem~\ref{dissipativity},
there exist $\epsilon>0$ such that for any solution
$(\chi(\tau),\eta(\tau),\xi(\tau))$ of \eqref{Transf:sys:complex-Lorenz}
with $t_{0}=t_0(\chi(0),\eta(0),\xi(0))$ the following inequality:
\begin{equation}\label{inequality-Transf}
  |\chi(t)|^{2}\leq \tfrac{\kappa}{\Lambda} Re \varrho-\epsilon,
\end{equation}
holds for all $t \geq t_{0}$.

From relations~\eqref{transformation-global}
we have $X(t)=\frac{\sqrt{2\sigma}}{\varepsilon}\chi(t)$,
and taking into account~\eqref{inequa2Lemma} we get:
\begin{equation}\label{etacondi}
 |\chi(t)|^{2}< \tfrac{\varepsilon^{2}\beta}{2\sigma}
 + \tfrac{2\delta\varepsilon^{2}}{\sigma(1+\gamma)} \qquad
 \forall \; t\geq t_{0}.
\end{equation}
It is easy to check that condition~\eqref{cond-glonal} with \eqref{etacondi} and
using the following constant in Theorem~\ref{dissipativity}:
\begin{equation*}
\vartheta = \tfrac{\pm 2 \sqrt{2\gamma\lambda(2 \sigma\!-\!b)(b-\lambda)[b^{2}\sigma(\gamma\!+\!1)(\sigma\!+\!1)\!+\!
\varepsilon(b\!-\!2\sigma)(\sigma(\gamma\!+\!1)\!+\!2\varepsilon\delta)]}}{\varepsilon b (2 \sigma-b)},
\end{equation*}
implies~\eqref{inequality-Transf}.

If $\Lambda > 0$, then expression~\eqref{deriv-Lyapunovfanc2} reads as follows:
\begin{equation}\label{deriv1-Lyapunovfanc2}
  \dot{V}(\chi,\eta,\xi)= -\tfrac{\kappa}{|\Lambda|}\xi^{2} - \chi\xi\eta^* - \chi^*\xi\eta
  - Re \varrho |\eta|^{2}.
\end{equation}
Since $\kappa>0$ and $Re \varrho=\frac{\varepsilon(1+\sigma)}{\sqrt{\sigma}}>0$
(see relation~\eqref{inequality-Transf}),
Eq.~\eqref{deriv1-Lyapunovfanc2} can be written as follows:
\begin{equation}\label{deriv2-Lyapunovfanc2}
  \dot{V}(\chi,\eta,\xi)\leq -\delta_{1}(\xi^{2}+|\eta|^{2})\leq0, \qquad \delta_{1} > 0.
\end{equation}

\noindent {\bf 3.2}.
If $\Lambda < 0$, then the mixed products in \eqref{deriv-Lyapunovfanc2} do not appear.
Hence, expression~\eqref{deriv-Lyapunovfanc2} can be written immediately in the form \eqref{deriv2-Lyapunovfanc2}.

Then, according to LaSalle principle any solution of system~\eqref{Transf:sys:complex-Lorenz}
(and thus system~(\ref{sys:complex-Lorenz}))
tends to an equilibrium state as $\tau\rightarrow\infty$.
\end{proof}

\begin{remark}
For some values of parameters of the complex Lorenz system~\eqref{sys:complex-Lorenz},
one can find similar conditions on the global stability as for the Lorenz system~\eqref{sys:lorenz}.
For instance, if $|r+1|<2$, $S_{0}$ of the system~\eqref{sys:complex-Lorenz} is globally stable,
and choosing $r_{2} = 0$ we obtain $r_{1}<1$ that corresponds to the condition
of global stability of zero equilibrium point of system~\eqref{sys:lorenz} (see Theorem~\ref{thm:gs}).
\end{remark}

\begin{remark}
In the special case, when $r_{2} = e = 0$, the remaining parameters
in the complex Lorenz system~\eqref{sys:complex-Lorenz} are the same
as the parameters of the real one \eqref{sys:lorenz}.
However, we can not strictly deduce the condition of global stability of
the stationary set $\{S_0,S_{\pm}\}$ of the Lorenz system~\eqref{sys:lorenz}
from the conditions of Theorem~\ref{theorem1}.
This due to the different behavior of these two systems in this case.
The complex Lorenz system~\eqref{sys:complex-Lorenz} still
has a continuum of equilibria $S_{\theta}$ in its phase space,
while the real one has only two symmetric equilibria $S_{\pm}$.
Remark that, for the complex Lorenz system~\eqref{sys:complex-Lorenz}
when $2\sigma \neq b$, we can only apply the LaSalle principle,
and yet it is not possible to apply the Leonov approach~\cite{Leonov-1991-Vest,KuznetsovAL-2016,Kuznetsov-2016-PLA}.
While for the real Lorenz system~\eqref{sys:lorenz},
if $2\sigma < b$, we can use the LaSalle principle,
and if $2\sigma>b$, we can apply the Leonov approach~\cite{KuznetsovMKKL-2020}.
It is also possible to prove global stability,
when in system~\eqref{sys:lorenz} we have
$2\sigma = b$~\cite{KuznetsovMKKL-2020,Leonov-2018-UMZh}.

Note also that, if in the third condition of Theorem~\ref{theorem1}
we take the parameters $\lambda, \gamma, \vartheta$ as follows:
$\gamma=\sigma$, $\vartheta=\sigma+\sigma r_{1}-2\sigma \sqrt{r_{1}}$, $\lambda\in (0, b)$
with $b\in(0, 1]$, then we get the following relations between
the conditions of global stability for systems
\eqref{sys:lorenz} and \eqref{sys:complex-Lorenz}:
\begin{equation*}
\begin{aligned}
  r < r_{gs} &= \tfrac{(\sigma+b)(1+b)}{\sigma} & \text{for system \eqref{sys:lorenz}}, \\
  r_1 < r_{1gs} &= \tfrac{(\sigma-\lambda)(1-\lambda)}{\sigma} & \text{for system \eqref{sys:complex-Lorenz}}.
\end{aligned}
\end{equation*}

\end{remark}

The case, corresponding to condition $3)$ of Theorem~\ref{theorem1},
when all trajectories tend to the stationary set, however,
not all equilibria of the stationary set are locally stable,
is illustrated in Fig.~\ref{fig:global_stability}.
\begin{figure}[!ht]
  \centering
  \includegraphics[width=\columnwidth]{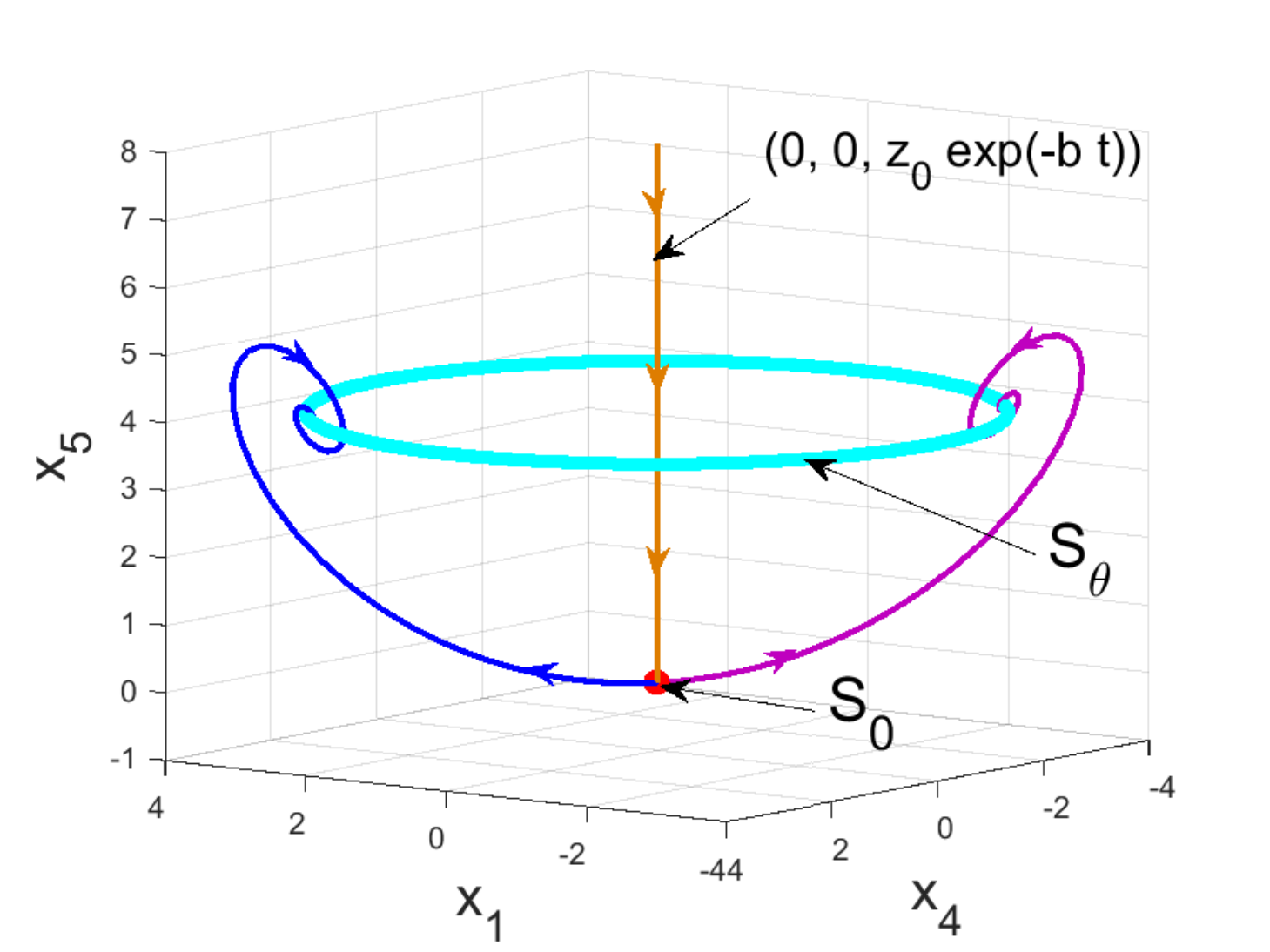}
  \caption{The absence of self-excited and hidden attractors and the global stability
  of the stationary set $\{S_{0}, S_{\theta}\}$ in the system~(\ref{13})
  with parameters \\ $\sigma=4$, $r_{1}=5$, $r_{2} = -e = 0.001$, $b=4$.
  Trajectories (blue, purple) in a small vicinity
  of the unstable equilibrium $S_{0}$ tend to
  to the stable set of equilibria~$S_{\theta}$ (trivial attractors).}
  \label{fig:global_stability}
\end{figure}

\begin{remark}
One of the significant differences between the complex Lorenz~\eqref{sys:complex-Lorenz}
and the real one is the following: in the case when $S_{0}$ is the only equilibria
of system~\eqref{sys:complex-Lorenz}, there could exist a nontrivial attractor,
that is a stable limit cycle or even a torus (see Fig.~\ref{fig:global_stability-torus}).
\end{remark}
\begin{figure}[!ht]
  \centering
  \includegraphics[width=\columnwidth]{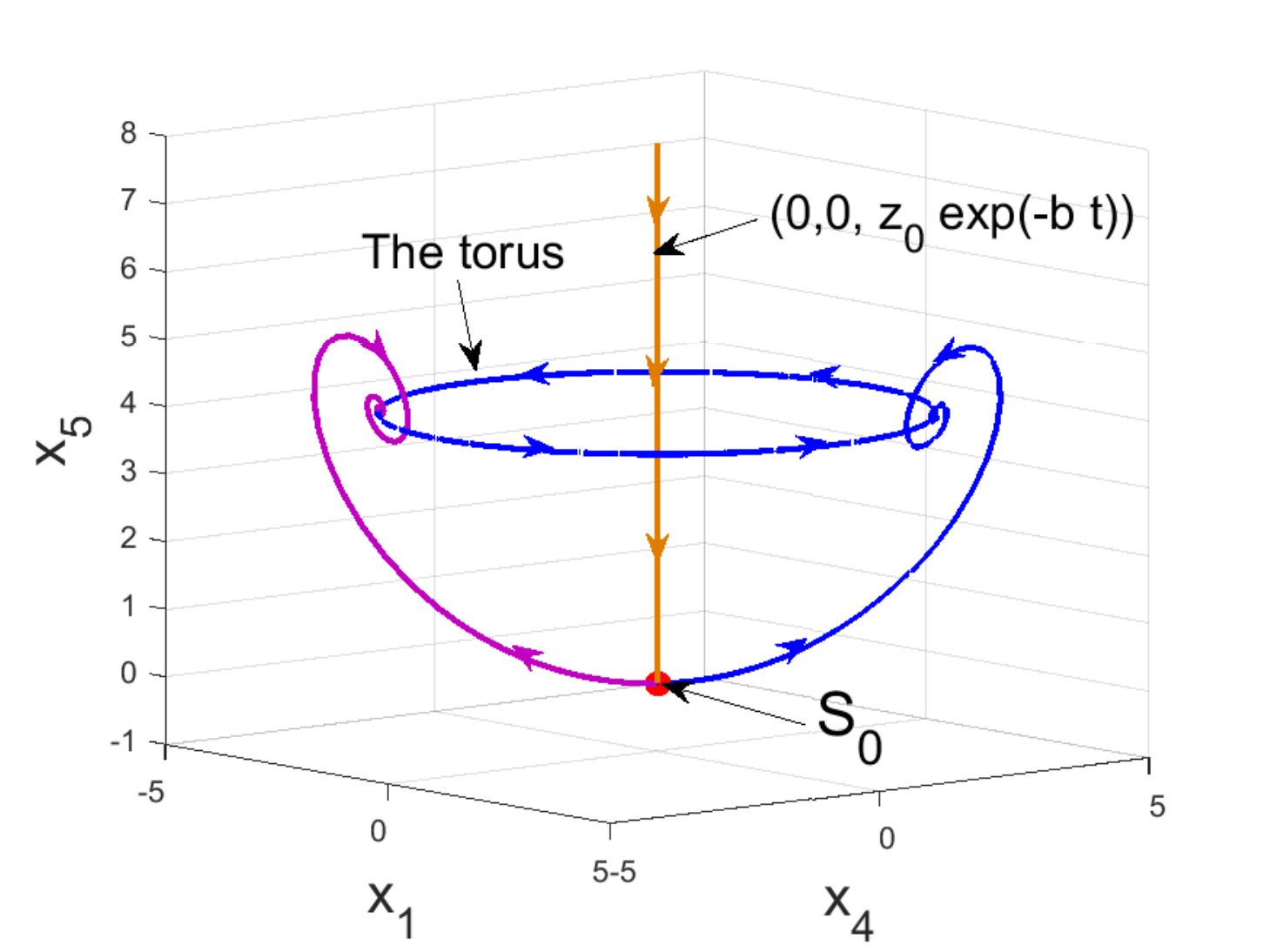}
  \caption{Co-existence of a nontrivial attractive torus
  and the unstable equilibria $S_{0}$ in the system~(\ref{13})
  with parameters $\sigma=4, r_{1}=5, r_{2}=0.002, e=-0.001, b=4$.
  Trajectories (blue, purple) in small neighborhoods of the unstable
  equilibrium $S_{0}$ attractors to the torus (nontrivial attractors).}
  \label{fig:global_stability-torus}
\end{figure}

Beyond the conditions of Theorem~\ref{theorem1},
the analysis of global stability and the birth of nontrivial attractors
in system~\eqref{sys:complex-Lorenz} can be carried out numerically.
For the convenience of this analysis, additional transformations
taking into account the structure
of a complex system may be useful~\cite{SiminosC-2011,FroehlichC-2012}.
For instance, in~\cite{VladimirovTD-1997,VladimirovTD-1998-IJBC},
Vladimirov et. al introduced the following
transformation for system~\eqref{sys:complex-Lorenz}.
If the conditions:
\begin{equation}\label{comp_lorenz:condition:transform:vladimirov}
  2\sigma>b, \quad
  \sigma(r_{1}-1)-\frac{e^{2}}{4}\equiv\eta > 0
\end{equation}
are satisfied, then using the time and
coordinate transformations:
\begin{equation}\label{transformation}
  \left\{
    \begin{array}{ll}
      t \rightarrow t', \quad \psi:  (X, Y, Z) \rightarrow (X', Y', Z'),  \\
     t'\!=\!\sqrt{\eta} t,  \; X'\!=\!\eta^{\frac{-3}{4}}\delta X, \;
     Y'\!=\!\eta^{\frac{-5}{4}}\sigma\delta[Y\!-\!(1\!+\!i\delta/2\sigma)X], \\
     Z'\!=\!\eta^{-1}\sigma(Z \!-\! XX^*/2), \quad
     \delta=\exp(iet/2)[(2\sigma\!-\!b)/2]^{\frac{1}{2}}.
    \end{array}
  \right.
\end{equation}
one can re-write system~\eqref{sys:complex-Lorenz} in the form:
\begin{equation}\label{sys:projective:homoc}
  \left\{
    \begin{array}{ll}
      \frac{dX'}{dt'}=Y', \\
      \frac{dY'}{dt'}=(1+i\kappa)X'-\mu Y'-X'Z'-\varrho X'|X'|^{2}, \\
      \frac{dZ'}{dt'}=-\beta Z'+|X'|^{2}
    \end{array}
  \right.
\end{equation}
where $\kappa=\frac{(2\sigma r_{2}+e(\sigma-1))}{2\eta}$,
$\mu=\frac{1+\sigma}{\sqrt{\eta}}$, $\varrho=\frac{\sqrt{\eta}}{2\sigma-b}$,
$\beta=\frac{b}{\sqrt{\eta}}$.
Since the transformation~\eqref{transformation} is continuous
and invertible then its a diffeomorphism, and the dynamical behavior
of systems~\eqref{sys:complex-Lorenz} and~\eqref{sys:projective:homoc}
are topologically equivalent~\cite{VladimirovTD-1997,VladimirovTD-1998-IJBC}.
Next, introduce the following real variables $\xi', \upsilon', w', Z'$:
\begin{equation}\label{projec1-coordinates}
  \xi' = (|X'|^{2}-|Y'|^{2})/2, \quad
  \upsilon' + iw' = X'^* Y'.
\end{equation}
Expression~\eqref{projec1-coordinates}
defines the projection map:
$\Pi : \mathcal{H} \rightarrow \mathcal{P}$.
This map projects all states $(X',Y',Z')$ in the phase space
$\mathcal{H}=\mathbb{C}\times \mathbb{C}\times\mathbb{R}$
of system~\eqref{sys:projective:homoc} (and system~\eqref{sys:complex-Lorenz}),
which are invariant with respect to
a common phase factor $\psi$ in $X'$ and $Y'$, i.e.
\begin{equation}\label{invariance:transf}
  \left(
    \begin{array}{c}
      X' \\
      Y' \\
      Z'
    \end{array}
  \right) \to
  \left(
    \begin{array}{c}
      X' \exp(i \psi) \\
      Y' \exp(i \psi) \\
      Z'
    \end{array}
  \right),
\end{equation}
into the same point ($\xi'$, $\upsilon'$, $w'$, $Z'$)
in $\mathcal{P} =\mathbb{R}^{4}$.
Indeed, if $X'' = X' \exp(i\psi)$ and $Y'' = Y'\exp(i\psi)$,
then
\begin{align*}
 \xi'' &= (|X''|^{2}-|Y''|^{2})/2=(|\exp(i\psi)|^{2}[|X'|^{2}-|Y'|^{2}])/2\\
       &=((\cos^{2}(\psi)+\sin^{2}(\psi))[|X'|^{2}-|Y'|^{2}])/2\\
       &=(|X'|^{2}-|Y'|^{2})/2 = \xi', \\
 \upsilon'' &+ iw'' = X''^* Y'' = X'^{\star}\exp(-i\psi)\, Y'\exp(i\psi) \\
  & \qquad\quad = X'^{\star}Y' = \upsilon' + iw'.
\end{align*}
As mentioned in~\cite{VladimirovTD-1997,VladimirovTD-1998-IJBC},
the set of points in $\mathcal{H}$ corresponding to the same point in $\mathcal{P}$
in geometry is called a ''ray'', and the space $\mathcal{P}$ is called a ''ray space''.

The derivative of expressions~\eqref{projec1-coordinates} with respect to
the dynamics of system \eqref{sys:projective:homoc} leads
to the following equations of motion in
the projective space $(\xi', \upsilon', w', Z') \in \mathcal{P}$:
\begin{equation}\label{projec1-system}
  \left\{
    \begin{array}{ll}
      \dot{\xi'} =\upsilon'+\mu(R-\xi')-\kappa w'-\upsilon'(1-Z'-\varrho(R+\xi')),\\
      \dot{\upsilon'}=-\mu\upsilon'+R-\xi'+(R+\xi')(1-Z'-\varrho(R+\xi')),  \\
      \dot{w'}=-\mu w'+\kappa(R+\xi'), \\
     \dot{ Z'}=-\beta Z'+R+\xi',
    \end{array}
  \right.
\end{equation}
where $R$ satisfies the following relations:
\begin{align*}
& R = (\xi'^{2} \!+\! \upsilon'^{2} \!+\! w'^{2})^{\frac{1}{2}} =
(|X'|^{2} \!+\! |Y'|^{2})/2, & \\
& R + \xi=|X'|^{2}, \qquad R-\xi=|Y'|^{2}. &
\end{align*}
\begin{figure}[!ht]
  \centering
  \includegraphics[width=\columnwidth]{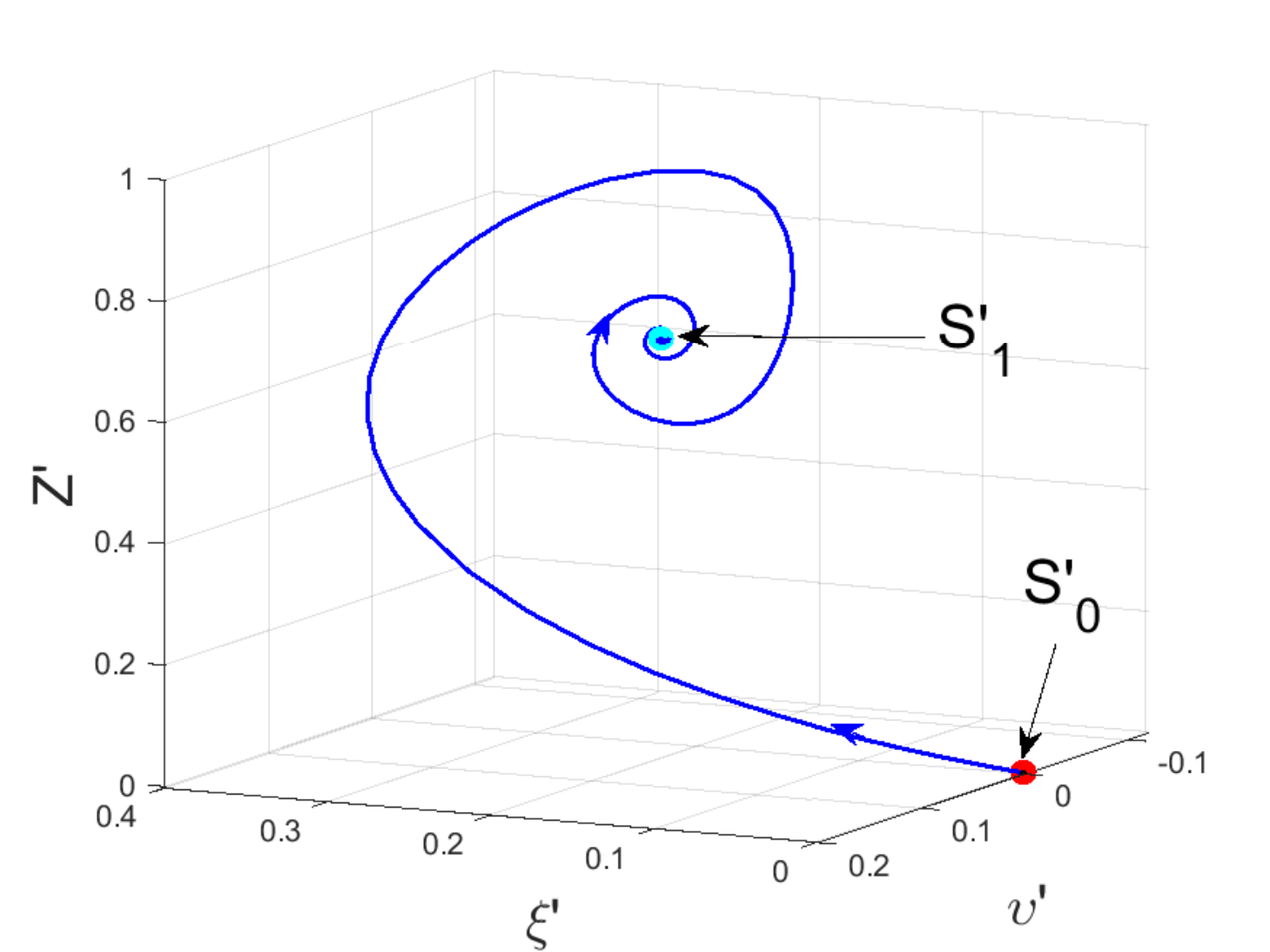}
  \caption{Global stability of the stationary set $\{S'_{0}, S'_{1}\}$
  in the projective system \eqref{projec1-system} with parameters
  $\sigma=4$, $r_{1}=5$, $r_{2} = -e = 0.001$, $b=4$.}
  \label{fig:global_stability_projec}
\end{figure}

Using projective system~\eqref{projec1-system} it is convenient to
study and visualize the dynamics of the initial system~\eqref{sys:complex-Lorenz}.
For instance, it possible to demonstrate the presence of global stability,
when conditions of Theorem~\ref{theorem1} and
condition~\eqref{comp_lorenz:condition:transform:vladimirov} hold.
Taking the same values of parameters
$\sigma=4$, $r_{1} = 5$, $r_{2}=0.001$, $e=-0.001$, $b=5$,
considered for system~\eqref{sys:complex-Lorenz} (see Fig.~\ref{fig:global_stability}),
one gets the same qualitative behavior
in the phase space of system~\eqref{projec1-system}
(see Fig.~\ref{fig:global_stability_projec}).

As in the classical Lorenz system~\eqref{sys:lorenz},
for the complex Lorenz system~\eqref{sys:complex-Lorenz}
it is also known that the separatrix of saddle $S_{0}$
can form a homoclinic loop, from which unstable cycles can arise
and violate global stability (however, a set of measure zero does
not affect the global attraction on a stationary set from
a practical point of view).
The following theorem provide the necessary condition
for the existence of homoclinic orbits
in the complex Lorenz system~\eqref{sys:complex-Lorenz}
(see \cite{VladimirovTD-1997,VladimirovTD-1998-IJBC,VladimirovTD-1998-TP}):
\begin{theorem}
If conditions~\eqref{comp_lorenz:condition:transform:vladimirov} are satisfied,
then the necessary condition for the presence of a homoclinic orbit
of the saddle equilibrium $S_{0}$ is as follows:
\begin{equation}\label{condi-homoclinic}
  r_{2} = \tfrac{e(1-\sigma)}{2\sigma}.
\end{equation}
\end{theorem}
\begin{corollary}\label{corollary:1}
For the "laser case", i.e., when $e=-r_{2}$,
the homoclinic orbit can only be obtained if $e=r_{2}=0$.
\end{corollary}
To prove condition~\eqref{condi-homoclinic},
Vladimirov et al.
considered the projective system~\eqref{projec1-system},
because unlike system~\eqref{sys:complex-Lorenz},
it has a 1-dimensional unstable manifold.
The projections of unstable manifold $W^{u}$ and stable manifold $W^{s}$
in the space $\mathcal{P}$, i.e. $\Pi(W^{u})$ and $\Pi(W^{s})$,
can intersect only along the $Z'$-axis and a possible homoclinic orbit exists only
if $\kappa = 0$, which coincides with condition~\eqref{condi-homoclinic}.

For the parameters $\sigma=10$, $r_{2}=4.5\times 10^{-4}$, $e=-0.001$, $b=\frac{8}{3}$
of system~\eqref{sys:complex-Lorenz}, condition \eqref{condi-homoclinic} is satisfied.
In this case, system~\eqref{sys:complex-Lorenz} has a unique equilibrium $S_{0}$
(since $e\neq-r_{2}$), and it is possible to
find numerically the approximate value of such homoclinic bifurcation
$r_{1h}\approx 13.9,$ when two symmetric
homoclinic orbits appear, forming a homoclinic butterfly
(see Fig.~\ref{fig:homoclinicorbit11}).
A further increase in the parameter $r_{1}$ leads
to the birth of two periodic saddle orbits from each homoclinic orbit.
For the parameters $\sigma=10, r_{2}=0, e=0, b=\frac{8}{3}$,
condition \eqref{condi-homoclinic} holds and system~\eqref{sys:complex-Lorenz}
has the following equilibrium points: $S_{0}, S_{\theta}$ (since $e=-r_{2}$),
with $r_{1h}\approx 13.9,$ system~\eqref{sys:complex-Lorenz}
has homoclinic orbits (see Fig.~\ref{fig:homoclinicorbit22}.)
\begin{figure*}[!ht]
 \centering
 \subfloat[
 {\scriptsize  Two symmetric homoclinic orbits (homoclinic butterfly)
in the complex Lorenz system (\ref{13})}
 ] {
 \label{fig:homoclinicorbit1}
 \includegraphics[width=0.33\textwidth]{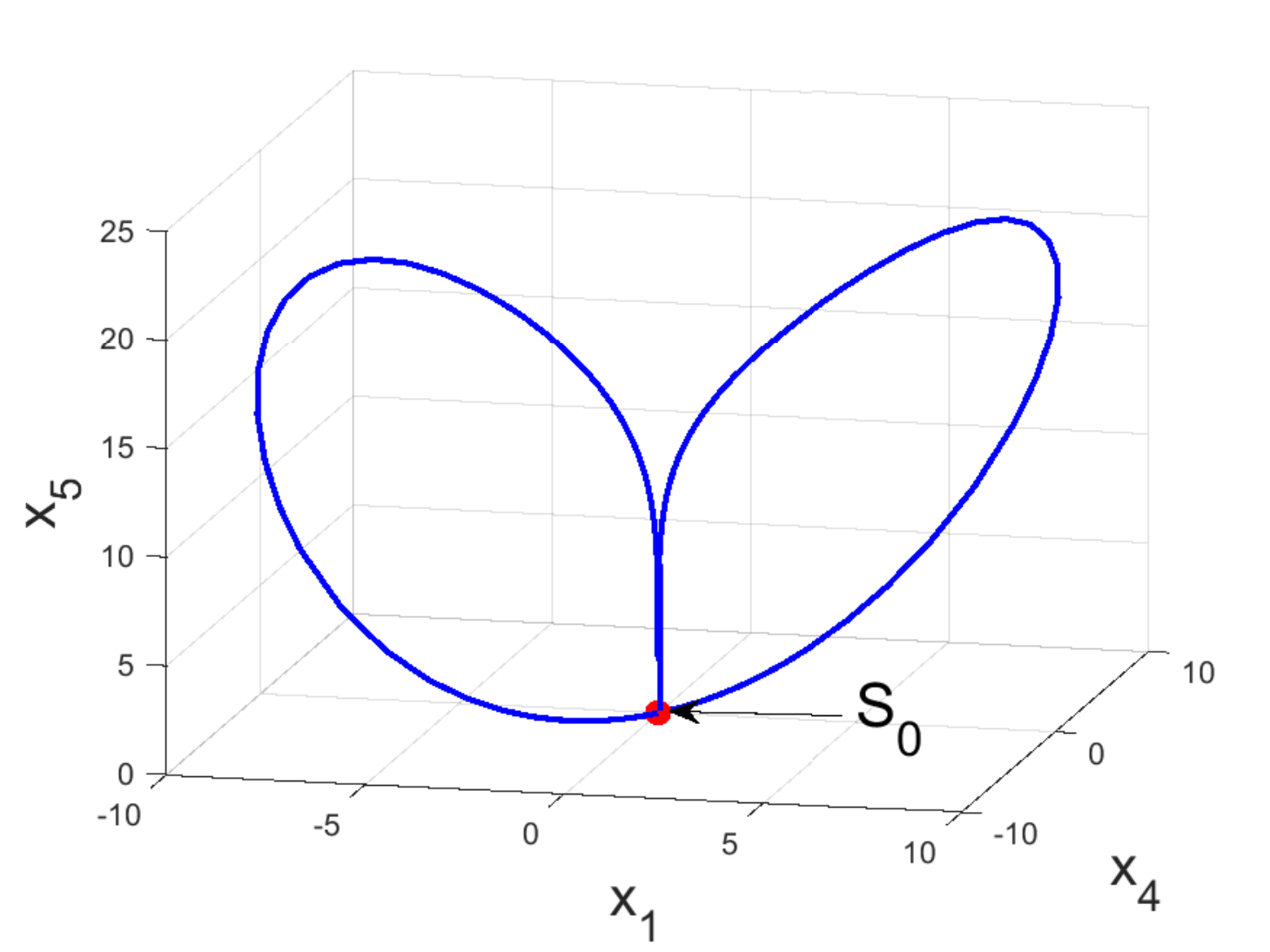}
 }~
 \subfloat[
 {\scriptsize  Some homoclinic orbits in the phase space of the original system \eqref{13}  }
 ] {
 \label{fig:homoclinicorbit3}
 \includegraphics[width=0.33\textwidth]{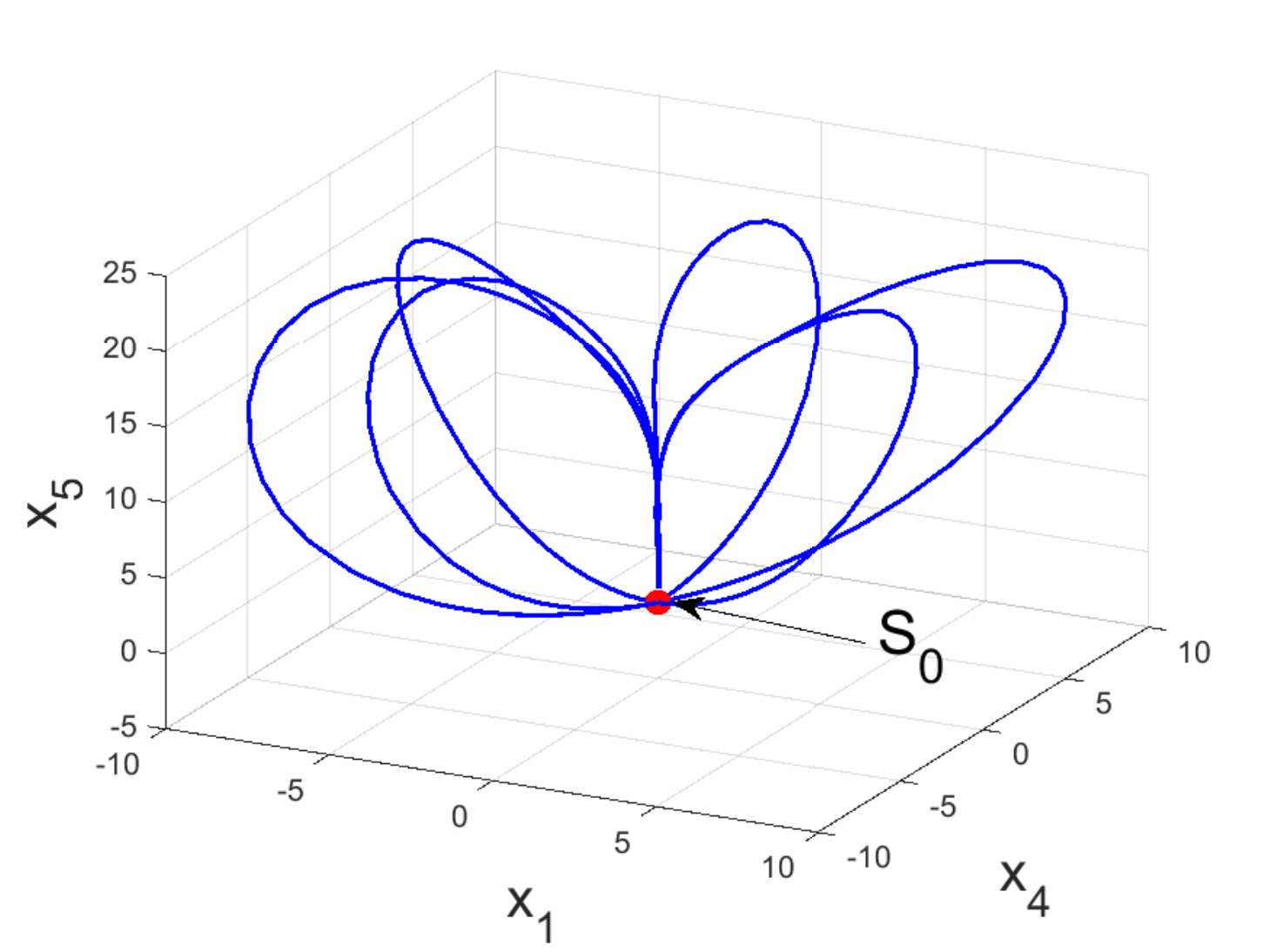}
 }
 \subfloat[
 {\scriptsize  Projection of a homoclinic orbit in the projective system \eqref{projec1-system}   }
 ] {
 \label{fig:homoclinicorbit1-projec}
 \includegraphics[width=0.33\textwidth]{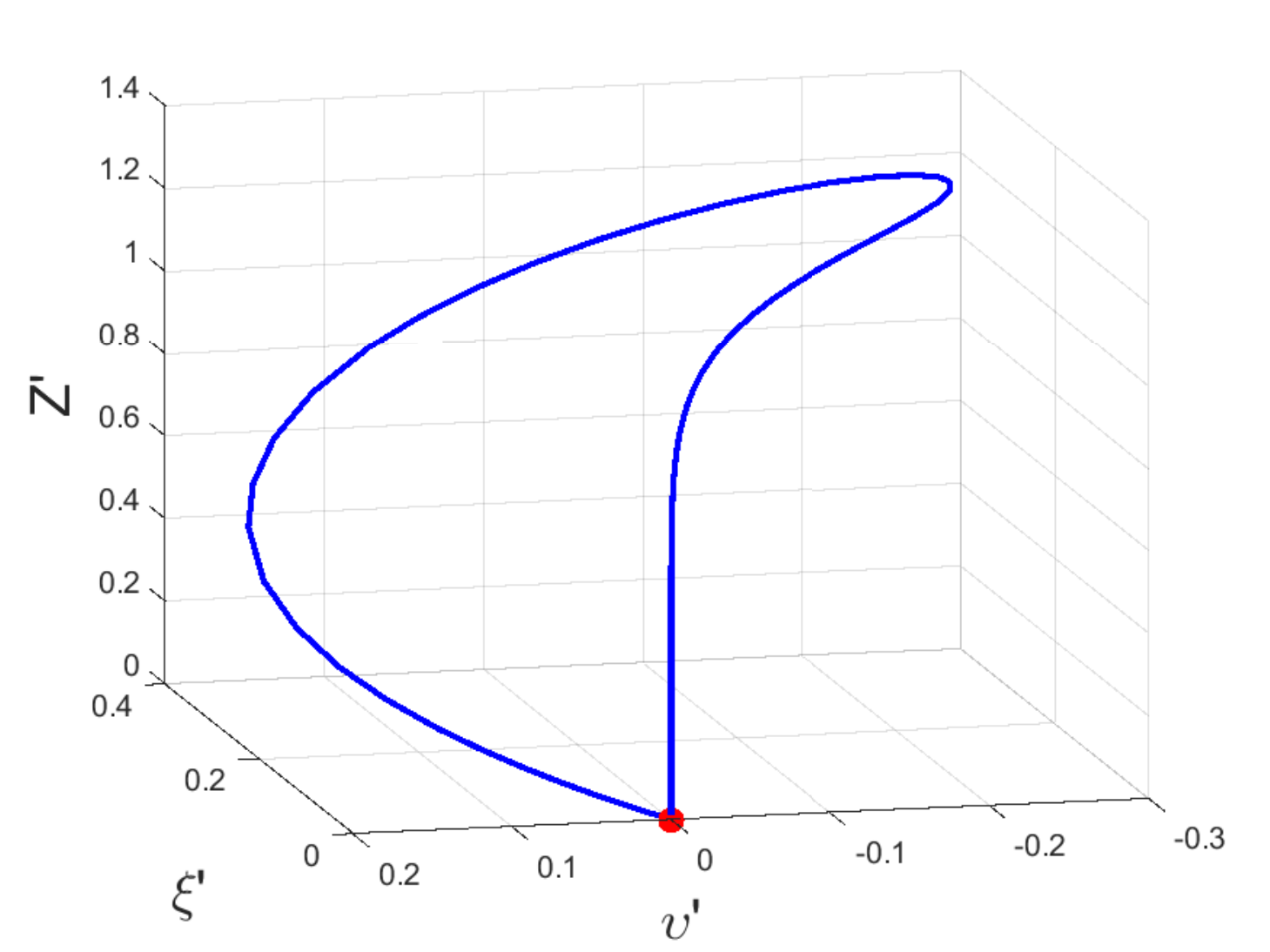}
 }
 \caption{ Visualization of homoclinic orbits with $\sigma=10, r_{2}= 4.5\times 10^{-4}, e=-0.001, b=\frac{8}{3}$ and $r_{1h}\approx 13.9$.
 }
 \label{fig:homoclinicorbit11}
\end{figure*}
\begin{figure*}[!ht]
 \centering
  \subfloat[
 {\scriptsize  Two symmetric homoclinic orbits in the complex Lorenz system \eqref{13}}
 ] {
 \label{fig:homoclinicorbit2}
 \includegraphics[width=0.33\textwidth]{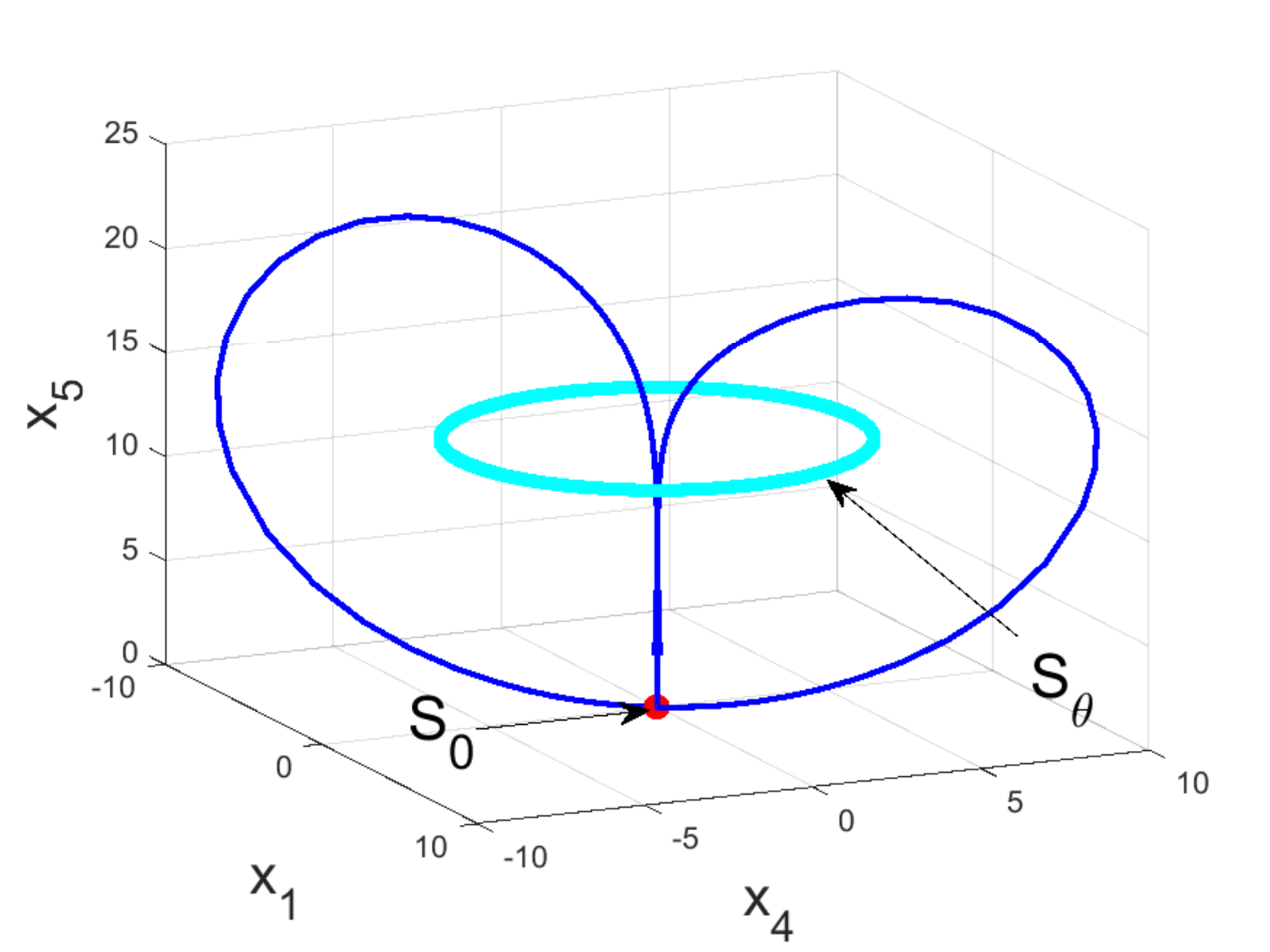}
 }~
 \subfloat[
 {\scriptsize  Some homoclinic orbits in the phase space of the original system \eqref{13}}
 ] {
 \label{fig:homoclinicorbit31}
 \includegraphics[width=0.33\textwidth]{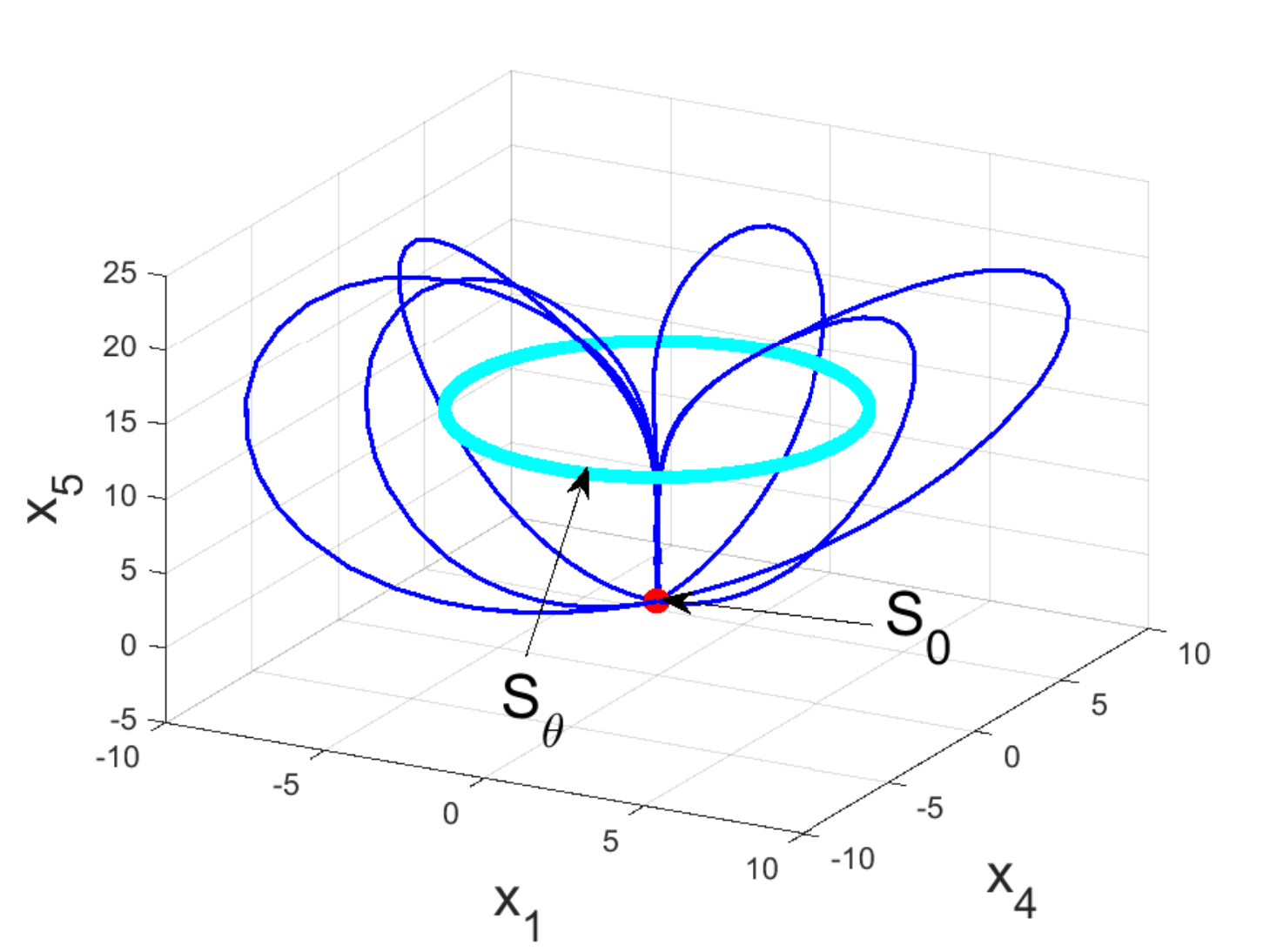}
 }
 \subfloat[
 {\scriptsize  Projection of a homoclinic orbit in the projective space \eqref{projec1-system}}
 ] {
 \label{fig:homoclinicorbit2-projec}
 \includegraphics[width=0.33\textwidth]{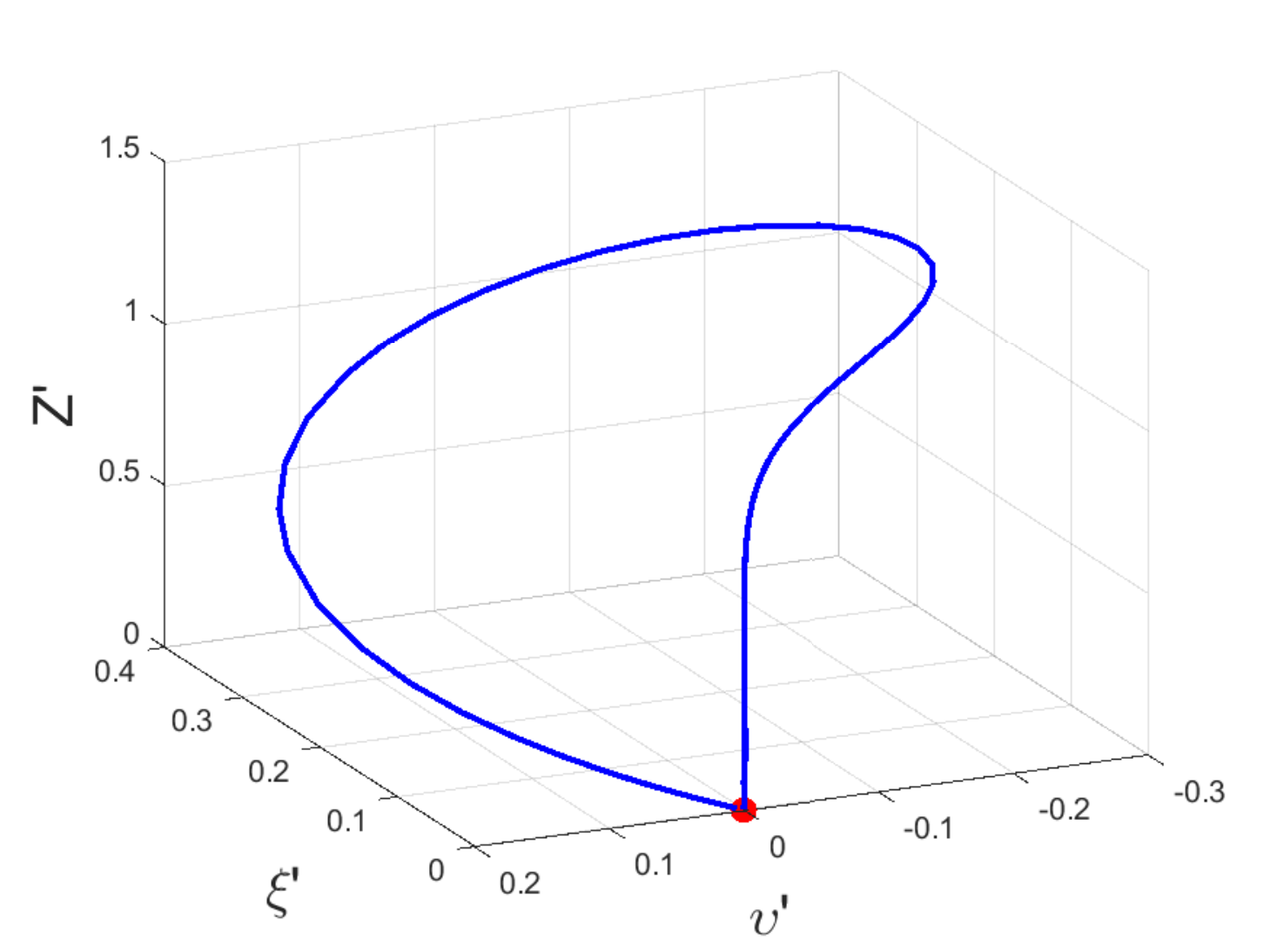}
 }
 \caption{ Visualization of homoclinic orbits with $\sigma=10, r_{2}=0, e=0, b=\frac{8}{3}$ and $r_{1h}\approx 13.9$.
 }
 \label{fig:homoclinicorbit22}
\end{figure*}

Finally, let us mention the action of the projection map $\Pi$
on various attractors in the space $\mathcal{H}$.
In other words, preimages of different kinds of attractors
in the original space $\mathcal{H}$ and the corresponding images
in the projection space $\mathcal{P}$.
Basing on some geometrical properties of the map $\Pi$,
in \cite{VladimirovTD-1997,VladimirovTD-1998-IJBC}
it is discussed that every limit set
in $\mathcal{H}/\mathcal{Z}$ can be represented locally
(in a neighborhood of the given ray)
by the direct product of a set in $\mathcal{P}/\mathcal{Z}$
and the ray, i.e., the set $\mathbb{R}^1$.
Here we exlude the sets of points $\mathcal{Z}$ on the $Z'$-axes
in phase spaces $\mathcal{H}$ and $\mathcal{P}$, which
are invariant with respect to systems~\eqref{sys:projective:homoc}
and~\eqref{projec1-system}, respectively.
Moreover, the triplet $(\mathcal{H}/\mathcal{Z},P/Z,\Pi)$ forms
a fiber bundle (see e.g.~\cite{KobayashiN-1963}), $\Pi$ is continuous map and, thus, it maps
connected/compact sets into connected/compact sets.

To illustrate this statement,
let us list the following different types of attractors:
\renewcommand{\labelenumi}{\bf(\arabic{enumi})}
\begin{enumerate}[leftmargin=1pt,labelwidth=-15pt,labelsep=1pt]
  \item For the zero equilibrium $S_{0}$ of system \eqref{sys:projective:homoc}:
  $\Pi(S_{0})=S'_{0}$, where $S'_{0}=(0,0,0,0)$
  is the zero equilibrium point of projective system~\eqref{projec1-system};
  \item For the circle of equilibria $S_{\theta}$ of system \eqref{sys:projective:homoc}:
  $\Pi(S_{\theta})=S'_{1}$, where $S'_{1}$
  is an equilibrium point of system~\eqref{projec1-system};
  \item For limit cycle of system \eqref{sys:projective:homoc} the image
  is an equilibrium $S'_{1}$ of system~\eqref{projec1-system} in $\mathcal{P}$;
  \item For a torus of system \eqref{sys:projective:homoc}
  the image is a limit cycle of system~\eqref{projec1-system} in $\mathcal{P}$;
  \item For a single homoclinic orbit of projective system~\eqref{projec1-system}
  one has (as a preimage of map $\Pi$) a continuum of homoclinic orbits
  in the original space $\mathcal{H}$,
  differing only by the common phase according
  to Eq.~\eqref{invariance:transf};
  \item For a chaotic attractor of projective system~\eqref{projec1-system} in $\mathcal{P}$
  one has (as a preimage) a chaotic attractor in~$\mathcal{H}$;
  \item For a transient chaotic set of projective system~\eqref{projec1-system} in $\mathcal{P}$,
  which collapses by colliding with stable equilibria or stable limit cycle,
  there exists (as a preimage) a transient chaotic set
  which after a certain time approaches eventually circle of equilibria $S_{\theta}$ or a torus, respectively.
\end{enumerate}

Summarizing the written above,
further in our experiments we will examine
the appearance of attractors in
projective system~\eqref{projec1-system}, which
has $1$-dimensional unstable manifold
of zero saddle equilibrium $S'_{0}$
and use map~\eqref{projec1-coordinates}
to get the corresponding preimages of the discovered
attractors in the original system~\eqref{sys:complex-Lorenz}.

\section{Outer estimation: the absence of trivial attractors}

System \eqref{sys:complex-Lorenz} possesses the absorbing set
(defined by Eq.~\eqref{inequa1Lemma}) and for $\sigma>b+1$,
$r_{1} > \max \{r_{1c}, r'_{1c}\}$ (see Lemma~\ref{Lemma1} and~\ref{Lemma2})
all equilibria are unstable.
Further, in this article for the system~\eqref{sys:lorenz} (and \eqref{13})
we will fix the classical values of parameters $\sigma=10$, $b=\frac{8}{3}$
of real Lorenz model~\eqref{sys:lorenz}
and, following the "laser case",
will specify two additional parameters: $r_{2} = -e = 0.001$.
For instance, when $r_{1} = 28 > r'_{1c}$,
it is possible to observe a self-excited chaotic attractor
with respect to all equilibria
of a stationary set (see Fig.~\ref{fig:lorenz:selfexcited}).
This gives an outer estimation of practical global stability.
\begin{figure*}[!ht]
 \centering
 \subfloat[
 {\scriptsize Initial data near the equilibrium $S_{0}$}
 ] {
 \label{fig:lorenz:attr:hidTrans:hid}
 \includegraphics[width=0.4\textwidth]{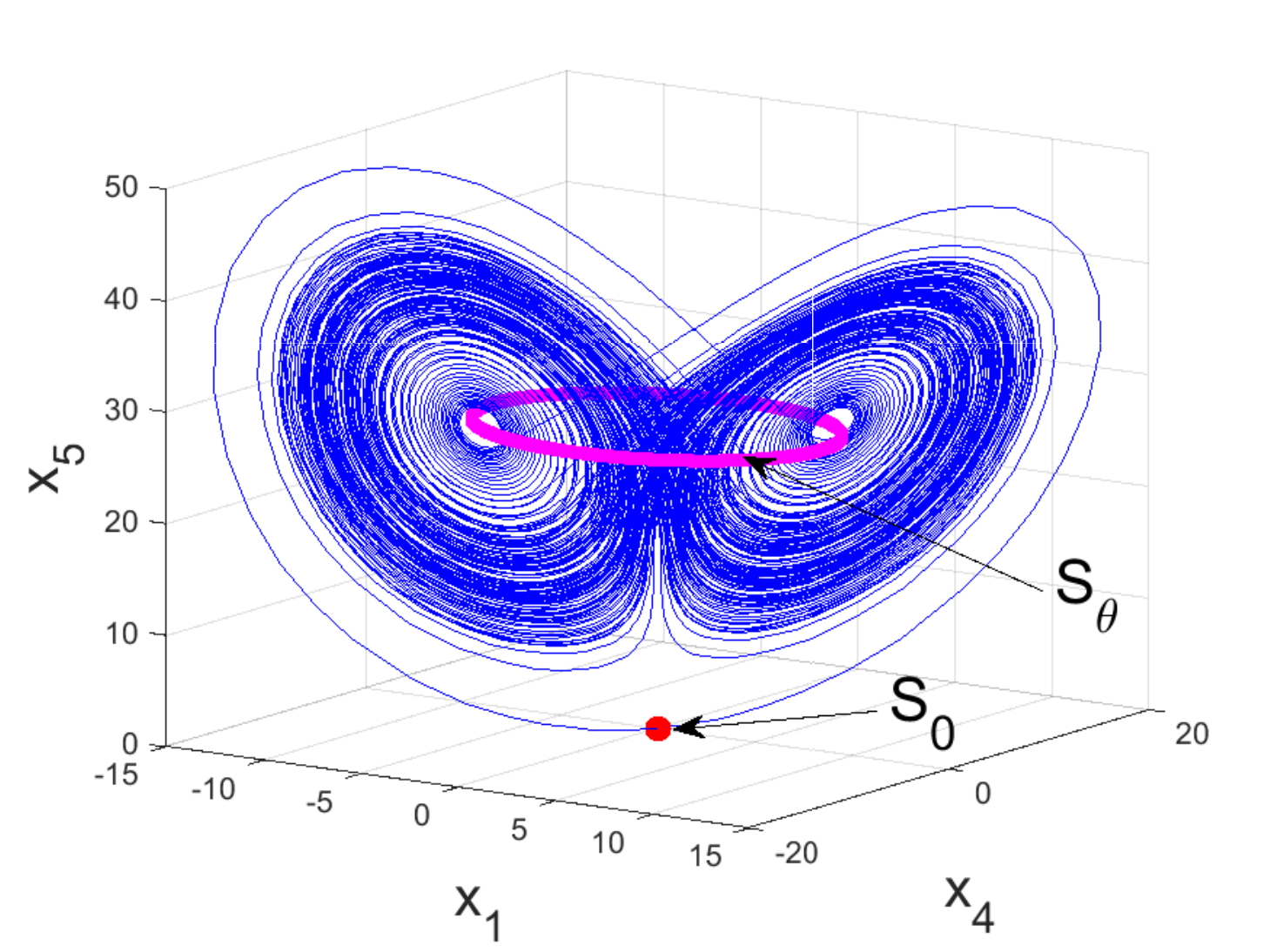}
 }~
 \subfloat[
 {\scriptsize Initial data near an equilibrium $S_{\theta }$}
 ] {
 \label{fig:lorenz:attr:hidTrans:trans}
 \includegraphics[width=0.4\textwidth]{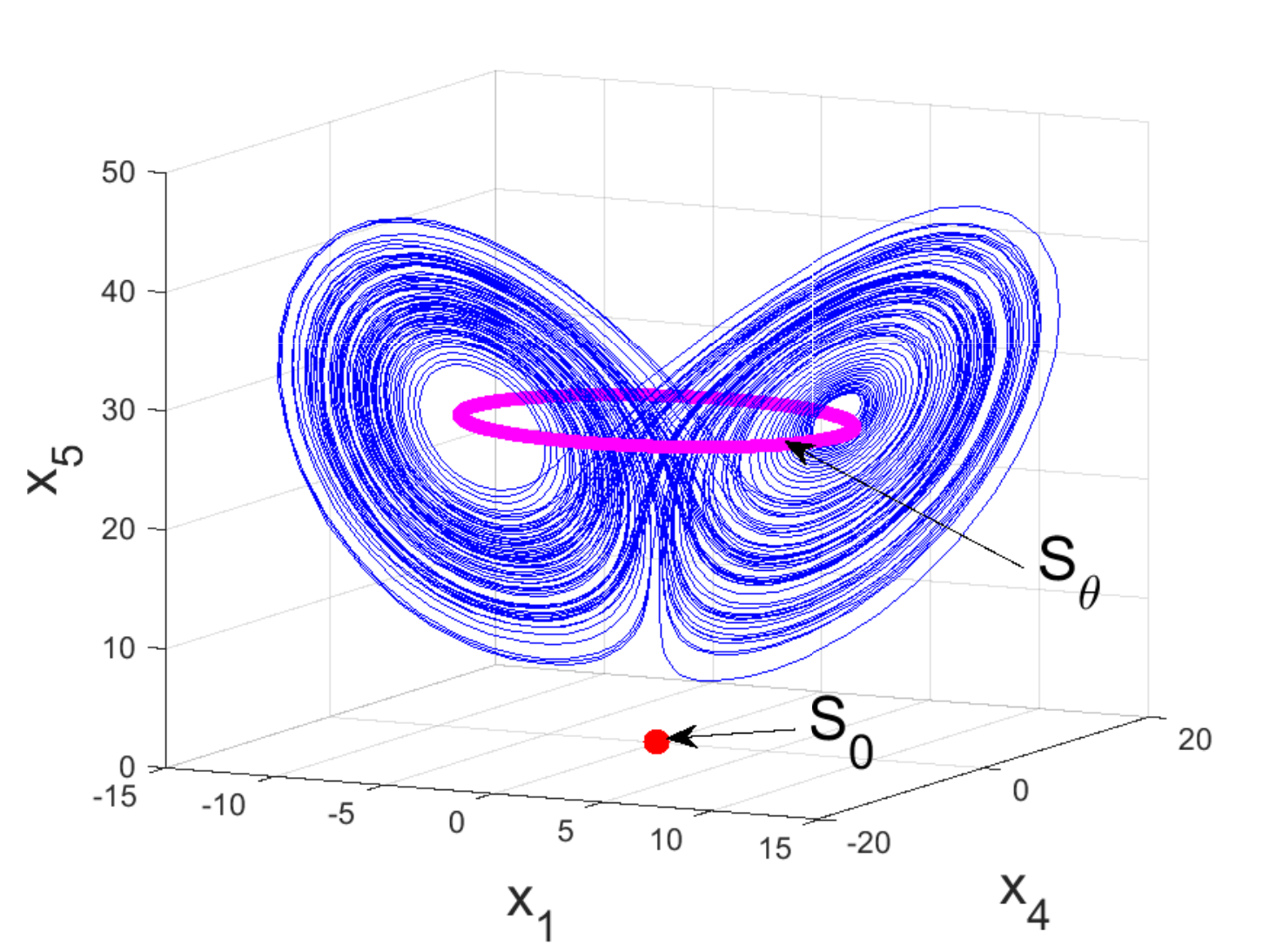}
 }
 \\
 \centering
 \subfloat[
 {\scriptsize Initial data near the equilibrium $S'_{0}$}
 ] {
 \label{fig:lorenz:proj:attr:hidTrans:hid}
 \includegraphics[width=0.4\textwidth]{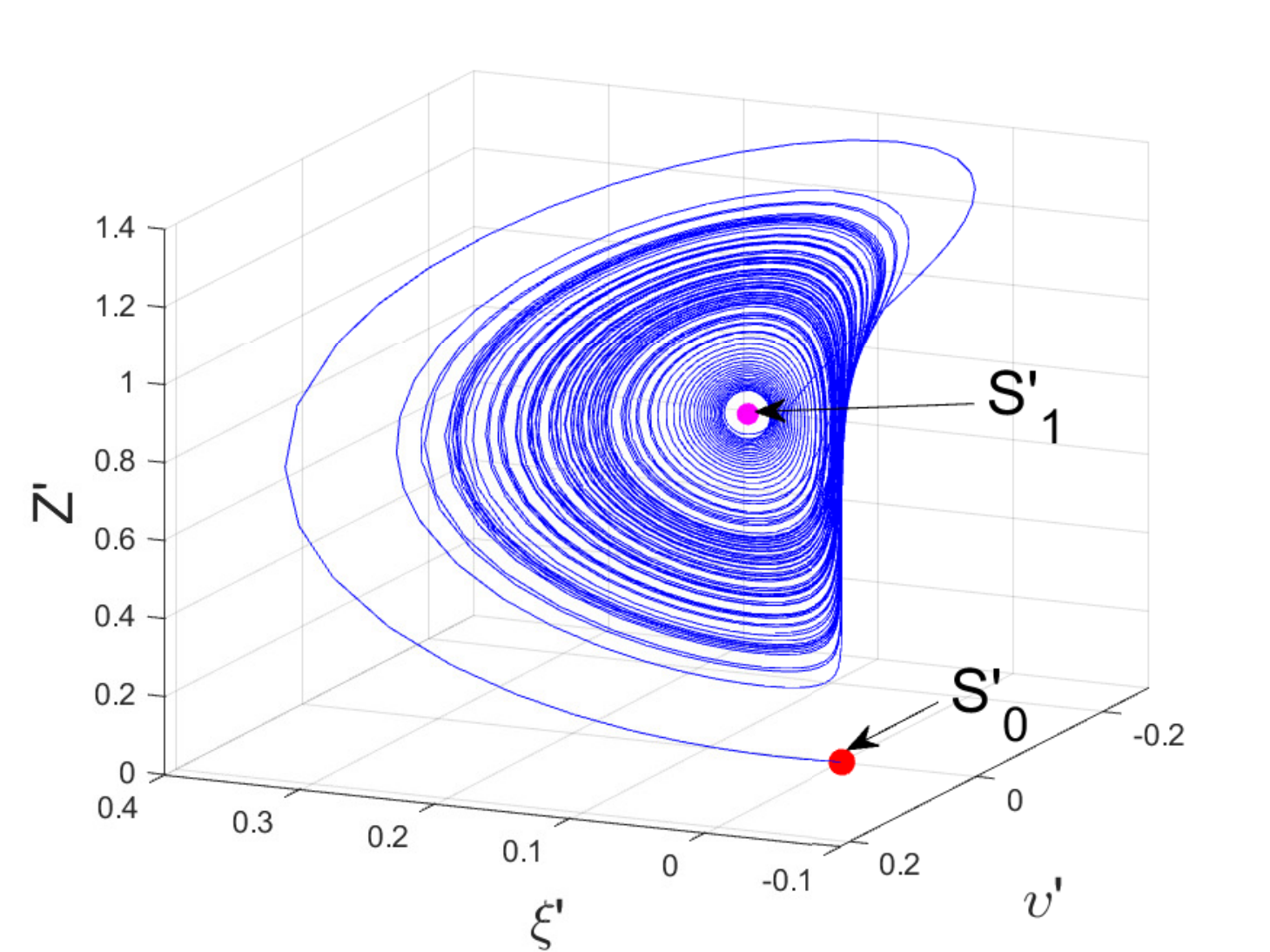}
 }~
 \subfloat[
 {\scriptsize Initial data near the equilibrium $S'_{1}$}
 ] {
 \label{fig:lorenz:proj:attr:hidTrans:trans}
 \includegraphics[width=0.4\textwidth]{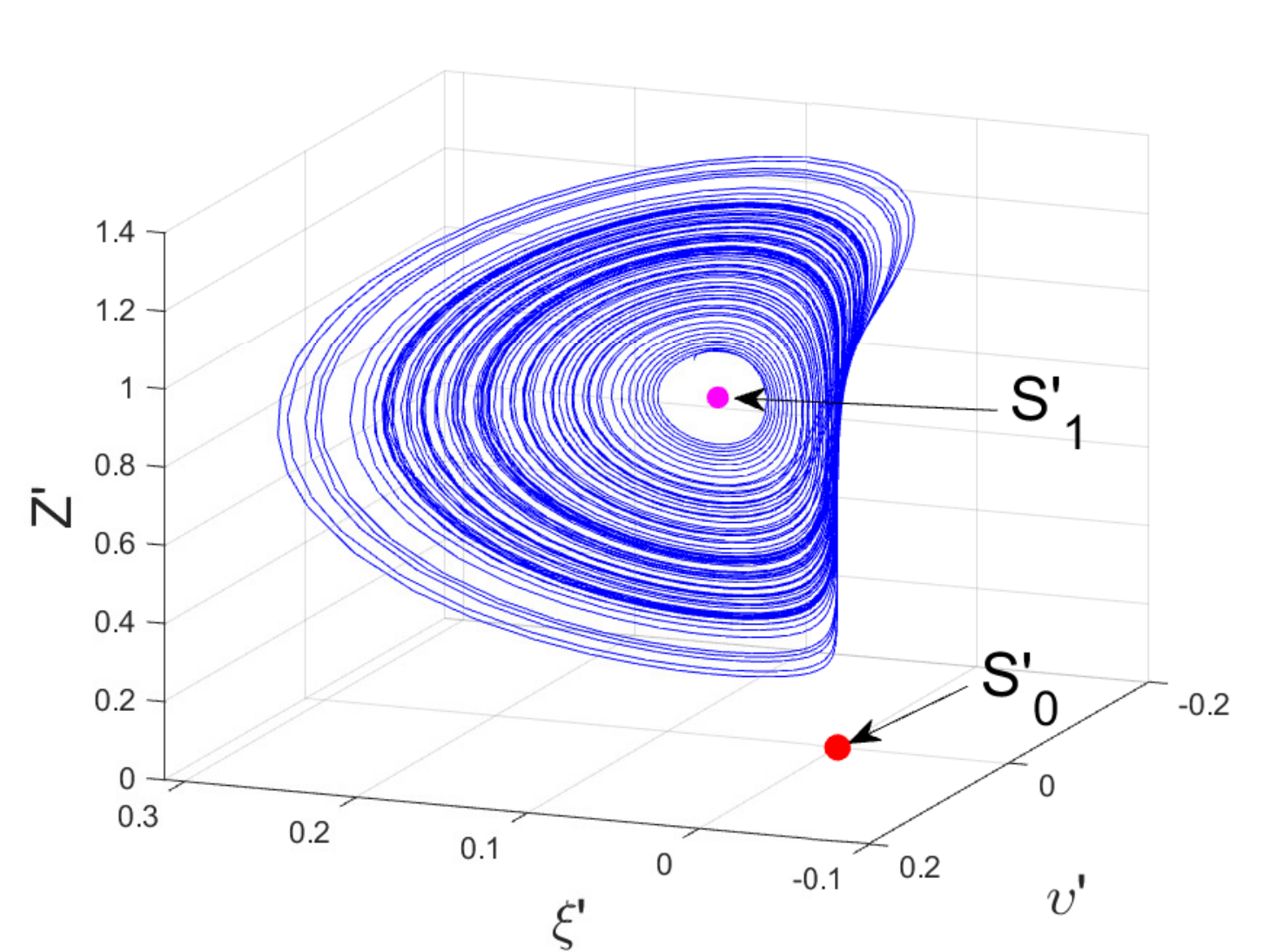}
 }
 \caption{
 (a), (b) Numerical visualization of the self-excited chaotic attractor in system~\eqref{13} with
 $r_{1} = 28$, $\sigma = 10$, $b = \tfrac{8}{3}$, $r_{2} = -e = 0.001$
 by integrating the trajectories with initial data from small
 vicinities of the unstable equilibria $S_{0}$, $S_{\theta}$;
 (c), (d) The corresponding images in the projective space $\mathcal{P}$.
 }
 \label{fig:lorenz:selfexcited}
\end{figure*}
\begin{figure*}[!ht]
 \centering
 \subfloat[
 {\scriptsize Initial data near the equilibrium $S_{0}$}
 ] {
 \label{fig:selfexcited11}
 \includegraphics[width=0.4\textwidth]{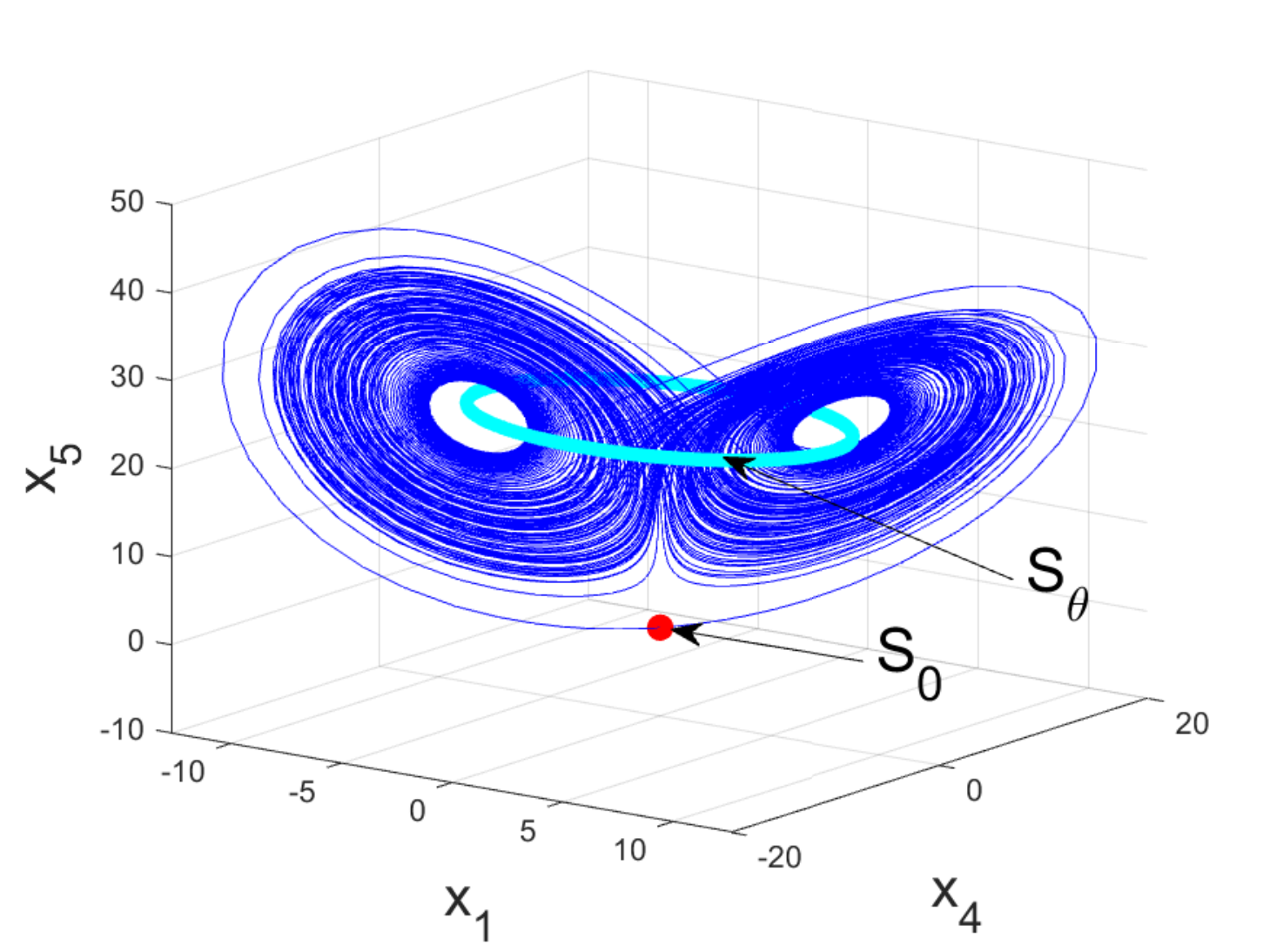}
 }~
 \subfloat[
 {\scriptsize Initial data near an equilibrium $S_{\theta}$}
 ] {
 \label{fig:selfexcited21}
 \includegraphics[width=0.4\textwidth]{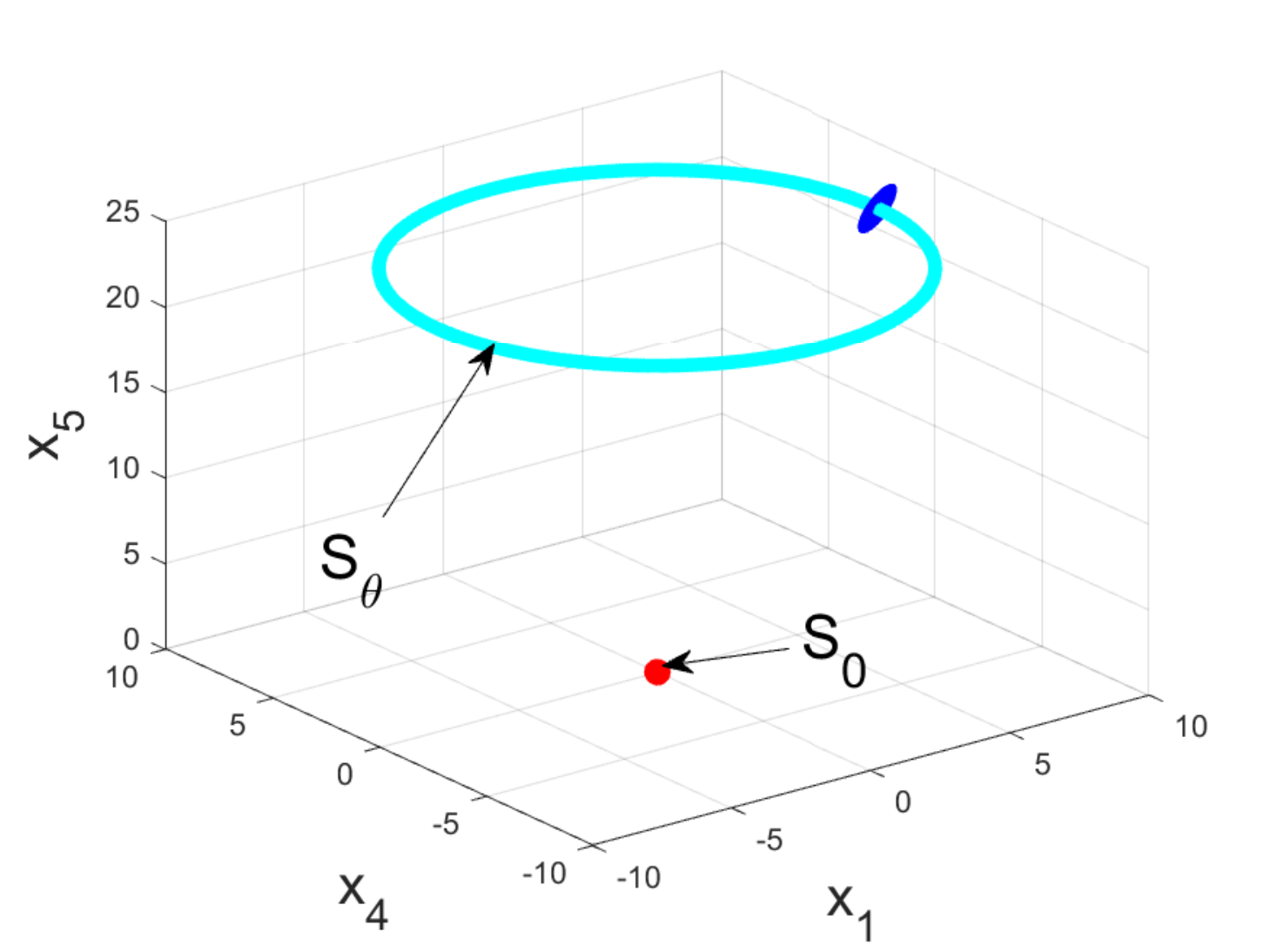}
 }
 \\
 \centering
 \subfloat[
 {\scriptsize Initial data near the equilibrium $S'_{0}$}
 ] {
 \label{fig:projectselfexcited11}
 \includegraphics[width=0.4\textwidth]{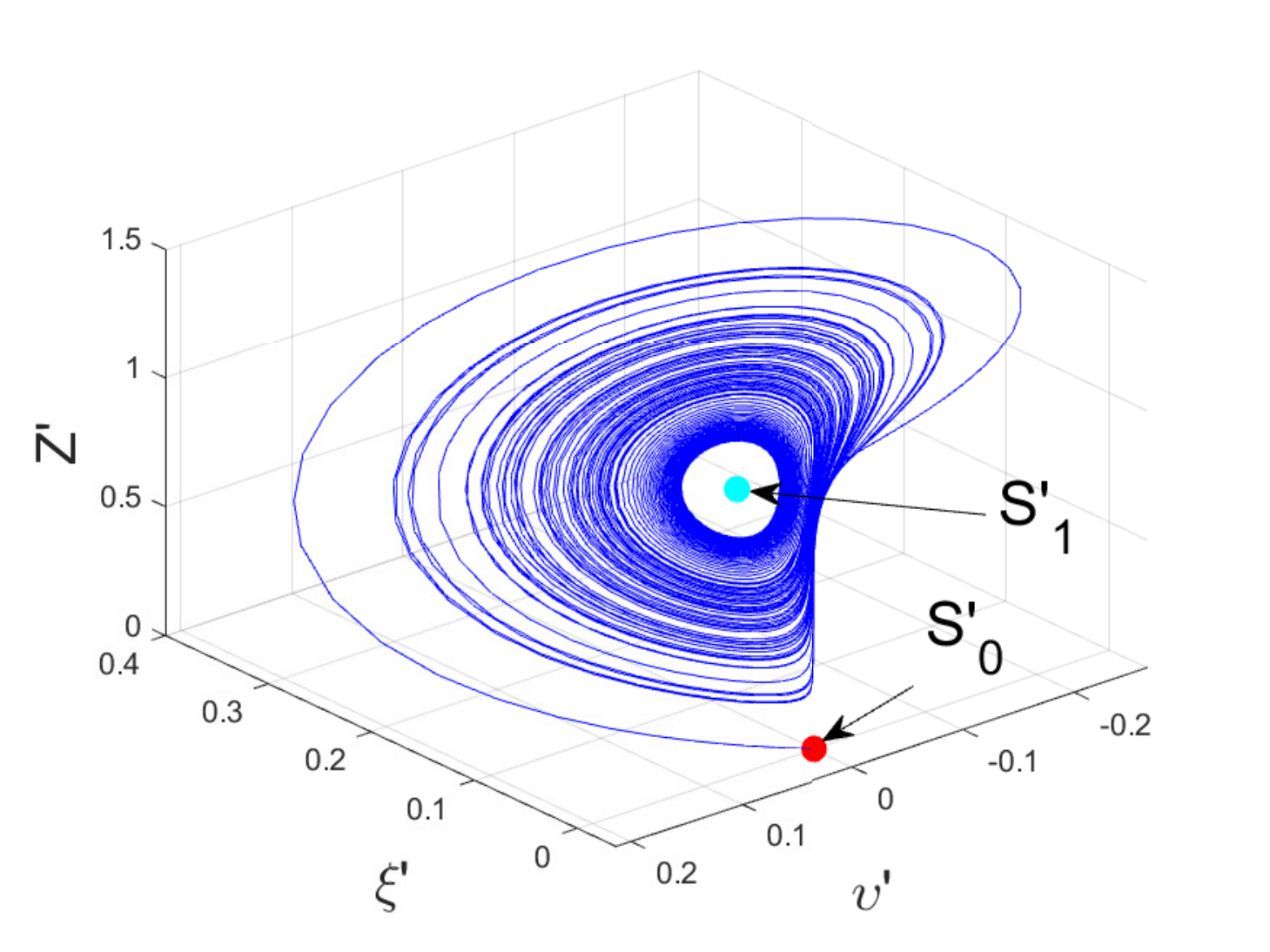}
 }~
 \subfloat[
 {\scriptsize Initial data near the equilibrium $S'_{1}$}
 ] {
 \label{fig:projectselfexcited21}
 \includegraphics[width=0.4\textwidth]{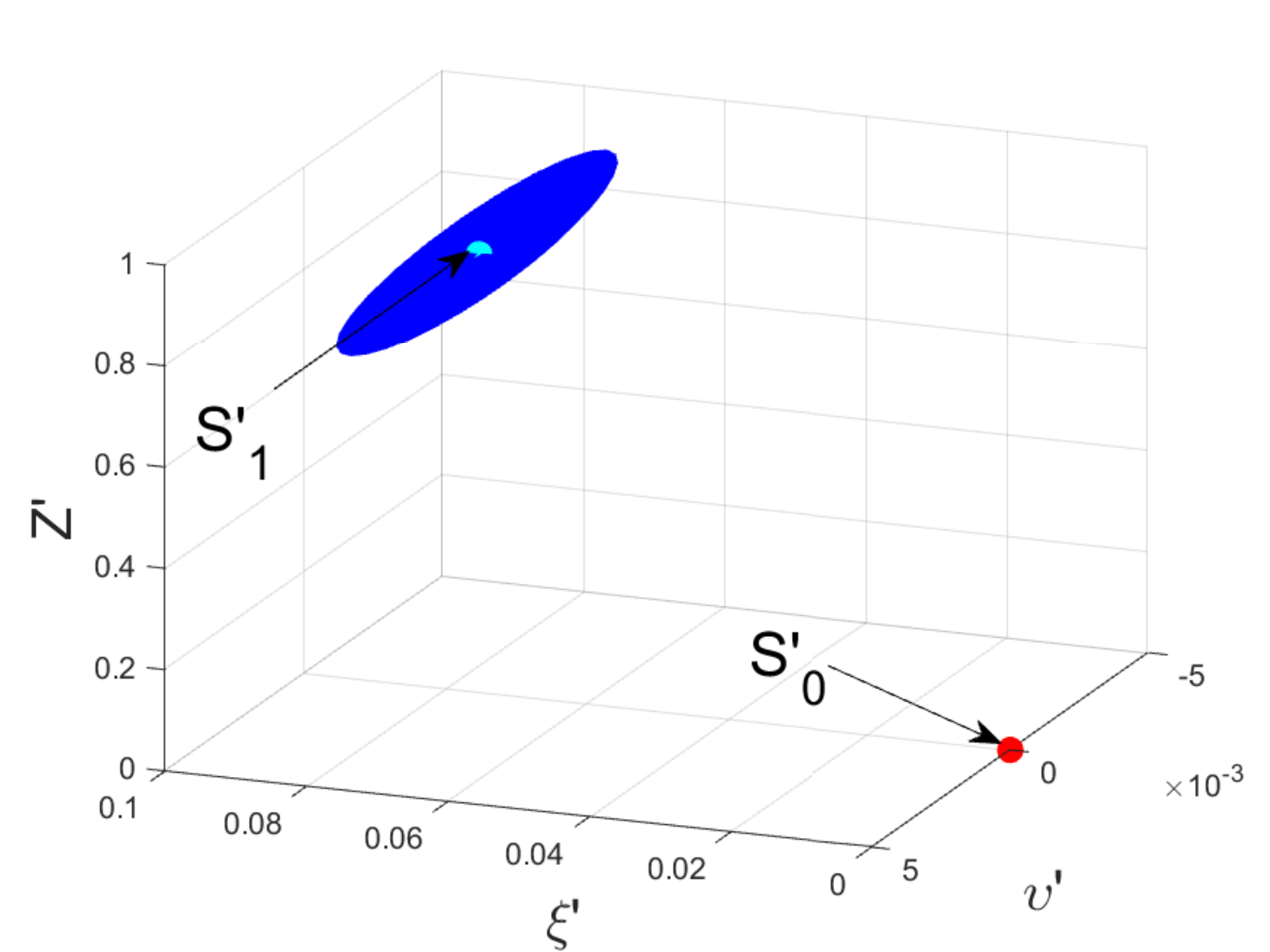}
 }
 \caption{
 (a), (b) Numerical visualization of the self-excited local chaotic attractor in
 system~\eqref{13} with $r_{1} = 24.6$, $\sigma = 10$, $b = \tfrac{8}{3}$, $r_{2} = - e = 0.001$ by a
 trajectory starting in the vicinity of the unstable equilibrium $S_{0}$.
 This attractor coexists with the stable set of equilibria $S_{\theta}$;
 (c), (d) The corresponding images in the projective space $\mathcal{P}$.
 }
 \label{fig:lorenz:selfexcited1}
\end{figure*}
\section{The boundary of practical stability and absence of nontrivial attractors}

Inside the absorbing set, it is possible to study numerically
the presence of nontrivial self-excited and hidden
attractors for parameters $r_{1}, \sigma, r_{2}, b$ not satisfying
conditions (\ref{cond-glonal}) of global stability,
i.e., by fixing $\sigma, r_{2}, b$ and by decreasing $r_{1}$ from $r'_{1c}$.
For $\sigma=10$, $r_{2}=0.001$, $b=\frac{8}{3}$ this gives us the following
region $r_{1}\in(r_{1gs}, r'_{1c})$, where $r_{1gs} \approx 1.96$ and $r'_{1c} \approx 24.73$.

A nontrivial self-excited attractor can be observed
numerically for $24.09 \lessapprox r_{1}<r'_{1c} \approx 24.73$.
In this case of nontrivial multistability, system (\ref{13})
possesses a local chaotic attractor $\mathcal{A}$ which is self-excited with
respect to equilibrium $S_{0}$ and coexists with
the set of trivial attractors $S_{\theta}$ (see Fig.~\ref{fig:lorenz:selfexcited1}).

\vspace{-1mm}
\subsection{Hidden attractor or hidden transient set?}
As we discussed above on the example
of the classical Lorenz model~\eqref{sys:lorenz},
for a system possessing a transient
chaotic set, the time of the transient process depends
strongly on the choice of initial data in the phase space
and also on the parameters of numerical solvers to integrate
a trajectory (e.g., the order of the method, the step of integration,
relative and absolute tolerances) \cite{KuznetsovMKKL-2020}.
This complicates
the task of distinguishing a transient chaotic set
from a sustained chaotic set (attractor) in numerical
experiments.

Consider system \eqref{projec1-system}, which represents
the complex Lorenz system \eqref{sys:complex-Lorenz}
in the projective space $\mathcal{P}$, with $r_{1}=23.5$,
$r_{1}=23.6$ and $r_{1}=23.7$.
For a trajectory with a certain initial point, which is
computed by a certain solver with specific parameters,
we estimate the moment of the end of transient behavior
as the moment when the trajectory falls into a
small vicinity of the stable equilibrium $S'_1$.
Using MATLAB's standard procedure \texttt{ode45} with default
parameters (relative tolerance $10^{-3}$, absolute tolerance
$10^{-6}$) for system \eqref{projec1-system} with parameters
$r_{1}=23.5,\, \sigma=10,\, r_{2} = -e = 0.001,\, b=\frac{8}{3}$
and for initial point $u_{0}=(0.5, 0.5, 0.5, 0.5)$, a transient
chaotic behavior is observed on the time interval $[0,\; 1.78 \times 10^{4}]$,
for initial point $u_{0}=(0.4, 0.4, -0.3, 0.5)$ --- on the time interval $[0,\; 2.7 \times 10^{4}]$
and for initial point $u_{0}=(0.4, 0.4, 0.4, 0.4)$ --- on the time interval $[0,\; 4.6 \times 10^{4}]$.
If we slightly change the bifurcation parameter $r_{1}$, taking
$r_{1}=23.6$ and leaving all other parameters the same,
and choosing initial point $u_{0}=(0.5, 0.5, 0.5, 0.5)$
a transient chaotic behavior is observed on the time interval $[0,\; 3.49\times 10^{7}]$.
For $r_{1}=23.7$ with the same parameters and initial point a transient
chaotic behavior continues over a time interval of more than $[0,\;  10^{8}]$.
Thus, one can observe that a small increment of the bifurcation parameter $r_{1}$
leads to a large increment of the time of transient process.
For $r_{1}\in (23.7, 24.09)$, it is a tough problem to estimate
the time of transient process since the time for numerical integration is very large.
Therefore, inside this interval of parameters
for the particular trajectories
it is hard to say whether the chaotic behavior
is a long-hidden transient process, or a hidden sustained attractor.
Note also that, if we consider parameters
$r_{1}=23.5, \sigma=10, r_{2}=-e=0.001, b=\frac{8}{3}$ with the same initial
points and use \texttt{ode45} procedure with relative tolerance $10^{-6}$,
for all these initial points the
chaotic transient behavior will last over a time interval
of more than $[0,\; 10^{7}]$, and corresponding transient
chaotic sets won't collapse.
To conclude this numerical study of transient chaotic behavior,
we may suggest to specify precisely numerical solver,
its parameters, initial data and time interval, along which
the transient behavior continues.

Further in Appendices A-E, we study in detail the visualization
of hidden transient chaotic sets in the complex Lorenz system in the form~\eqref{13}.
We consider the "laser case", i.e., $e=-r_{2}$, when system~\eqref{13}
has the equilibrium $S_{0}$ and a circle of equilibria $S_{\theta}$.
Two special techniques,
i.e. numerical continuation method (NCM) and perpetual points method,
are utilized to visualize hidden sets.

\section{Conclusion}
In this work, the complex Lorenz system, describing
baroclinic instability in the atmosphere either the physics
of detuned lasers, is considered.
Analytical and numerical results demonstrate
that the complex Lorenz system's dynamics is very rich,
and, in comparison with the real Lorenz system,
the loss of global stability may be connected with
appearance of significantly different dynamical regimes.
As for the Lorenz system,
for the complex Lorenz system, on the one hand, it is demonstrated
the possibility of using the apparatus of Lyapunov functions
to estimate the boundary of global stability.
On the other hand, it is shown that this problem is connected with the localization of
sustained hidden chaotic attractors and their distinguishing from
the hidden long-term transient chaotic sets.
An attempt to study the problem of existence of sustained hidden chaotic attractor
in the complex Lorenz system is made
in the framework of a special analytical transformation,
tacking into account the symmetry of the model.

\section*{Acknowledgement}
This work is supported by the joint Grant from
the Department of Science \& Technology (DST, India)
and the Russian Science Foundation (RSF, Russia)
project 19-41-02002 (sects. I-IV) and
by the General Administration of Missions
(Ministry of Higher Education of Egypt).

\bibliographystyle{elsarticle-num}

\begin{thebibliography}{10}
\expandafter\ifx\csname url\endcsname\relax
  \def\url#1{\texttt{#1}}\fi
\expandafter\ifx\csname urlprefix\endcsname\relax\def\urlprefix{URL }\fi
\expandafter\ifx\csname href\endcsname\relax
  \def\href#1#2{#2} \def\path#1{#1}\fi

\bibitem{KuznetsovMKKL-2020}
N.~Kuznetsov, T.~Mokaev, O.~Kuznetsova, E.~Kudryashova, The {L}orenz system:
  hidden boundary of practical stability and the {L}yapunov dimension,
  Nonlinear Dynamics 102 (2020) 713--732.
\newblock \href {http://dx.doi.org/10.1007/s11071-020-05856-4}
  {\path{doi:10.1007/s11071-020-05856-4}}.

\bibitem{Vidyasagar-1978}
M.~Vidyasagar, Nonlinear Systems Analysis, Prentice-Hall, 1978.

\bibitem{HaddadC-2011}
W.~Haddad, V.~Chellaboina, Nonlinear Dynamical Systems and Control: {A}
  {L}yapunov-Based Approach, Princeton University Press, 2011.

\bibitem{LeonovRS-1992}
G.~Leonov, V.~Reitmann, V.~Smirnova, Nonlocal Methods for Pendulum-like
  Feedback Systems, Teubner, Stuttgart-Leipzig, 1992.

\bibitem{YakubovichLG-2004}
V.~Yakubovich, G.~Leonov, A.~Gelig, Stability of Stationary Sets in Control
  Systems with Discontinuous Nonlinearities, World Scientific, Singapure, 2004,
  [Transl from Russian: A.Kh. Gelig and G.A. Leonov and V.A. Yakubovich, Nauka,
  1978].

\bibitem{BarbashinK-1952}
E.~Barbashin, N.~Krasovsky, On the stability of a motion in the large, Dokl.
  Akad. Nauk SSSR (In Russian) 86~(3) (1952) 453--456.

\bibitem{Lasalle-1960}
J.~LaSalle, Some extensions of {L}iapunov's second method, IRE Transactions on
  circuit theory 7~(4) (1960) 520--527.

\bibitem{LeonovPS-1996}
G.~Leonov, D.~Ponomarenko, V.~Smirnova, Frequency-Domain Methods for Nonlinear
  Analysis. Theory and Applications, World Scientific, Singapore, 1996.

\bibitem{KuznetsovLYYKKRA-2020-ECC}
N.~Kuznetsov, M.~Lobachev, M.~Yuldashev, R.~Yuldashev, E.~Kudryashova,
  O.~Kuznetsova, E.~Rosenwasser, S.~Abramovich, The birth of the global
  stability theory and the theory of hidden oscillations, in: 2020 European
  Control Conference Proceedings, 2020, pp. 769--774.
\newblock \href {http://dx.doi.org/10.23919/ECC51009.2020.9143726}
  {\path{doi:10.23919/ECC51009.2020.9143726}}.

\bibitem{LaSalleL-1961}
J.~LaSalle, S.~Lefschetz, Stability by {L}iapunov's direct method: with
  applications, Academic Press, New-York-London, 1961.

\bibitem{LakshmikanthamLM-1990}
V.~Lakshmikantham, S.~Leela, A.~Martynyuk, Practical stability of nonlinear
  systems, World Scientific, 1990.

\bibitem{AfraimovicBSh-1977}
V.~Afra{\i}movic, V.~Bykov, L.~Silnikov, On the origin and structure of the
  {L}orenz attractor 234~(2) (1977) 336--339.

\bibitem{AuerbachCEGP-1987}
D.~Auerbach, P.~Cvitanovi{\'c}, J.-P. Eckmann, G.~Gunaratne, I.~Procaccia,
  Exploring chaotic motion through periodic orbits, Physical Review Letters
  58~(23) (1987) 2387.

\bibitem{Cvitanovic-1991}
P.~Cvitanovi{\'c}, Periodic orbits as the skeleton of classical and quantum
  chaos, Physica D: Nonlinear Phenomena 51~(1-3) (1991) 138--151.

\bibitem{ShilnikovTCh-2001}
L.~P. Shilnikov, A.~L. Shilnikov, D.~V. Turaev, L.~Chua, Methods of Qualitative
  Theory in Nonlinear Dynamics: {P}art {2}, World Scientific, 2001.

\bibitem{Leonov-2013-IJBC}
G.~Leonov, Shilnikov chaos in {L}orenz-like systems, International Journal of
  Bifurcation and Chaos 23~(03), art. num. 1350058.
\newblock \href {http://dx.doi.org/10.1142/S0218127413500582}
  {\path{doi:10.1142/S0218127413500582}}.

\bibitem{LeonovKM-2015-EPJST}
G.~Leonov, N.~Kuznetsov, T.~Mokaev, Homoclinic orbits, and self-excited and
  hidden attractors in a {L}orenz-like system describing convective fluid
  motion, The European Physical Journal Special Topics 224~(8) (2015)
  1421--1458.
\newblock \href {http://dx.doi.org/10.1140/epjst/e2015-02470-3}
  {\path{doi:10.1140/epjst/e2015-02470-3}}.

\bibitem{LeonovMKM-2020-IJBC}
G.~Leonov, R.~Mokaev, N.~Kuznetsov, T.~Mokaev, Homoclinic bifurcations and
  chaos in the {F}ishing principle for the {L}orenz-like systems, International
  Journal of Bifurcation and Chaos 30, (https://doi.org/S0218127420501205).
\newblock \href {http://dx.doi.org/S0218127420501205}
  {\path{doi:S0218127420501205}}.

\bibitem{Levinson-1944}
N.~Levinson, Transformation theory of non-linear differential equations of the
  second order, Annals of Mathematics (1944) 723--737.

\bibitem{LeonovKV-2011-PLA}
G.~Leonov, N.~Kuznetsov, V.~Vagaitsev, Localization of hidden {C}hua's
  attractors, Physics Letters A 375~(23) (2011) 2230--2233.
\newblock \href {http://dx.doi.org/10.1016/j.physleta.2011.04.037}
  {\path{doi:10.1016/j.physleta.2011.04.037}}.

\bibitem{LeonovK-2013-IJBC}
G.~Leonov, N.~Kuznetsov, Hidden attractors in dynamical systems. {F}rom hidden
  oscillations in {H}ilbert-{K}olmogorov, {A}izerman, and {K}alman problems to
  hidden chaotic attractors in {C}hua circuits, International Journal of
  Bifurcation and Chaos in Applied Sciences and Engineering 23~(1), {a}rt. no.
  1330002.
\newblock \href {http://dx.doi.org/10.1142/S0218127413300024}
  {\path{doi:10.1142/S0218127413300024}}.

\bibitem{Kuznetsov-2016}
N.~Kuznetsov, Hidden attractors in fundamental problems and engineering models.
  {A} short survey, Lecture Notes in Electrical Engineering 371 (2016) 13--25,
  (Plenary lecture at International Conference on Advanced Engineering Theory
  and Applications 2015).
\newblock \href {http://dx.doi.org/10.1007/978-3-319-27247-4\_2}
  {\path{doi:10.1007/978-3-319-27247-4\_2}}.

\bibitem{Kuznetsov-2020-TiSU}
N.~Kuznetsov, Theory of hidden oscillations and stability of control systems,
  Journal of Computer and Systems Sciences International 59~(5) (2020)
  647--668.
\newblock \href {http://dx.doi.org/10.1134/S1064230720050093}
  {\path{doi:10.1134/S1064230720050093}}.

\bibitem{Lorenz-1963}
E.~Lorenz, Deterministic nonperiodic flow, J. Atmos. Sci. 20~(2) (1963)
  130--141.

\bibitem{LeonovB-1992}
G.~Leonov, V.~Boichenko, Lyapunov's direct method in the estimation of the
  {H}ausdorff dimension of attractors, Acta Applicandae Mathematicae 26~(1)
  (1992) 1--60.

\bibitem{Leonov-2018-UMZh-rus}
G.~Leonov, Lyapunov functions in the global analysis of chaotic systems, Ukr.
  Mat. Journal 70~(1) (2018) 40--62, (in Russian).

\bibitem{Sparrow-1982}
C.~Sparrow, The {L}orenz Equations: Bifurcations, Chaos, and Strange
  Attractors, Applied Mathematical Sciences, Springer New York, 1982.

\bibitem{YuanYW-2017-HA}
Q.~Yuan, F.-Y. Yang, L.~Wang, A note on hidden transient chaos in the {L}orenz
  system, International Journal of Nonlinear Sciences and Numerical Simulation
  18~(5) (2017) 427--434.

\bibitem{MunmuangsaenS-2018-HA}
B.~Munmuangsaen, B.~Srisuchinwong, A hidden chaotic attractor in the classical
  {L}orenz system, Chaos, Solitons \& Fractals 107 (2018) 61 -- 66.

\bibitem{GrebogiOY-1983}
C.~Grebogi, E.~Ott, J.~Yorke, Fractal basin boundaries, long-lived chaotic
  transients, and unstable-unstable pair bifurcation, Physical Review Letters
  50~(13) (1983) 935--938.

\bibitem{LaiT-2011}
Y.~Lai, T.~Tel, Transient Chaos: Complex Dynamics on Finite Time Scales,
  Springer, New York, 2011.

\bibitem{DancaK-2017-CSF}
M.-F. Danca, N.~Kuznetsov, Hidden chaotic sets in a {H}opfield neural system,
  Chaos, Solitons \& Fractals 103 (2017) 144--150.
\newblock \href {http://dx.doi.org/https://doi.org/10.1016/j.chaos.2017.06.002}
  {\path{doi:https://doi.org/10.1016/j.chaos.2017.06.002}}.

\bibitem{ChenKLM-2017-IJBC}
G.~Chen, N.~Kuznetsov, G.~Leonov, T.~Mokaev, Hidden attractors on one path:
  {G}lukhovsky-{D}olzhansky, {L}orenz, and {R}abinovich systems, International
  Journal of Bifurcation and Chaos in Applied Sciences and Engineering 27~(8),
  {a}rt. num. 1750115.

\bibitem{GibbonM-1982}
J.~Gibbon, M.~McGuinness, The real and complex {L}orenz equations in rotating
  fluids and lasers, Physica D: Nonlinear Phenomena 5~(1) (1982) 108--122.

\bibitem{FowlerGM-1982}
A.~Fowler, J.~Gibbon, M.~McGuinness, The complex {L}orenz equations, Physica D:
  Nonlinear Phenomena 4~(2) (1982) 139--163.

\bibitem{FowlerGM-1983}
A.~Fowler, J.~Gibbon, M.~McGuinness, The real and complex {L}orenz equations
  and their relevance to physical systems, Physica D: Nonlinear Phenomena
  7~(1-3) (1983) 126--134.

\bibitem{Ning-1990}
C.~Ning, H.~Haken, Detuned lasers and the complex {L}orenz equations:
  subcritical and supercritical hopf bifurcations, Physical Review A 41~(7)
  (1990) 3826.

\bibitem{NeimarkL-1992}
Y.~I. Neimark, P.~S. Landa, Stochastic and Chaotic Oscillations, Kluwer
  Academic Publishers, Dordrecht, The Netherlands, 1992.

\bibitem{RauhHA-1996}
A.~Rauh, L.~Hannibal, N.~Abraham, Global stability properties of the complex
  {L}orenz model, Physica D: Nonlinear Phenomena 99~(1) (1996) 45--58.

\bibitem{LeonovR-1987}
G.~Leonov, V.~Reitmann, Attraktoreingrenzung fur nichtlineare Systeme (in
  German), Teubner, 1987.

\bibitem{Chueshov-2002}
I.~D. Chueshov, Introduction to the Theory of Infinite-Dimensional Dissipative
  Systems, ACTA Scientific Publishing House, Kharkov, 2002.

\bibitem{Leonov-1991-Vest}
G.~Leonov, On estimations of {H}ausdorff dimension of attractors, Vestnik St.
  Petersburg University: Mathematics 24~(3) (1991) 38--41, [Transl. from
  Russian: Vestnik Leningradskogo Universiteta. Mathematika, 24(3), 1991,
  pp.~41-44].

\bibitem{KuznetsovAL-2016}
N.~Kuznetsov, T.~Alexeeva, G.~Leonov, Invariance of {L}yapunov exponents and
  {L}yapunov dimension for regular and irregular linearizations, Nonlinear
  Dynamics 85~(1) (2016) 195--201.
\newblock \href {http://dx.doi.org/10.1007/s11071-016-2678-4}
  {\path{doi:10.1007/s11071-016-2678-4}}.

\bibitem{Kuznetsov-2016-PLA}
N.~Kuznetsov, The {L}yapunov dimension and its estimation via the {L}eonov
  method, Physics Letters A 380~(25-26).
\newblock \href {http://dx.doi.org/10.1016/j.physleta.2016.04.036}
  {\path{doi:10.1016/j.physleta.2016.04.036}}.

\bibitem{Leonov-2018-UMZh}
G.~Leonov, Lyapunov functions in the global analysis of chaotic systems,
  Ukrainian Mathematical Journal 70~(1) (2018) 42--66.
\newblock \href {http://dx.doi.org/10.1007/s11253-018-1487-y}
  {\path{doi:10.1007/s11253-018-1487-y}}.

\bibitem{SiminosC-2011}
E.~Siminos, P.~Cvitanovi{\'c}, Continuous symmetry reduction and return maps
  for high-dimensional flows, Physica D: Nonlinear Phenomena 240~(2) (2011)
  187--198.

\bibitem{FroehlichC-2012}
S.~Froehlich, P.~Cvitanovi{\'c}, Reduction of continuous symmetries of chaotic
  flows by the method of slices, Communications in Nonlinear Science and
  Numerical Simulation 17~(5) (2012) 2074--2084.

\bibitem{VladimirovTD-1997}
A.~Vladimirov, V.~Toronov, V.~Derbov, On the complex {L}orenz equations, in:
  Proceedings of SPIE-The International Society for Optical Engineering, 1997,
  pp. 97--106.

\bibitem{VladimirovTD-1998-IJBC}
A.~Vladimirov, V.~Toronov, V.~Derbov, The complex {L}orenz model: {G}eometric
  structure, homoclinic bifurcation and one-dimensional map, International
  Journal of Bifurcation and Chaos 8~(04) (1998) 723--729.

\bibitem{VladimirovTD-1998-TP}
A.~Vladimirov, V.~Toronov, V.~Derbov, Properties of the phase space and
  bifurcations in the complex {L}orenz model, Technical Physics 43~(8) (1998)
  877--884.

\bibitem{KobayashiN-1963}
S.~Kobayashi, K.~Nomizu, Foundations of differential geometry, Vol. 1,2, New
  York, London, 1963.

\bibitem{KuznetsovLMPS-2018}
N.~Kuznetsov, G.~Leonov, T.~Mokaev, A.~Prasad, M.~Shrimali, Finite-time
  {L}yapunov dimension and hidden attractor of the {R}abinovich system,
  Nonlinear Dynamics 92~(2) (2018) 267--285.
\newblock \href {http://dx.doi.org/10.1007/s11071-018-4054-z}
  {\path{doi:10.1007/s11071-018-4054-z}}.

\bibitem{LeonovKM-2015-CNSNS}
G.~Leonov, N.~Kuznetsov, T.~Mokaev, Hidden attractor and homoclinic orbit in
  {L}orenz-like system describing convective fluid motion in rotating cavity,
  Communications in Nonlinear Science and Numerical Simulation 28 (2015)
  166--174.
\newblock \href {http://dx.doi.org/10.1016/j.cnsns.2015.04.007}
  {\path{doi:10.1016/j.cnsns.2015.04.007}}.

\bibitem{Prasad-2015}
A.~Prasad, Existence of perpetual points in nonlinear dynamical systems and its
  applications, International Journal of Bifurcation and Chaos 25~(2), {a}rt.
  num. 1530005.

\bibitem{DudkowskiJKKLP-2016}
D.~Dudkowski, S.~Jafari, T.~Kapitaniak, N.~Kuznetsov, G.~Leonov, A.~Prasad,
  Hidden attractors in dynamical systems, Physics Reports 637 (2016) 1--50.
\newblock \href {http://dx.doi.org/10.1016/j.physrep.2016.05.002}
  {\path{doi:10.1016/j.physrep.2016.05.002}}.

\bibitem{Prasad-2016}
A.~Prasad, A note on topological conjugacy for perpetual points, International
  Journal of Nonlinear Science 21~(1) (2016) 60--64.

\end{thebibliography}

\clearpage
\appendix

\section{Localization via numerical continuation method}\label{appendix:A}
One of the effective methods for numerical localization
of hidden attractors in multidimensional dynamical
systems is based on the homotopy and numerical continuation
method (NCM). It is based on the assumption
that the position of the attractor changes continuously
with the parameters changing. The idea is to construct
a sequence of similar systems such that for the first
(starting) system, the initial point for numerical computation
of oscillating solution (starting attractor) can
be obtained analytically. For example, it is often possible
to consider the starting system with a self-excited
starting attractor; then, the transformation of this starting
attractor in the phase space is tracked numerically
while passing from one system to another; the last system
corresponds to the system in which a hidden attractor
is searched \cite{KuznetsovMKKL-2020,KuznetsovLMPS-2018,LeonovKM-2015-CNSNS}.
\par In our experiment, we fix parameters $ \sigma, r_{2}, e, b$, define parameters $r_{1}=26.1-\varepsilon$ and $\sigma=\sigma^{\ast}=10$. For $r_{2}=0.001$, $e=-0.001$, $b=8/3$, $\varepsilon=0.7$, we obtain $r_{1}=r_{1}^{0}=25.4$ and take $P_{0}(r_{1}^{0},\sigma^{\ast})$ as
the initial point of line segment on the plane $(r_{1},\sigma)$.
The eigenvalues of the Jacobian matrix at the equilibria $S'_{0}$, $S'_{1}$  of system \eqref{projec1-system} for these parameters are the following:
\begin{equation*}
\begin{aligned}
S'_{0} :  2.00002, -0.1707, -0.7042, -0.7042,\\
S'_{1} : 0.0013 \pm  0.6238i, -0.8775, -0.7042.
\end{aligned}
\end{equation*}
\par Consider on the plane $(r_{1}, \sigma)$ a line segment, intersecting
a boundary of stability domain of the equilibria $S_{\theta}$ with the final point $P_{2}(r_{1}^{2},\sigma^{\ast})$, where $r_{1}^{2}=r_{1}^{0}-2\varepsilon=24$, i.e. the equilibrium $S_{\theta}$ changes from saddle point to stable focus-node
\begin{equation*}
\begin{aligned}
S'_{0} : 2.00003, -0.1758,  -0.7253,-0.7253 \\
S'_{1} :  -0.0015 \pm 0.6257i, -0.8982, -0.7253.
\end{aligned}
\end{equation*}

The initial point $P_{0}(r_{1}^{0},\sigma^{\ast})$ corresponds to the
parameters for which in the system \eqref{projec1-system}, there exists a self-excited
attractor. Then for the considered line segment,
a sufficiently small partition step is chosen. At each
iteration step of the procedure, an attractor in the phase
space of the system \eqref{projec1-system} is computed. The last computed
point at each step is used as the initial point for the
computation at the next step. In this experiment, we
use NCM with 3 steps on the path $P_{0}(r_{1}^{0},\sigma^{\ast})\rightarrow P_{1}(r_{1}^{1},\sigma^{\ast})\rightarrow P_{2}(r_{1}^{2},\sigma^{\ast}),$ with
$r_{1}^{1}=\frac{1}{2}(r_{1}^{0}+r_{1}^{2})$ (see Fig.~\ref{fig:path}). At the first step, we have
a self-excited attractor with respect to unstable equilibria $S_{0}$ and $S_{\theta}$; at the second step, the equilibria $S_{\theta}$ become
stable, but the attractor remains self-excited with respect to equilibrium $S_{0}$; at the third step,
it is possible to visualize a hidden chaotic set of system \eqref{projec1-system}
(see Fig.~\ref{fig:project:lorenz:attr:hid}).
In Fig.~\ref{fig:hidden:Loenz:project} visualizations of the hidden chaotic set
in both initial space $\mathcal{H}$ and projective space $\mathcal{P}$ are depicted.
\begin{figure}[!ht]
  \centering
  \includegraphics[width=\columnwidth]{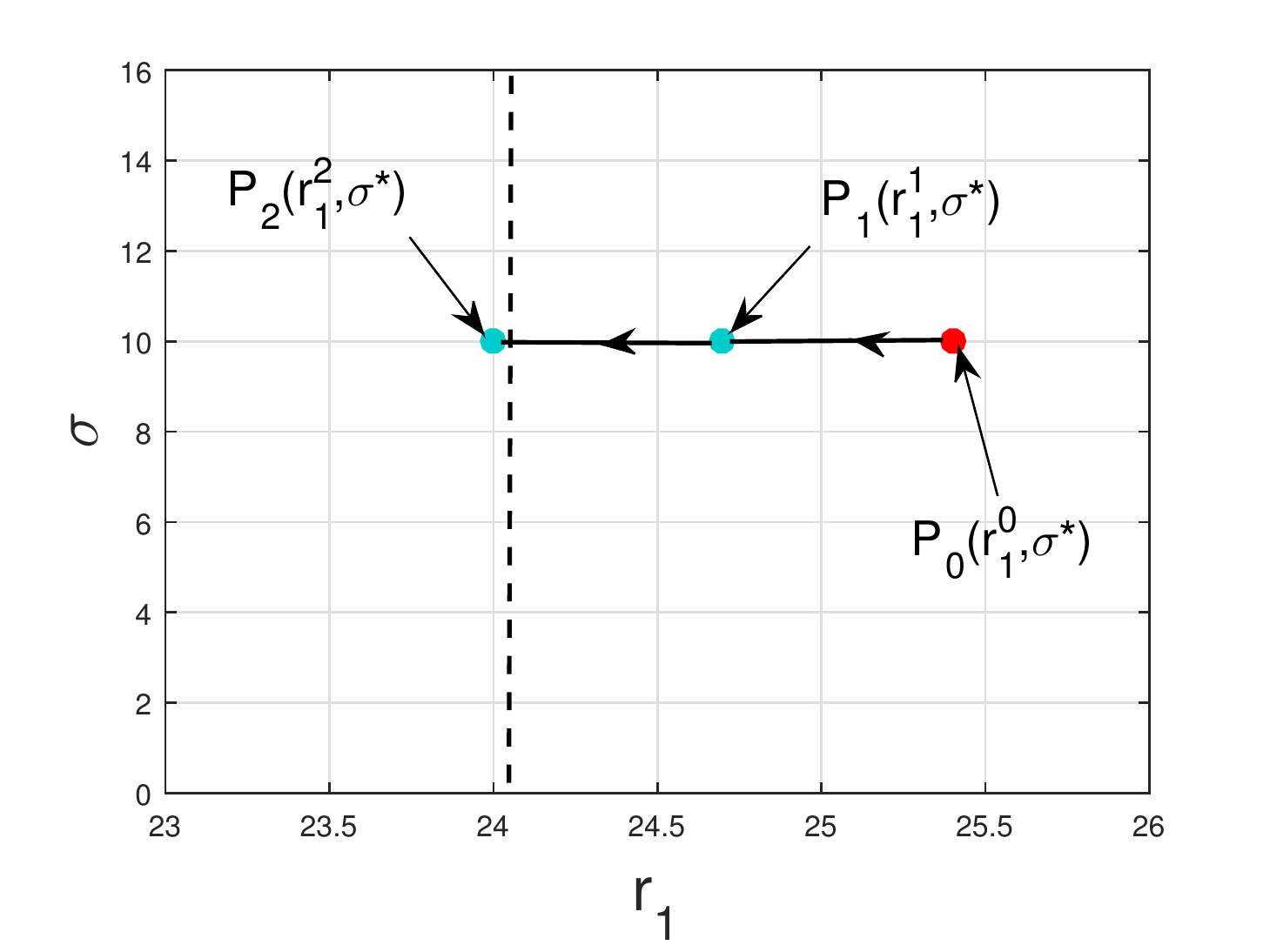}
  \caption{Path
  $P_{0}(r_{1}^{0},\sigma^{\ast}) \rightarrow
   P_{1}(r_{1}^{1},\sigma^{\ast}) \rightarrow P_{2}(r_{1}^{2},\sigma^{\ast})$,
   in parameters plane $(r_{1},\sigma)$ for the localization of hidden chaotic set
   in system~\eqref{projec1-system} with $r_{2}=0.001$, $e=-0.001$, $b=\frac{8}{3}$.
   Here $r_{1}^{0}=25.4, r_{1}^{1}=24.7, r_{1}^{2}=24, \sigma=\sigma^{\ast}=10$;\\
   $(\bullet)~P_{0}(r_{1}^{0},\sigma^{\ast})$ : self-excited attractor w.r.t. $S'_{0,1}$;\\
   $(\bullet)~P_{1}(r_{1}^{1},\sigma^{\ast})$ : self-excited attractor w.r.t. $S'_{0}$;\\
   $(\bullet)~P_{2}(r_{1}^{2},\sigma^{\ast})$ : hidden chaotic set.}
   \label{fig:path}
\end{figure}

\begin{figure*}[!ht]
  \centering
  \includegraphics[width=\textwidth]{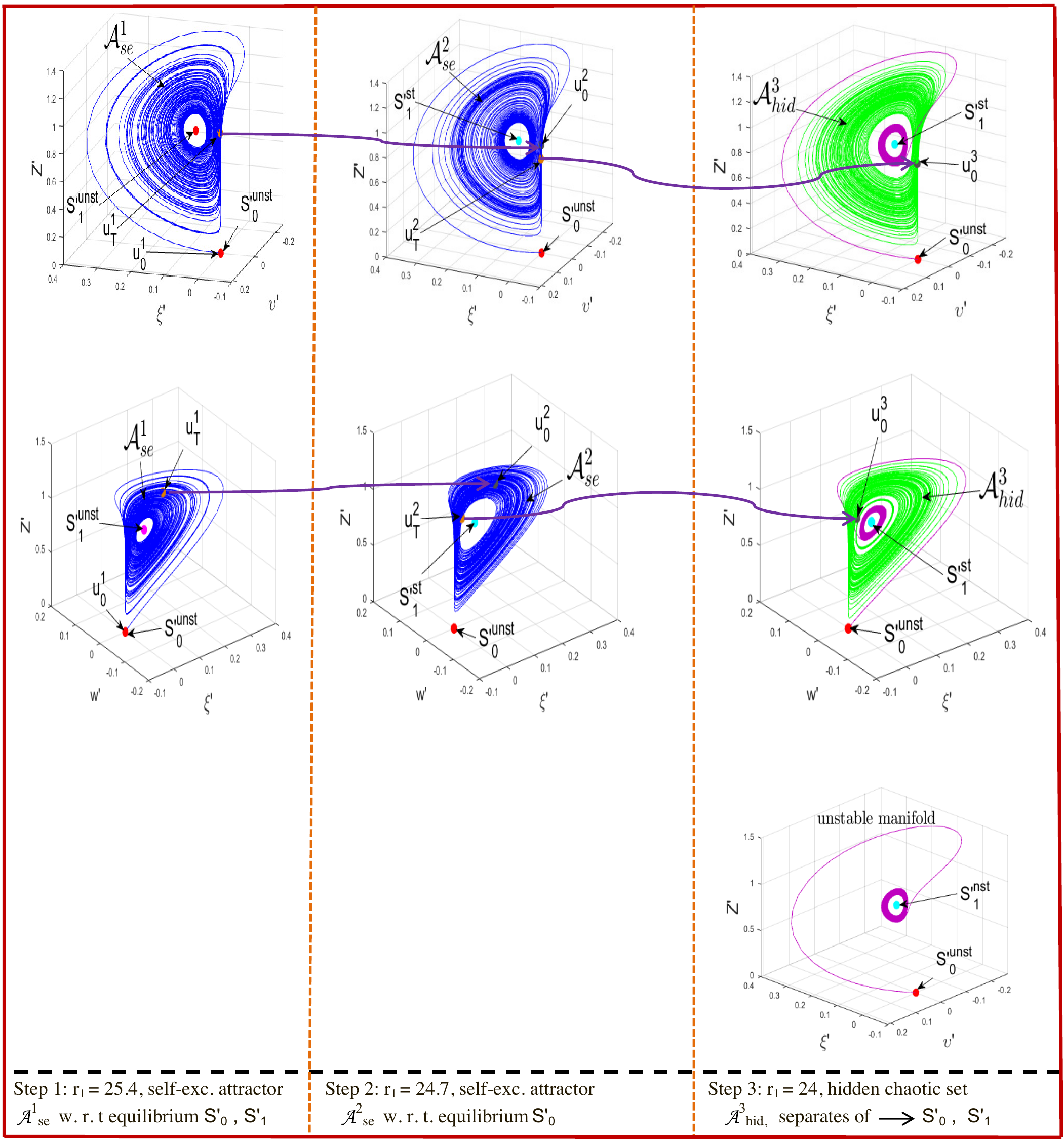}
  \caption{Localization, by NCM, of hidden chaotic set in system~\eqref{projec1-system}
  with $(\sigma, r_{1}, r_{2}, e, b)=( 10,  24, 0.001, -0.001, \frac{8}{3})$.
  Trajectories $ u^{i}(t)=(x_{1},x_{2},x_{3},x_{4},x_{5})$ (blue and green)
  are defined on the time interval $[0, T=10^{3}],$ and initial point (gray)
  on $(i+1)$-th iteration is defined as $u_{0}^{i+1}:=u^{i}_{T}$ (violet arrows),
  where $u^{i}_{T}=u^{i}(T)$ is a final point (orange).}
  \label{fig:project:lorenz:attr:hid}
\end{figure*}
 \begin{figure*}[!ht]
  \centering
  \subfloat[
  {\scriptsize Initial space $\mathcal{H}$ }
  ] {
  \label{fig:lorenz:attr:hid}
  \includegraphics[width=0.4\textwidth]{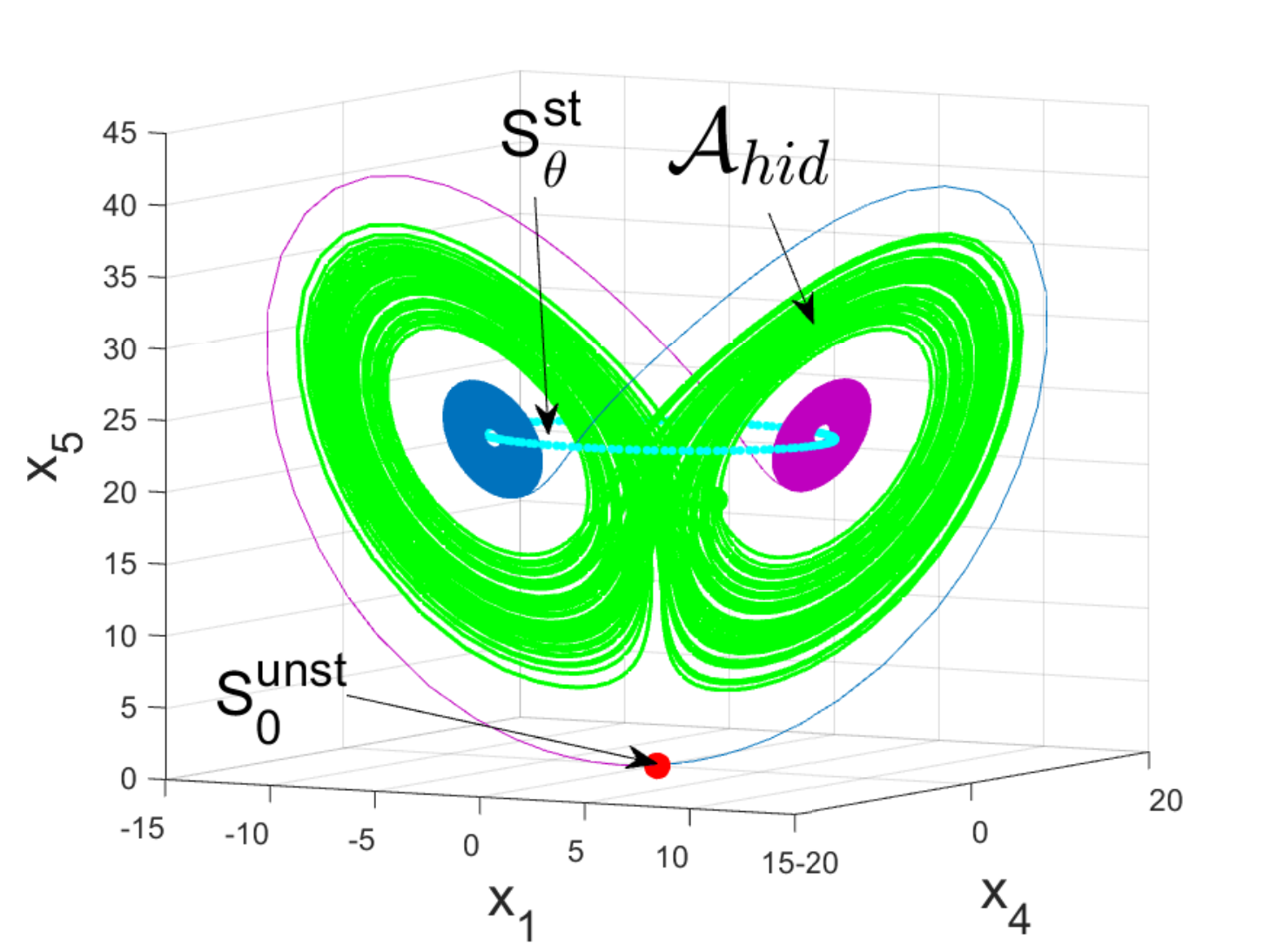}
  }~
  \subfloat[
  {\scriptsize Projective space $P$}
  ] {
  \label{fig:project_hidden1}
  \includegraphics[width=0.4\textwidth]{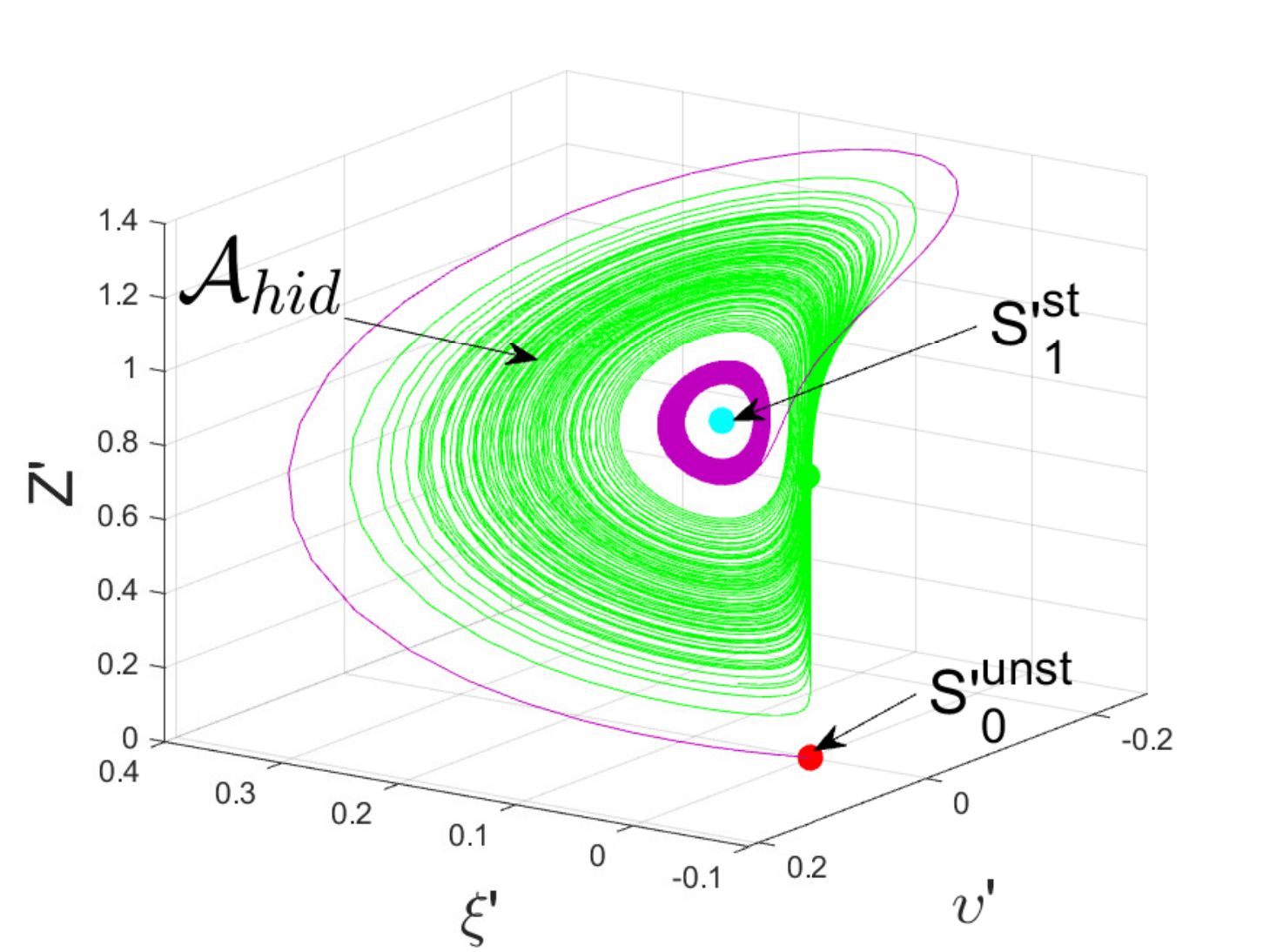}
  }
  \caption{
   Localization of hidden chaotic set in system (\ref{13}) and the corresponding image in system \eqref{projec1-system}
   for $(\sigma, r_{1}, r_{2}, e, b)=( 10, 24, 0.001, -0.001, 8/3)$.
  }
  \label{fig:hidden:Loenz:project}
 \end{figure*}

For $\sigma=10, r_{1}=24, r_{2}=0.002, e=-0.001, b=\frac{8}{3}$,
here the parameter $r_{2}$ is slightly different from that in the above case and $e\neq-r_{2}$.
Therefore $S_{0}$ is the only equilibria of the system \eqref{13}.
In this case, we can also visualize a hidden chaotic set in the system \eqref{13}.
To verify that the set in Fig.~\ref{fig:hiddem:1} is hidden,
around equilibrium $S_{0}$, we choose a small spherical
vicinity of radius $\delta=0.2$ and take $N$ random initial points on it
(in our experiment, $N = 100$ as in Fig.~\ref{fig:sphereinitialpoints}).
Using MATLAB, we integrate system \eqref{13} with these initial points to
explore the obtained trajectories.
We repeat this procedure several times in order to get different initial
points for trajectories on the sphere.
We get the following results: all the obtained trajectories tend to
the torus and do not tend to the chaotic set (see Fig.~\ref{fig:spheretorus}),
this torus is closer\footnote{
  The distance between the torus and the equilibria $S_{\theta}$ can be estimated
  by considering the last point on torus and plane,
  which contains the set of equilibria $S_{\theta}$.
  If we consider in $\mathbb{R}^{5}$ the point $P(x_{1}^{\ast}, x_{2}^{\ast},
  x_{3}^{\ast}, x_{4}^{\ast}, x_{5}^{\ast})$ and the plane
  $A(x_{1}- x_{01})+B(x_{2}-x_{02})+C(x_{3}-x_{03})+D(x_{4}-x_{04})+E(x_{5}-x_{05})=0$,
  then the distance from the point $P$ to the plane is:
  $$
    d = \tfrac{|A x_{1}^{\ast}+B x_{2}^{\ast}+C x_{3}^{\ast}+D x_{4}^{\ast}+E x_{5}^{\ast}+F|}
    {\sqrt{A^{2}+B^{2}+C^{2}+D^{2}+E^{2}}}.
  $$
  To determine the equation of a plane in $\mathbb{R}^{5}$
  one can consider five points, which do not lie on the same straight line.
  The normal vector to a plane can be obtained by taking
  the cross product of four vectors on this plane.
  In our case equation of the plane which contains the set of equilibria $S_{\theta}$ is:
  $x_{5}-23=0$.
  From Fig.~\ref{fig:self:torus} the last point on the
  torus is $P(-2.5011, 7.6357, -2.5958, 7.9228, 22.8110)$,
  then the distance from this point to the plane $x_{5}-23=0$ is:
  $d = 0.1890$.
 } to the set of equilibria $S_{\theta}$ (see Fig.~\ref{fig:self:torus}).
 This gives us a reason to classify the chaotic set,
 obtained in the system~\eqref{13}, as the hidden one (see Fig.~\ref{fig:hiddem:1}).
 \begin{figure*}[!h]
  \centering
  \subfloat[
  {\scriptsize }
  ] {
  \label{fig:self:torus}
  \includegraphics[width=0.4\textwidth]{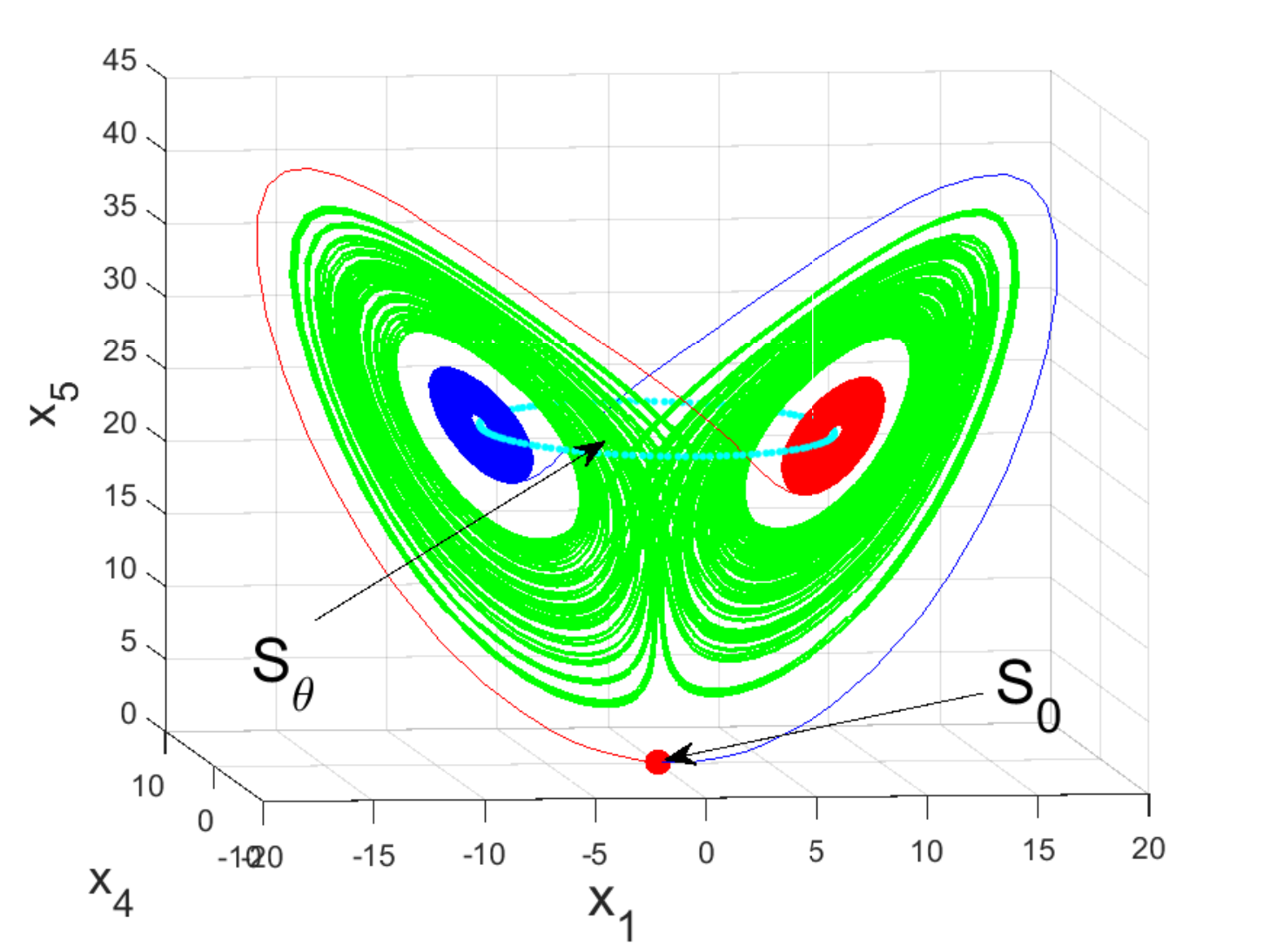}
  }~
  \subfloat[
  {\scriptsize }
  ] {
  \label{fig:hiddem:1}
  \includegraphics[width=0.4\textwidth]{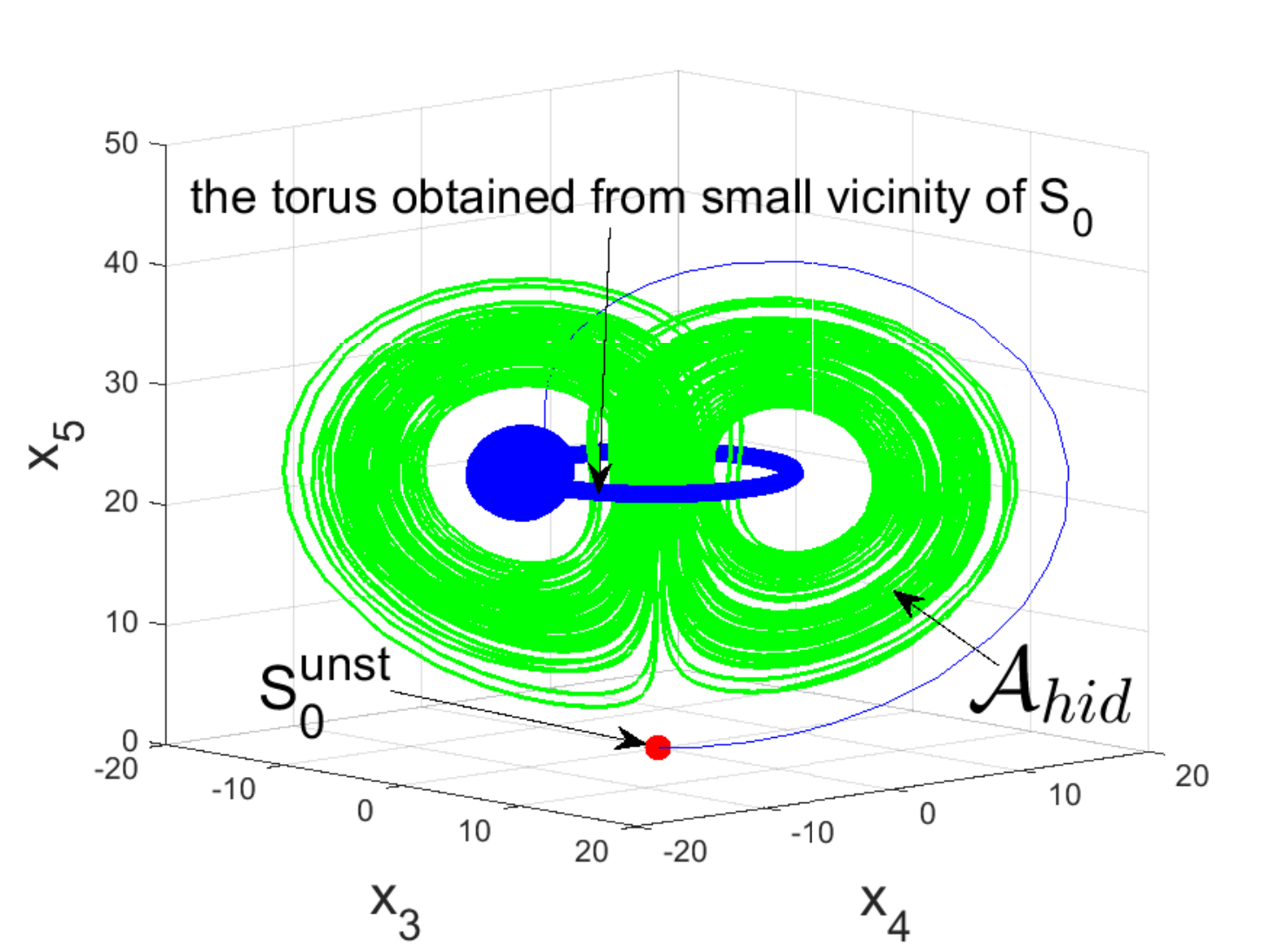}
  }
  \caption{
   (a) A trajectory that attracts to a torus (blue) in the time interval $[0, 7000]$,
   which closer to the equilibria $S_{\theta}$ (magenta)
   (b) Visualization of a hidden chaotic set in system \eqref{13} with time interval $[0, 100]$,
   $(\sigma, r_{1}, r_{2}, e, b)=(10,  24, 0.002, -0.001, \frac{8}{3})$
   and the initial point is $(5, 5, 5, 5, 5)$.
  }
  \label{fig:lorenz:attr:hid1}
 \end{figure*}
 \begin{figure*}[!h]
  \centering
  \subfloat[
  {\scriptsize }
  ] {
  \label{fig:spheretorus}
  \includegraphics[width=0.4\textwidth]{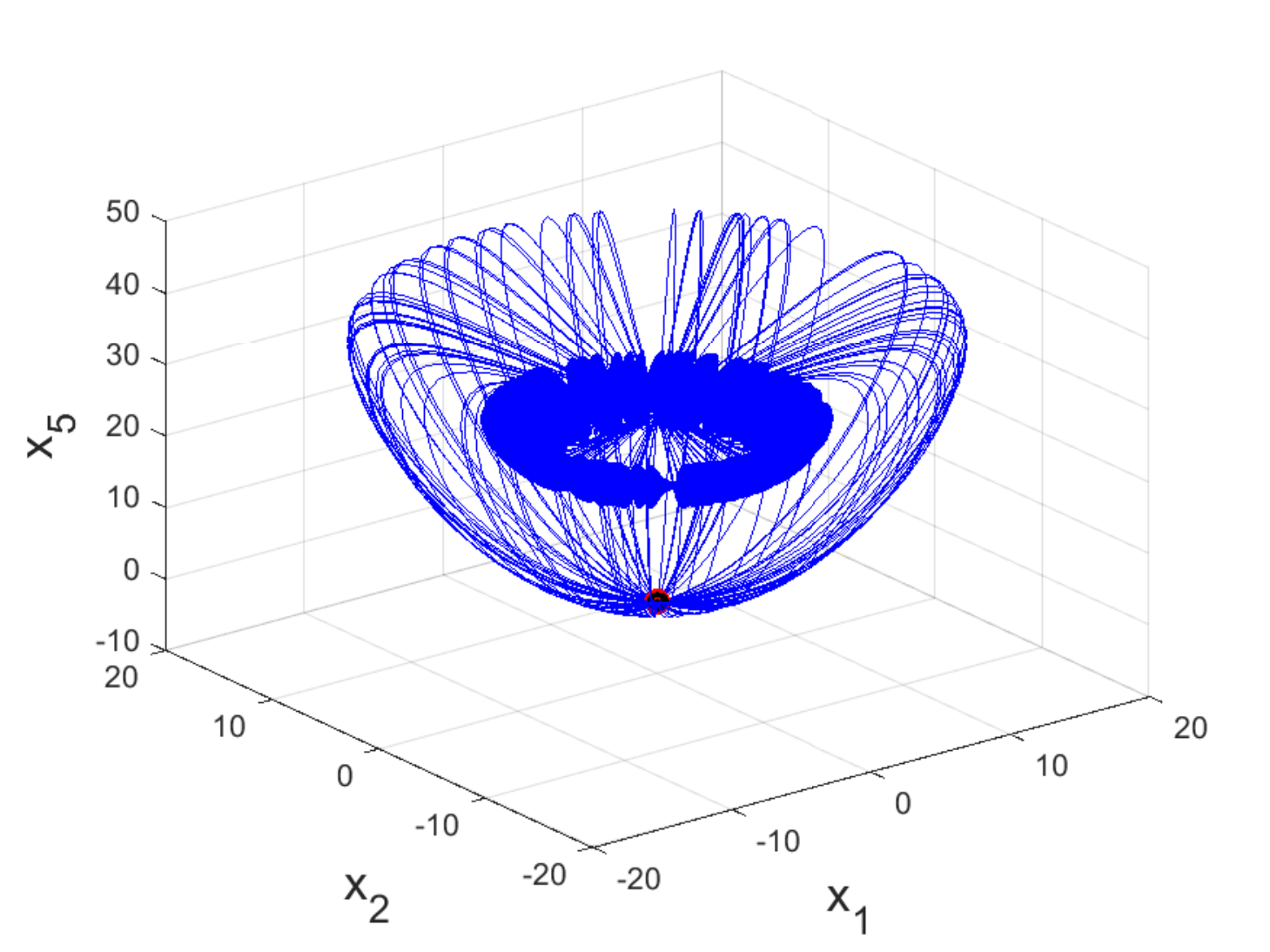}
  }~
  \subfloat[
  {\scriptsize }
  ] {
  \label{fig:sphereinitialpoints}
  \includegraphics[width=0.4\textwidth]{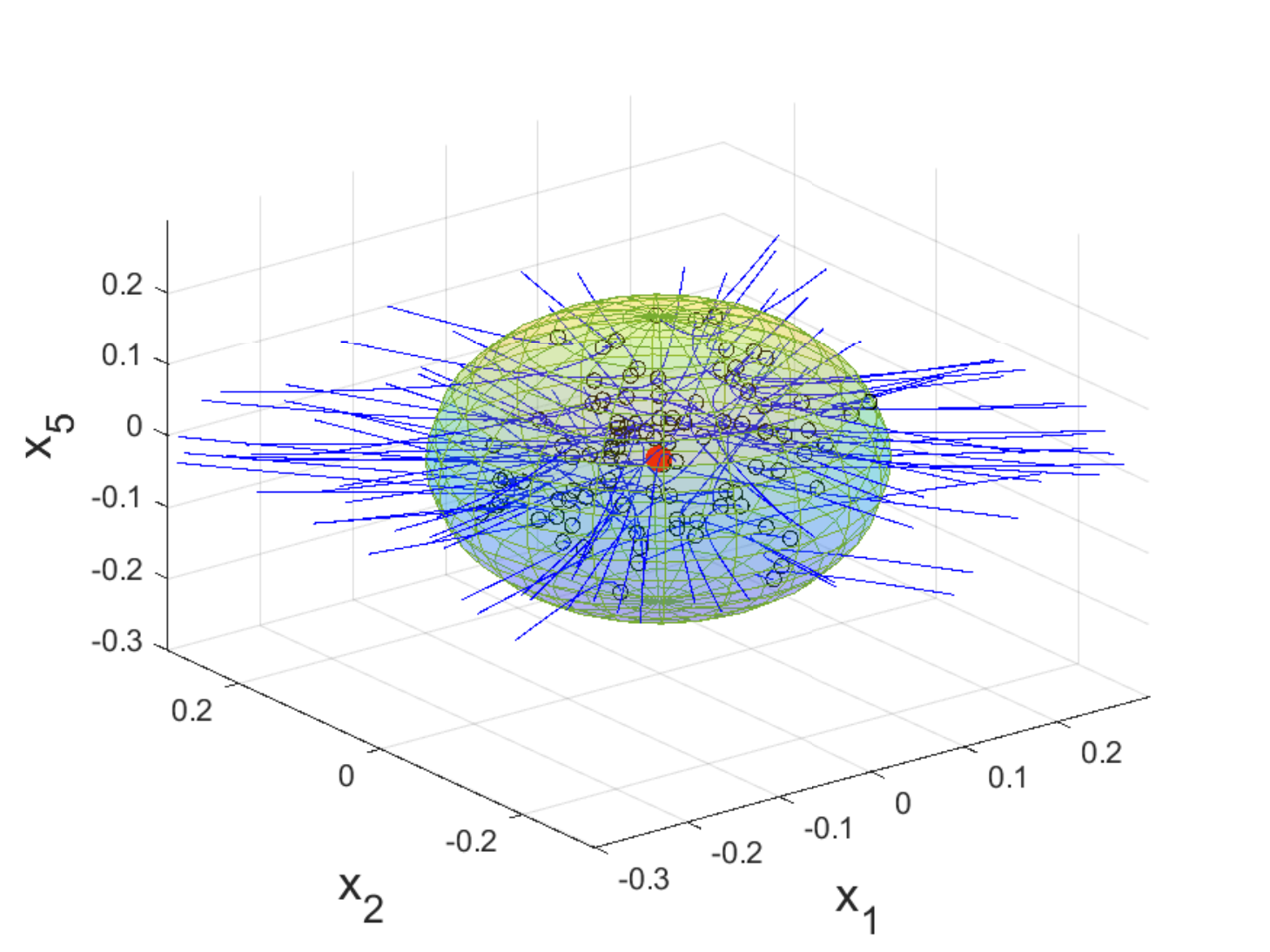}
  }
  \caption{(a) Visualization of trajectories that tend to the torus in the time interval $[0, 7000],$ with starting points from a spherical of random points in a vicinity of $S_{0}$; (b) Magnification of the area around $S_{0}$ to show the spherical of random points and the considered 100 trajectories that tend to the torus.}
  \label{fig:lorenz:spheretorus}
 \end{figure*}

\clearpage
\section{Localization using perpetual points}
Consider system \eqref{13} as an autonomous differential
equation of general form:
\begin{equation}\label{18}
  \dot{u}=f(u),
\end{equation}
where $u=(x_{1},x_{2},x_{3},x_{4},x_{5}) \in \mathbb{R}^{5}$ represents
the right-hand side of system \eqref{13}.
\par The equilibrium points of a dynamical system are the
ones at which the velocity and acceleration of the system
simultaneously become zero.~In this subsection, we
state that there are points, termed as perpetual points
\cite{Prasad-2015}, which may help to visualize hidden attractors.
\par For system \eqref{18}, the equilibrium points $u_{ep}$ are
defined by the equation~$ \dot{u} = f (u_{ep}) = 0$. Consider
a derivative of system \eqref{18} with respect to time:
\begin{equation}\label{perpet}
\ddot{u}=J(u)f(u)=g(u),
\end{equation}
where $J(u)=\big[\frac{\partial f_{i}(u)}{\partial u_{j}}\big]_{i,j=1}^{5}$ is the Jacobian matrix.
Here, $g(u)$ may be termed as an acceleration vector.
System~\eqref{perpet} shows the variation of acceleration in the
phase space. Similar to the equilibrium points estimation, where
we set the velocity vector to zero, we can also get a set
of points, where $\ddot{u}=g(u_{pp})=0$ in \eqref{perpet}, i.e. the points
corresponding to the zero acceleration. At these points,
the velocity  $\dot{u}$ may be either zero or nonzero. This set
includes the equilibrium points $u_{ep}$ with zero velocity
as well as a subset of points with nonzero velocity.
These nonzero velocity points $u_{pp}$ are termed as perpetual
points \cite{DudkowskiJKKLP-2016,Prasad-2015,Prasad-2016}.
\begin{remark}
The analytical formula of the perpetual points of system~\eqref{13} can not be derived,
meanwhile the numerical solution of the algebraic system $\ddot{u}=g(u_{pp})=0$
is $x_{1}=x_{2}=x_{3}=x_{4}=x_{5}=0$, which coincide with the equilibrium $S_{0}$,
so system~\eqref{13} has no perpetual points.
For system~\eqref{projec1-system}, which structure seems simpler than the original system \eqref{13},
$x'_{1}=x'_{2}=x'_{3}=x'_{4}=x'_{5}=0$ is the only solution of the system $\ddot{u}=g(u_{pp})=0$,
thus, there are also no nonzero velocity points $u_{pp}$.
\end{remark}

\section{Numerical verification of basins of attraction near the zero equilibrium point}

For the system \eqref{13} the eigenvalues at $S_{0}$ are: $10.6323 \pm 0.0003i, -21.6323 \pm 0.0007i,  -2.6667.$ So the equilibria $S_{0}$ has 2-dimensional unstable manifold. For the fist two eigenvalues $\lambda_{1,2}=10.6323 \pm 0.0003i$ with positive real part, the corresponding eigenvectors are:
$v_{1}=(-0.3084,-0.3084i,-0.6363,-0.6363i,0)$ and $v_{2}=(-0.3084,0.3084i,-0.6363,0.6363i,0)$ respectively. In the case of complex eigenvalues we can use the following formula to plot trajectories in a vicinity of $S_{0}$
\begin{equation}\label{trajectories}
 u(t)=\Omega~e^{\alpha t}[\eta \sin(\beta t+\delta)+\mu \cos(\beta t+\delta)],
\end{equation}
where $\Omega$ is a constant that represent the size of the vicinity and should be small, $\alpha$, $\beta$ are real and imaginary parts of the eigenvalues respectively, $\delta$ is an auxiliary angle, $\eta ={\rm Re}~v_1$ and $\mu = {\rm Im}~v_1$.
Initial points of unstable manifolds of $S_{0}$ can be obtained by putting $t=0$ in Eq.~(\ref{trajectories}) as
\begin{equation}\label{initialpoints}
  u(0)=\Omega [\eta \sin(\delta)+\mu \cos(\delta)],
\end{equation}
In Fig.~\ref{fig:localstableunstabe} we plotted local stable, unstable manifolds of $S_{0}$ and random initial points on local 2-D unstable manifolds of $S_{0}$.
Fig.~\ref{fig:Projectionsx1x2x5} shows unstable manifolds with random initial points from local 2-D unstable manifolds of $S_{0}$ (in our experiment, we choose 100 random points for the angle $\delta$ inside interval $(0,2\pi)$ and fix the size of vicinity as $\Omega=0.5$. We
repeat this procedure several times in order to get different initial points). From Fig.~\ref{fig:Projectionsx1x2x5} one can see that, these unstable manifolds form by their scrolling a tube around the circle of equilibria $S_{\theta}$. This experiment confirms that all trajectories which start from a vicinity of the unstable equilibria $S_{0}$ go to the circle of equilibria $S_{\theta}$, and this means that the green set in Fig.~\ref{fig:lorenz:attr:hid} is not self-excited. We can do the same experiment in the case of the torus as in Fig~\ref{fig:hiddem:1}. In Fig.~\ref{fig:unstablemanifoldestoruses} we plotted trajectories  with random initial points on local 2-D unstable manifolds of $S_{0}$, which tend to the torus. So the green set in Fig~\ref{fig:hiddem:1} is not self-excited. In order to demonstrate separation between basins of attraction of $S_{\theta}$ and the green set (as in Fig.~\ref{fig:lorenz:attr:hid}) which is located between the blue tube that is formed by scrolling of these trajectories and the domain which is defined by separatricies. We need not only to separate basin of attraction of this green set from these blue manifolds, but also from trajectory going from infinity to $S_{0}$ which can be derived analytically by formula $u(t) = (0, 0, 0, 0, z_0 \exp(-b t))$, $z_0\in \mathbb{R}$. Because of the system \eqref{13} is 5-dimensional it is difficult to show such separation, so we can say that the green set it is possible to be an attractor.
\begin{figure}[!ht]
  \centering
  \includegraphics[width=\columnwidth]{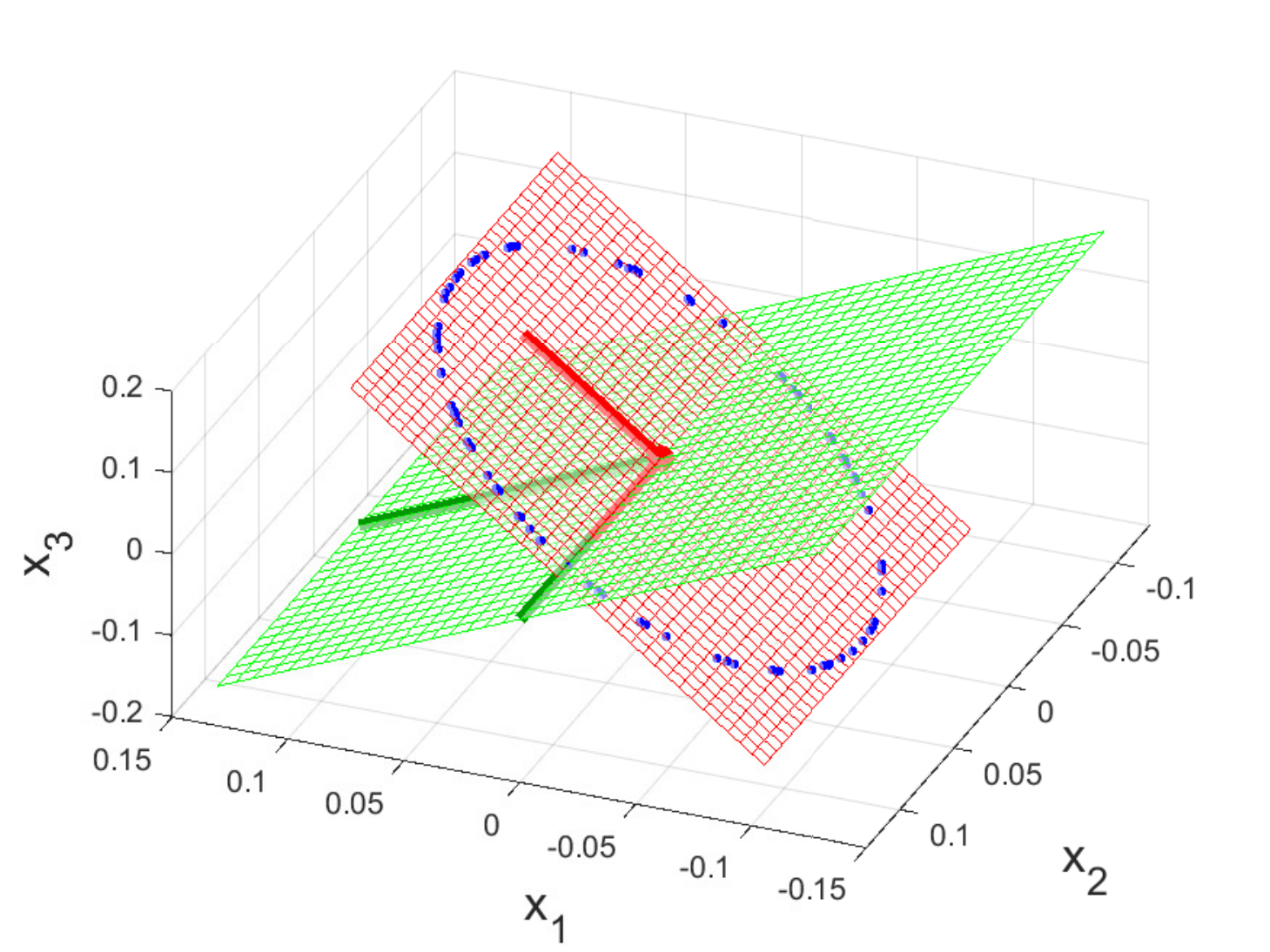}\\
  \caption{Visualization of local stable, unstable manifolds of $S_{0}$
  and random initial points on local 2-D unstable manifolds of $S_{0}$.
  (Red) local 2-D unstable manifold corresponding to $\lambda_{1,2}=10.6323 \pm 0.0003i$,
  (Green) local 2-D stable submanifold corresponding to $\lambda_{3,4}=21.6323 \pm 0.0007i$.}
   \label{fig:localstableunstabe}
\end{figure}
\begin{figure}[!ht]
  \centering
  \includegraphics[width=\columnwidth]{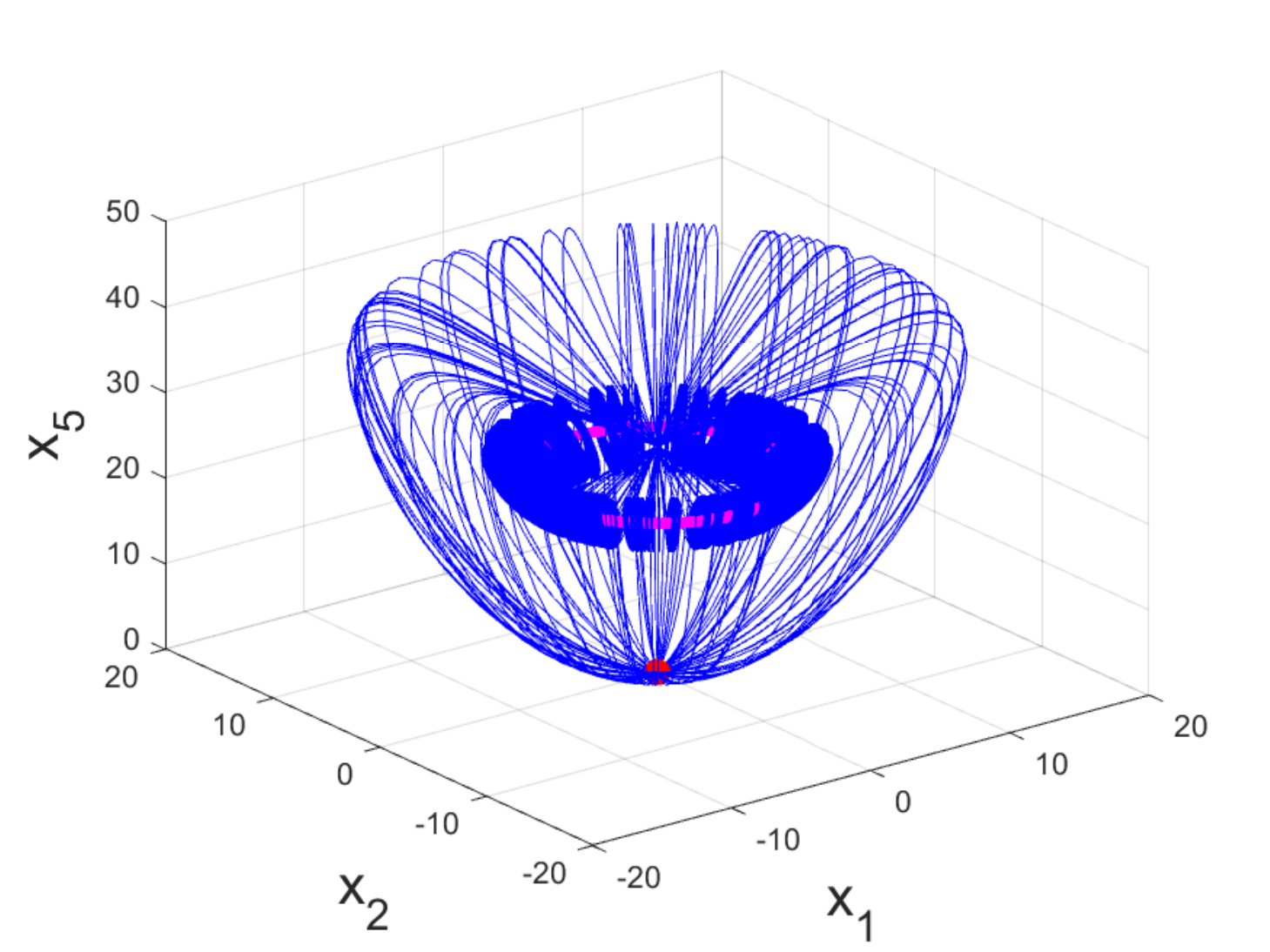}
  \caption{Visualization of unstable manifolds in system~\eqref{13} with
  $\sigma=10$, $r_{1}=24$, $r_{2}= -e = 0.001$, $b=\frac{8}{3}$
  and $100$ random initial points in a vicinity of $S_{0}$,
  that tend to the circle of equilibria $S_{\theta}$.}
  \label{fig:Projectionsx1x2x5}
\end{figure}
 \begin{figure}[!ht]
  \centering
  \includegraphics[width=\columnwidth]{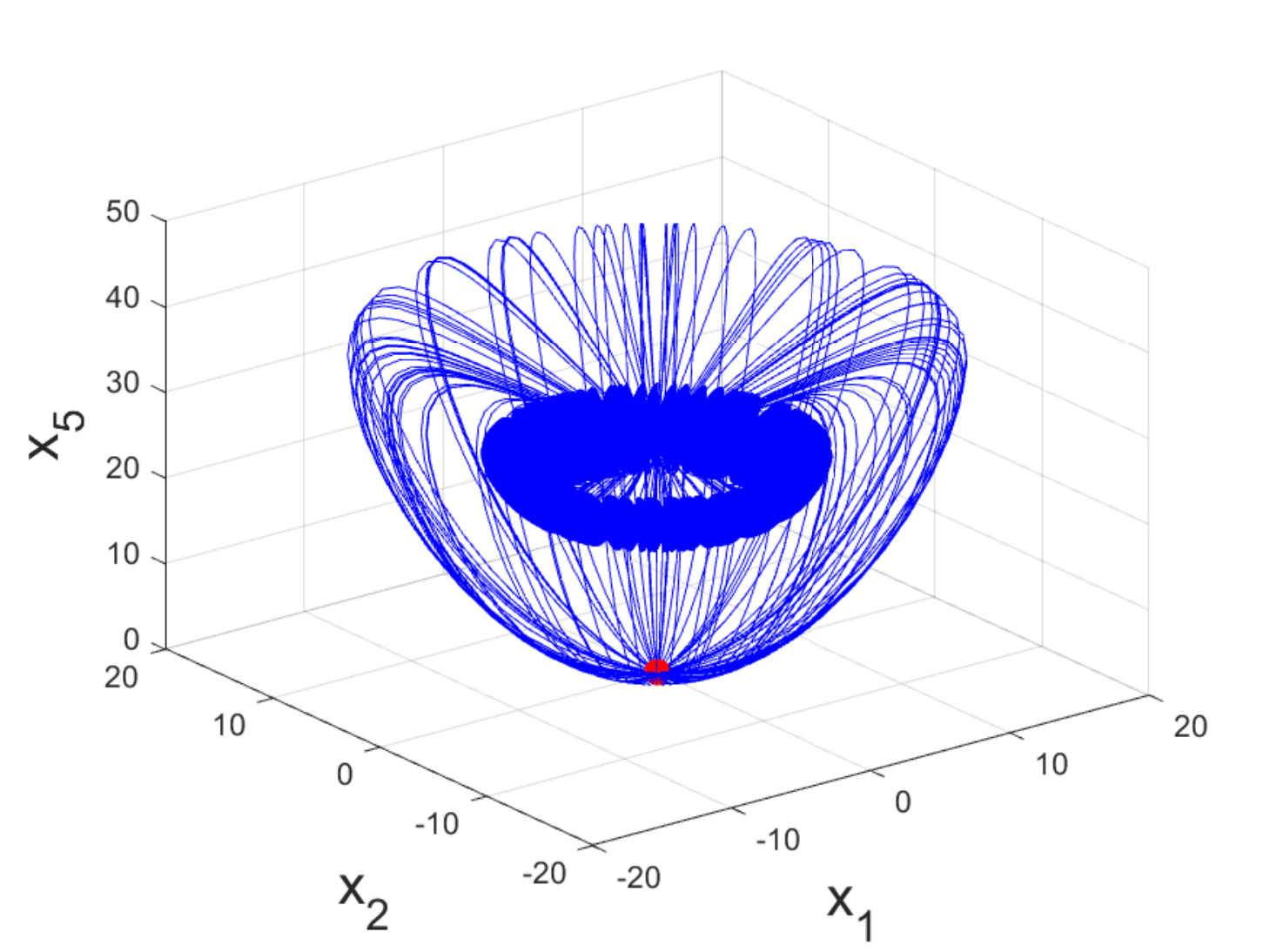}
  \caption{Visualization of trajectories in the time interval $[0, 7000]$
  with in system \eqref{13} with $\sigma=10$, $r_{1}=24$, $r_{2}=0.002$, $e=-0.001$, $b=\frac{8}{3}$
  and $100$ random initial points in a vicinity of $S_{0}$ tending to the torus.}
   \label{fig:unstablemanifoldestoruses}
\end{figure}

\section{Numerical verification of basins of attraction around the set of equilibria}
In order to numerically verify the attractiveness of the chaotic set
the following numerical experiment can be used.
In this experiment, we consider initial points on
the various line segments in the vicinity around
the set of equilibria $S_{\theta}$
in Figs.~\ref{fig:lorenz:attr:hid} and~\ref{fig:hiddem:1}
(see Fig.~\ref{fig:white region}).
From Fig.~\ref{fig:white region} we can conclude the following:
around the equilibria $S_{\theta}$ there is a vicinity such that
trajectories starting from this vicinity approach equilibria $S_{\theta}$
and outside it approach attractor.
For Fig.~\ref{fig:hiddem:1} the vertical and horizontal diagonals are considered
with radius $2.5$ and $2$ respectively and partition step $0.01$.
And for the others line segments as in Figs.~\ref{fig:notdiagonalfar}
and~\ref{fig:notdiagonalnear} we used the same partition step.
In Fig.~\ref{fig:twotrajec} we plotted two trajectories with initial points
from blue and red parts of vertical diagonal.
The same procedure can be used for Fig.~\ref{fig:lorenz:attr:hid} with vertical
and horizontal diagonals have radius $3$ and $2.5$ respectively and partition step $0.01$
on all considered line segments.
Fig.~\ref{fig:twotrajec1} shows two trajectories with starting points
from blue and red parts of vertical diagonal.
In addition, another procedure can be used to verify that the green
set in Fig.~\ref{fig:lorenz:attr:hid} not fall on the circle of equilibria $S_{\theta}$.
Around one scroll of trajectory of an unstable manifold from its outer sides
we consider line segments with length 5 and partition step 0.01 as in Fig.~\ref{fig:outerlines}.
In Fig.~\ref{fig:trajcfromouterlines} we plotted a trajectory with starting point
from the upper outer line segment.
\begin{figure*}[!ht]
 \centering
 \subfloat[
 {\scriptsize The starting trajectories from all points lie on blue parts and red parts of these diagonals approach attractor and equilibria respectively}
 ] {
 \label{fig:diagonallines}
 \includegraphics[width=0.3\textwidth]{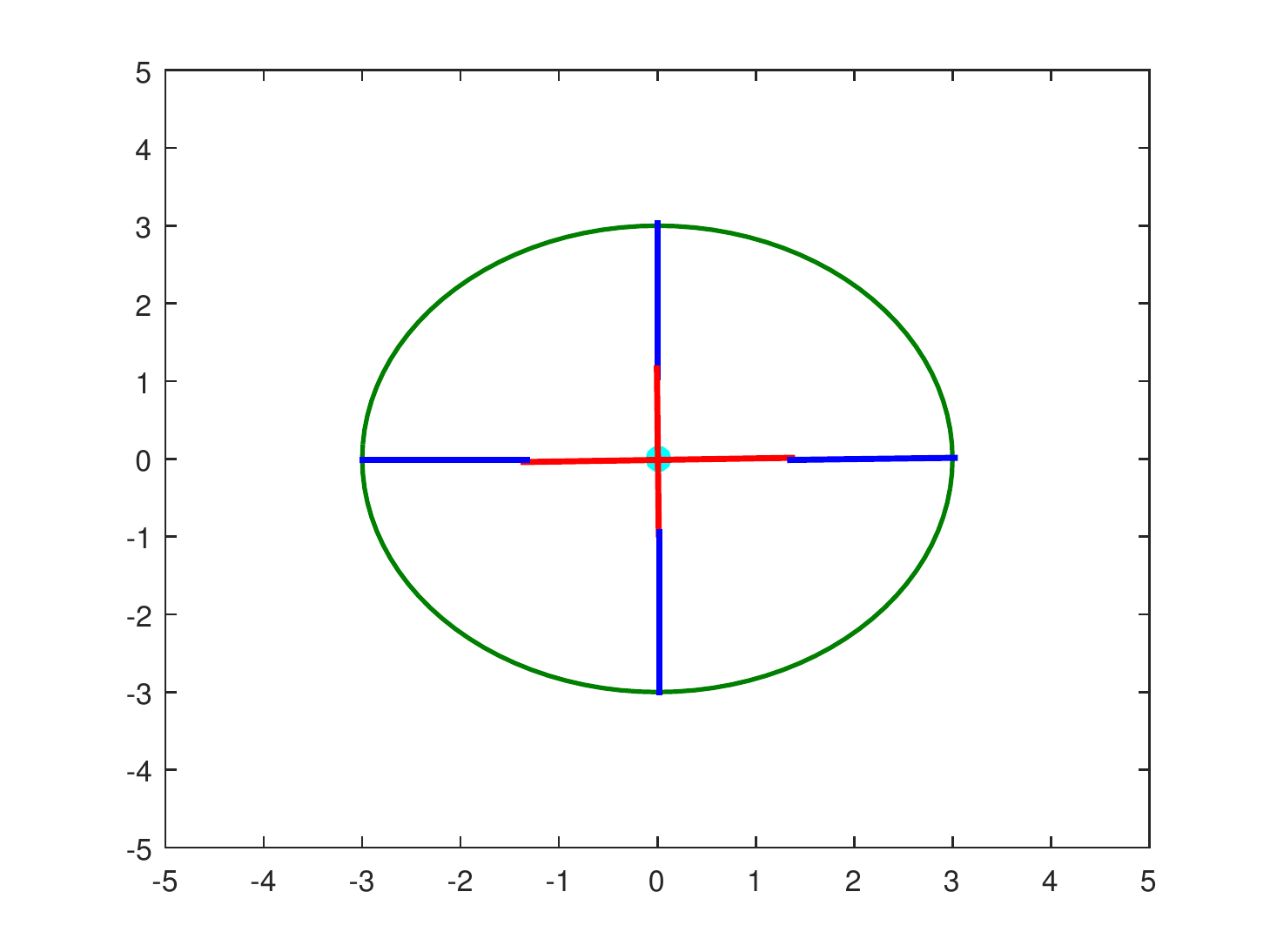}
 }~
 \subfloat[
 {\scriptsize The starting trajectories from all points lie on these blue line segments approach attractor  }
 ] {
 \label{fig:notdiagonalfar}
 \includegraphics[width=0.3\textwidth]{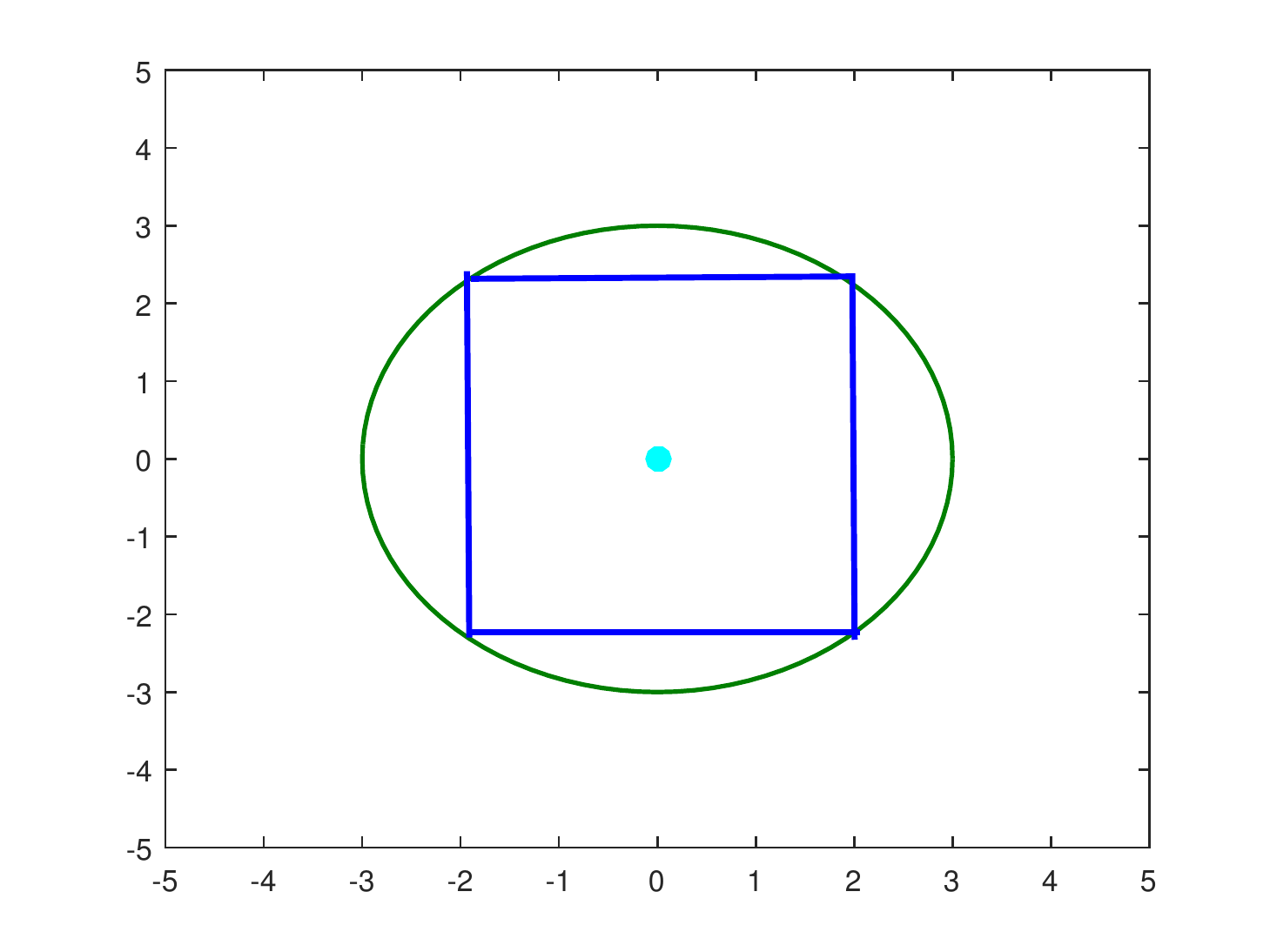}
 }
 \subfloat[{\scriptsize The starting trajectories from all points lie on these red line segments approach equilibria}
 ] {
 \label{fig:notdiagonalnear}
 \includegraphics[width=0.3\textwidth]{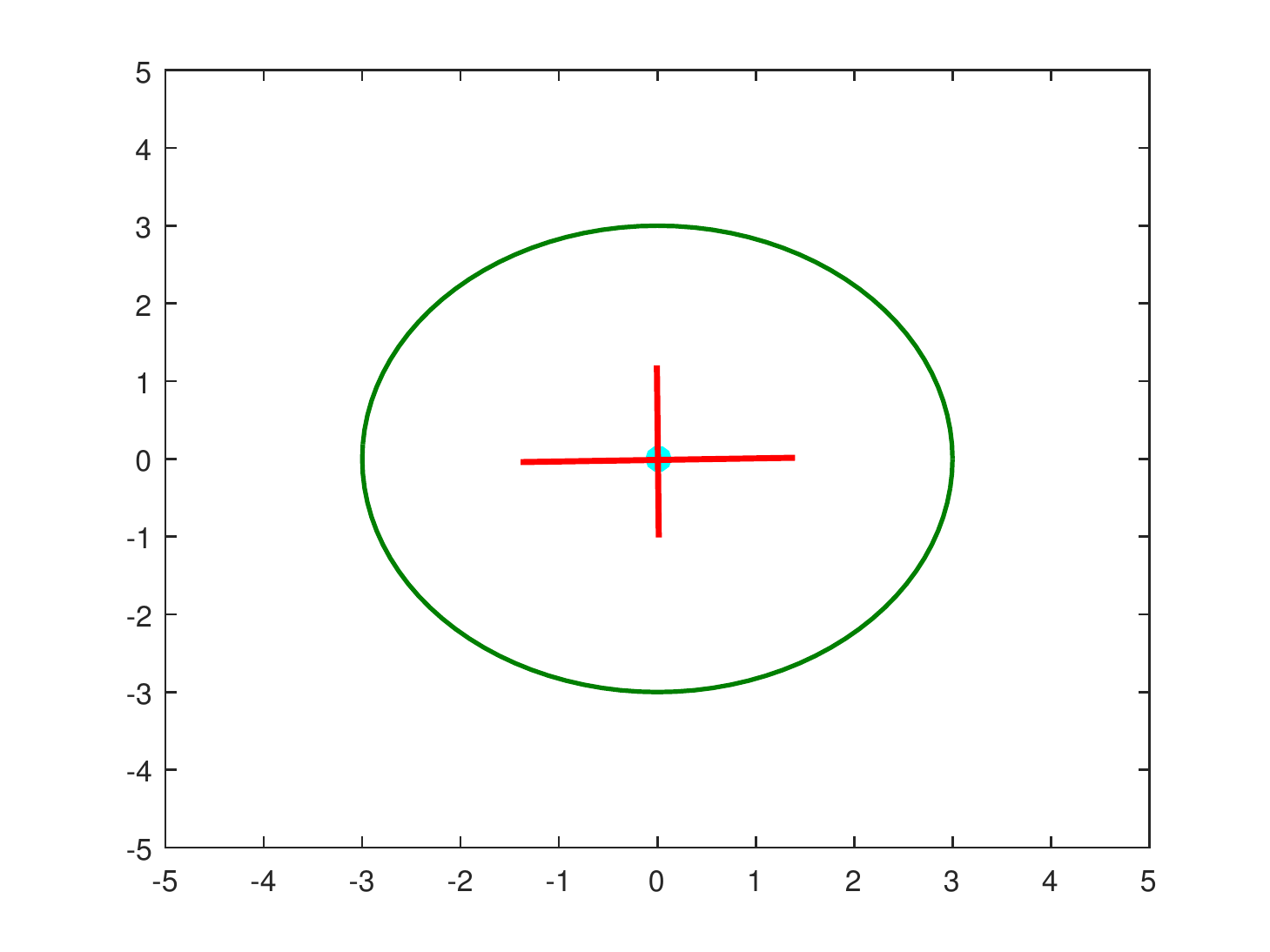}
 }
 \caption{
 Visualization of the white region around the set of equilibria $S_{\theta}$ in Figs.~\ref{fig:lorenz:attr:hid1} and ~\ref{fig:lorenz:attr:hid} and the considered line segments.
 }
 \label{fig:white region}
\end{figure*}
 \begin{figure*}[!ht]
  \centering
  \subfloat[
  {\scriptsize }
  ] {
  \label{fig:trajcfromblue}
  \includegraphics[width=0.4\textwidth]{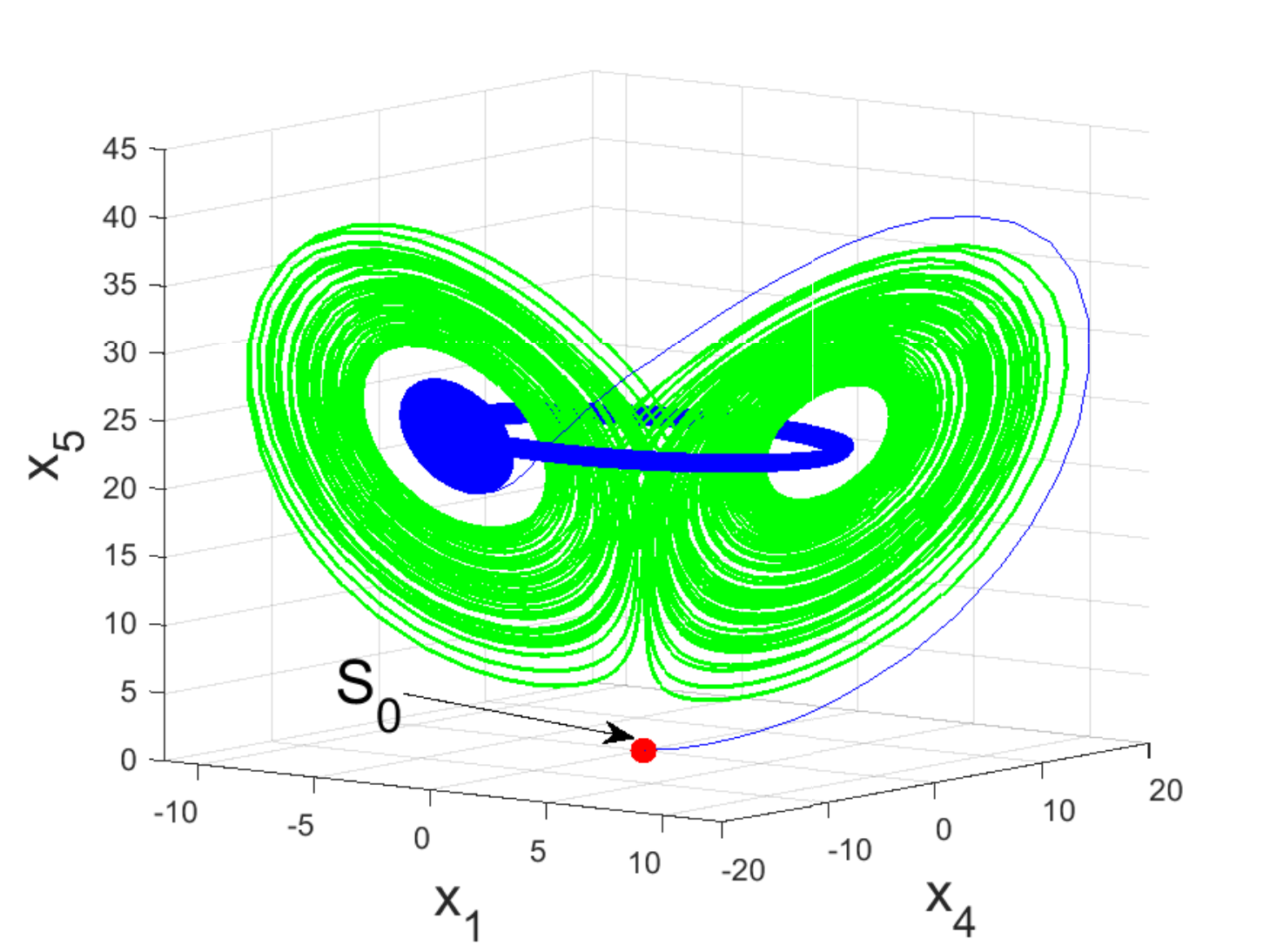}
  }~
  \subfloat[
  {\scriptsize }
  ] {
  \label{fig:trajcfromred}
  \includegraphics[width=0.4\textwidth]{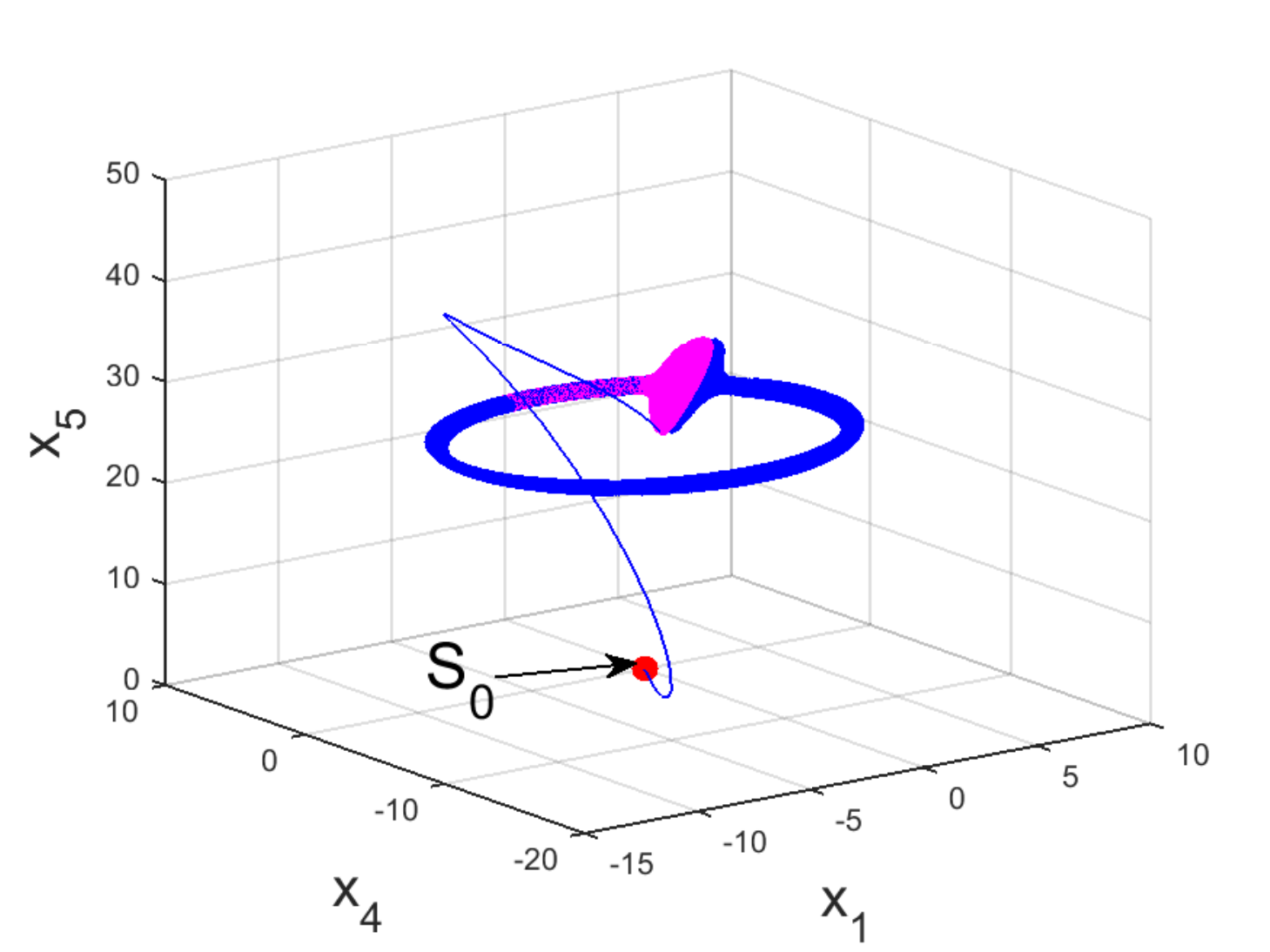}
  }
  \caption{
   Visualization of two trajectories with initial points from blue and red  parts of vertical diagonal in the white gape of Fig. \ref{fig:lorenz:attr:hid1}.
  }
  \label{fig:twotrajec}
 \end{figure*}
 \begin{figure*}[!ht]
  \centering
  \subfloat[
  {\scriptsize }
  ] {
  \label{fig:trajcfromblue1}
  \includegraphics[width=0.4\textwidth]{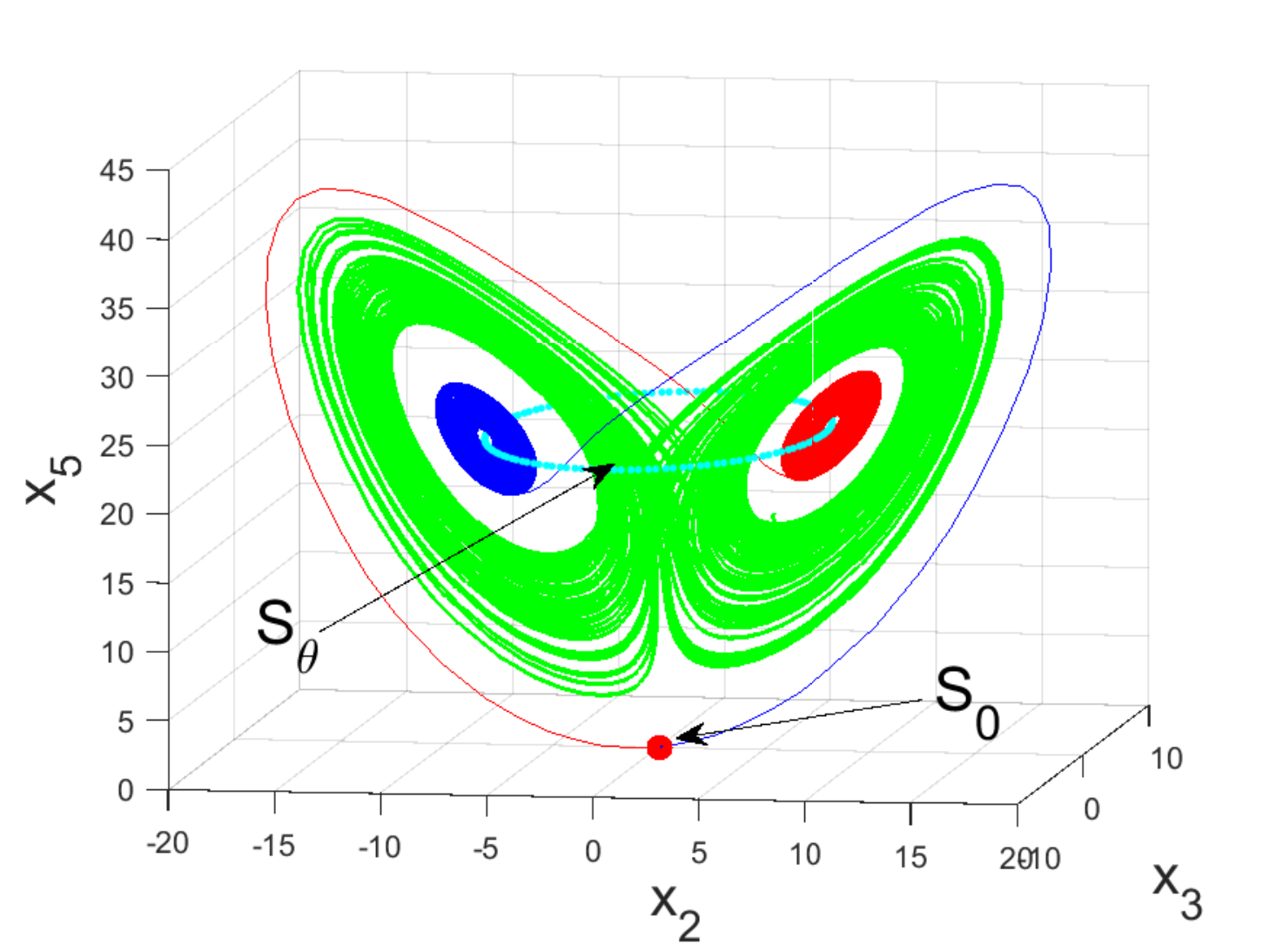}
  }~
  \subfloat[
  {\scriptsize }
  ] {
  \label{fig:trajcfromred1}
  \includegraphics[width=0.4\textwidth]{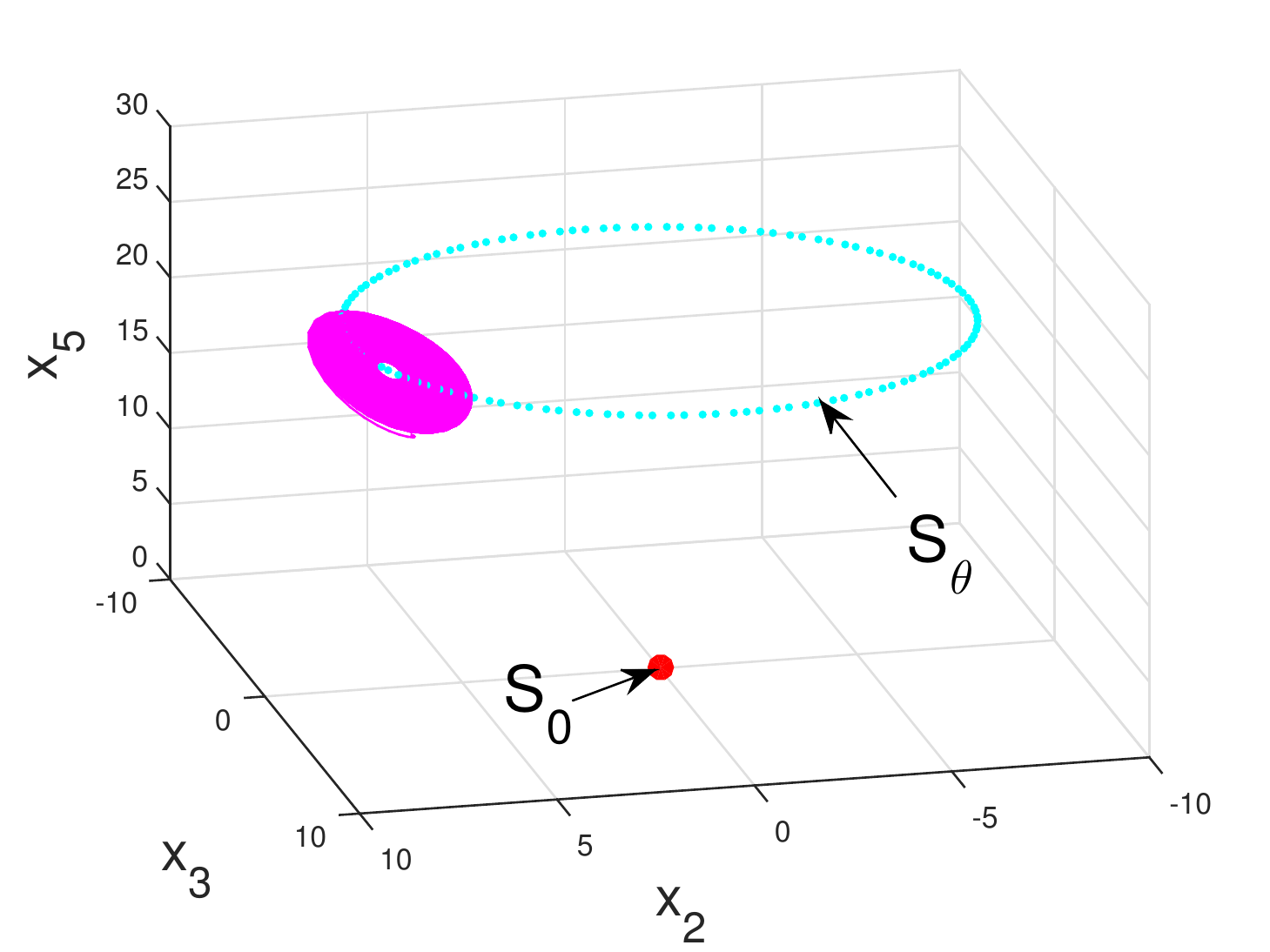}
  }
  \caption{
   Visualization of two trajectories with initial points from blue and red  parts of vertical diagonal in the white gape of Fig.\ref{fig:lorenz:attr:hid}.
  }
  \label{fig:twotrajec1}
 \end{figure*}
 \begin{figure*}[!ht]
  \centering
  \subfloat[
  {\scriptsize  }
  ] {
  \label{fig:outerlines}
  \includegraphics[width=0.4\textwidth]{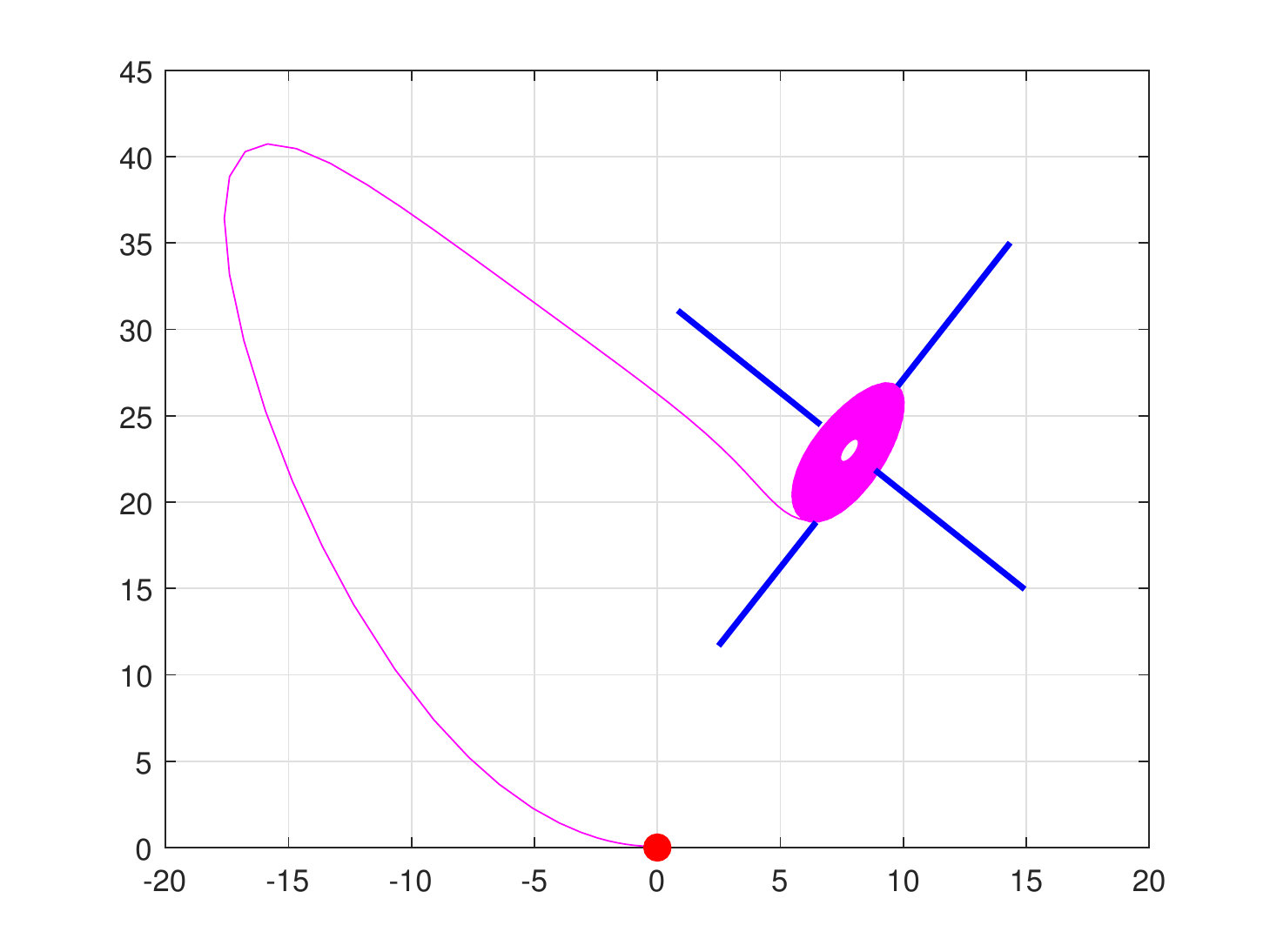}
  }~
  \subfloat[
  {\scriptsize }
  ] {
  \label{fig:trajcfromouterlines}
  \includegraphics[width=0.4\textwidth]{./trajcfromouterlines11.pdf}
  }
  \caption{
  (a) Visualization of line segments around the unstable manifold in Fig.~\ref{fig:lorenz:attr:hid1}, starting trajectories from all points on these line segments approach attractor;
  (b) Localization of a trajectory with initial point on the upper outer line segment.
  }
  \label{fig:twotrajec2}
 \end{figure*}\\\\

 \section{Attempts to prove the attractiveness of the transient set}
To demonstrate strictly the attractiveness of the green sets as in Figs.~\ref{fig:lorenz:attr:hid} and ~\ref{fig:hiddem:1}, another technique can be utilized. The idea of this technique is based on the fact that the complex Lorenz system \eqref{sys:complex-Lorenz} has a projective space, in which the states differing only by a common phase of variables $X$ and $Y$ are considered to be equivalent (see Section \ref{Inner:estimation}).
\par Now, it is important to mention the following properties
of the projective system \eqref{projec1-system} \cite{VladimirovTD-1997,VladimirovTD-1998-IJBC}:\\
(i) All physical information and dynamics in the phase space of the system \eqref{sys:complex-Lorenz} can be derived from the equations of motion \eqref{projec1-system} in the projective space.\\
(ii) The projective system \eqref{projec1-system} has a 1-dimensional unstable manifold and 3-dimensional stable manifold.
\par Note that the projective system \eqref{projec1-system} has the following two equilibria:
$S'_{0}=(0, 0, 0, 0)$  and $S'_{1}=(\xi'_{1},\upsilon'_{1}, w'_{1}, Z'_{1})$, where
\[
\xi'_{1}\!=\!\tfrac{\beta(\mu^{2}\!+ \kappa^{2})(\mu^{2}\!-\kappa^{2})}{2\mu^{4}(\beta \varrho+1)}, \
\upsilon'_{1}\!=\!0, \
w'_{1}\!=\!\tfrac{\beta\upsilon(\mu^{2}\!+\kappa^{2})}{\mu^{3}(\beta \varrho+1)}, \
Z'_{1}\!=\!\tfrac{\mu^{2}\!+\kappa^{2}}{\mu^{2}(\beta \varrho+1)}.
\]
The equilibria $S'_{0}$ and $S'_{1}$ in the projective space represent the projections of the equilibria $S_{0}$ and $S_{\theta}$ of the system \eqref{13} respectively. For the torus as in Fig.~\ref{fig:self:torus} its projection in the projective space $\mathcal{P}$ is a limit cycle $\Gamma$.

Because of the projection system \eqref{projec1-system} has a $1$-dimensional unstable manifold
and the projective mapping \eqref{projec1-coordinates} preserves the ''chaoticity'' of attractors.
So, it is reasonable to use the projection system \eqref{projec1-system} to identify whether
the chaotic sets in Figs.~\ref{fig:lorenz:attr:hid} and ~\ref{fig:hiddem:1}
are attractors or transient chaotic sets.
For the parameters $\sigma=10, r_{1}=24, r_{2}=0.001, e=-0.001, b=\frac{8}{3}$
(see Fig.~\ref{fig:lorenz:attr:hid}), around the equilibria $S'_{0}$ we choose
a small spherical vicinity of radius 0.002 and consider 100 random initial points on it.
From Fig.~\ref{fig:projective} one can observe the following: all the considered unstable manifolds
of $S'_{0}$ go in one direction and approach $S'_{1}$.
We did the same experiment for the case of the torus and got the same conclusion
(see Fig.~\ref{fig:projective:torus}).
 \begin{figure*}[!ht]
  \centering
  \subfloat[
  {\scriptsize }
  ] {
  \label{fig:greentrajc_projective}
  \includegraphics[width=0.45\textwidth]{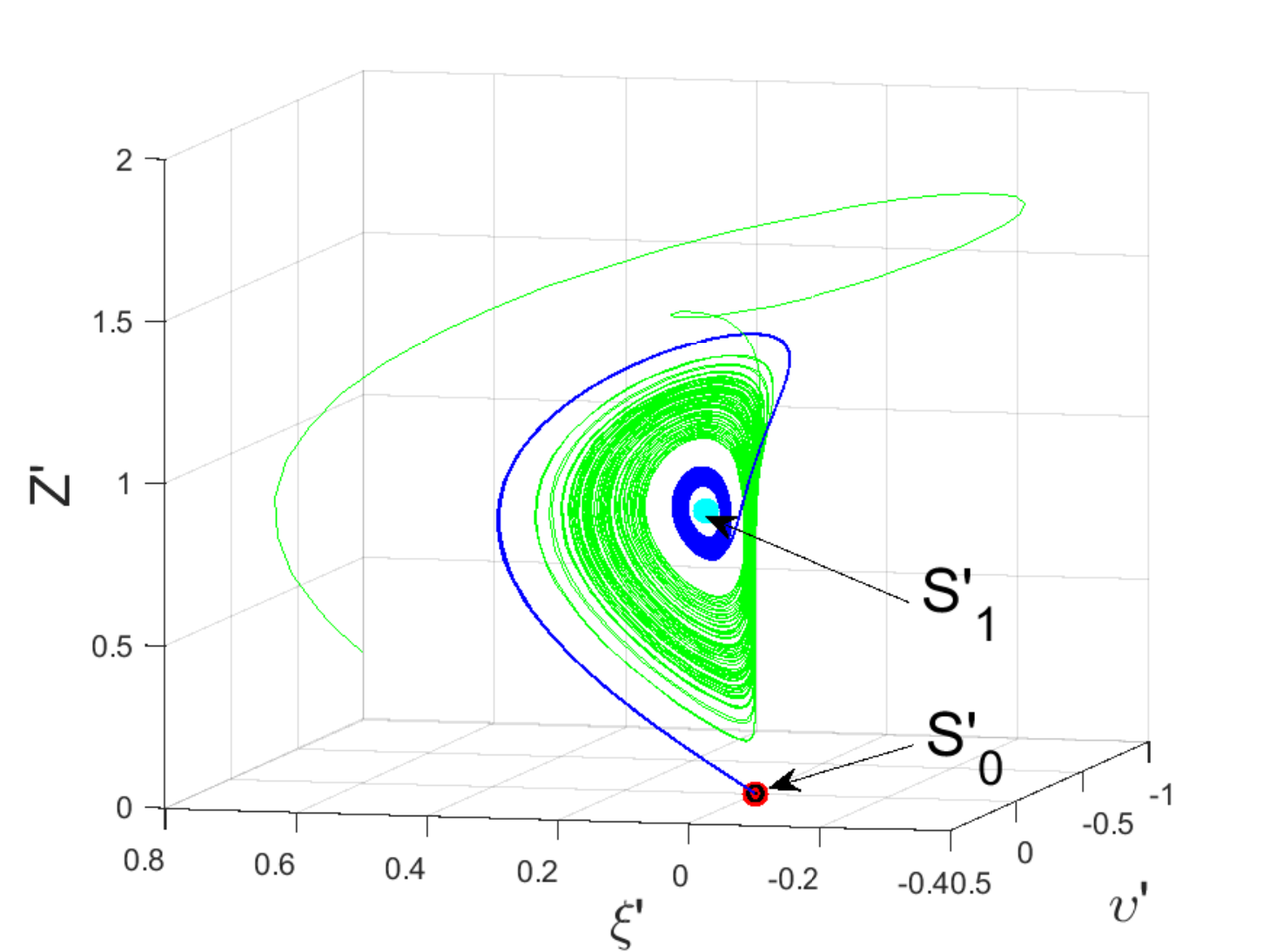}
  }~
  \subfloat[
  {\scriptsize }
  ] {
  \label{fig:sphere_projective}
  \includegraphics[width=0.45\textwidth]{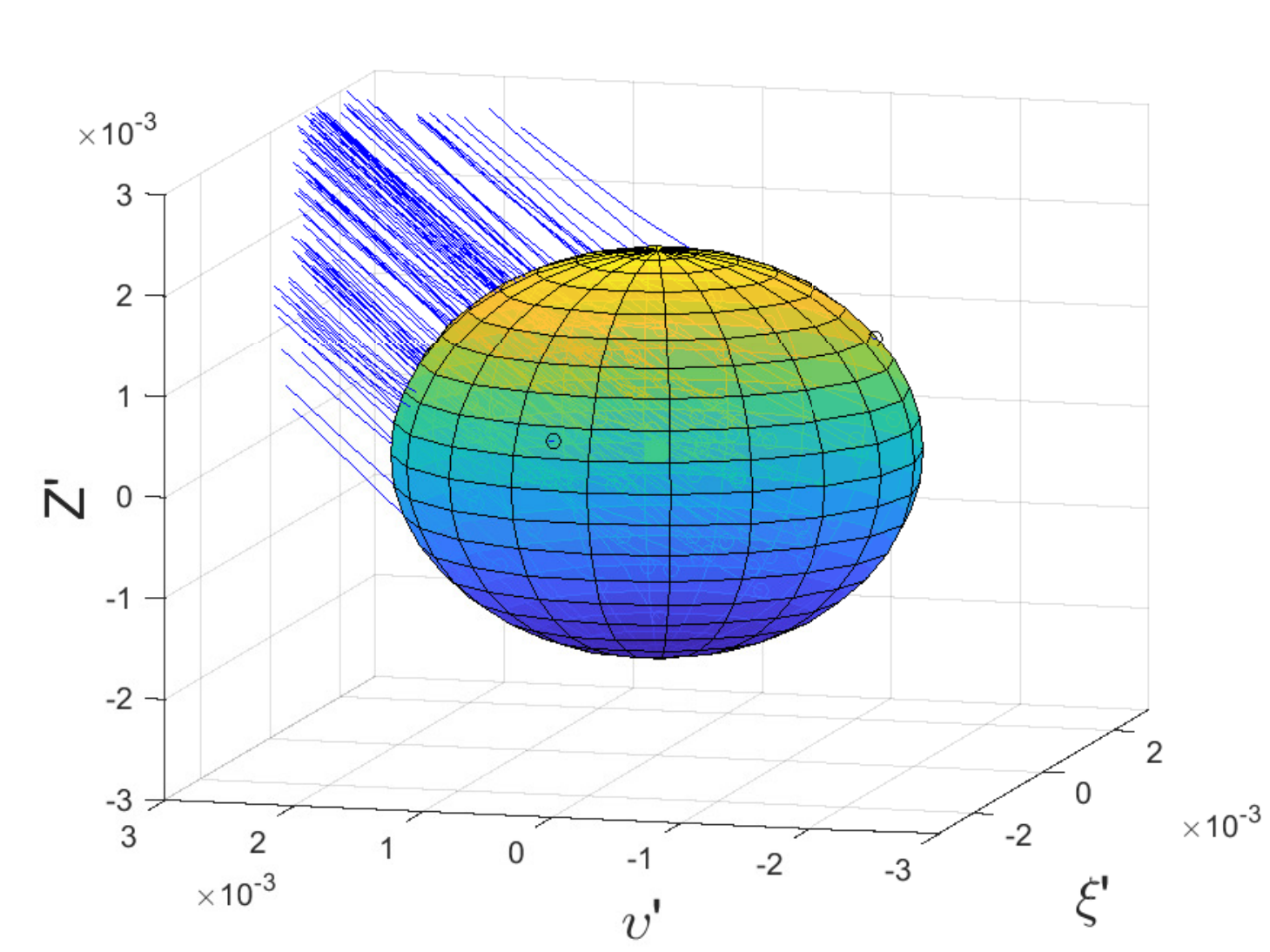}
  }
  \caption{(a) Localization of the chaotic set (green) in the projective space and $100$ unstable manifolds
  with starting points from a spherical of random points in a vicinity of $S'_{0}$,
  with $(\sigma, r_{1}, r_{2}, e, b)=(10, 24, 0.001, -0.001, \frac{8}{3})$;\\
  (b) Magnification of the area around $S'_{0}$ to show the sphere of random initial points
  and the corresponding $100$ trajectories tending to $S'_{1}$.}
  \label{fig:projective}
 \end{figure*}
 \begin{figure*}[!ht]
  \centering
  \subfloat[
  {\scriptsize }
  ] {
  \label{fig:greentrajctorus_projective}
  \includegraphics[width=0.45\textwidth]{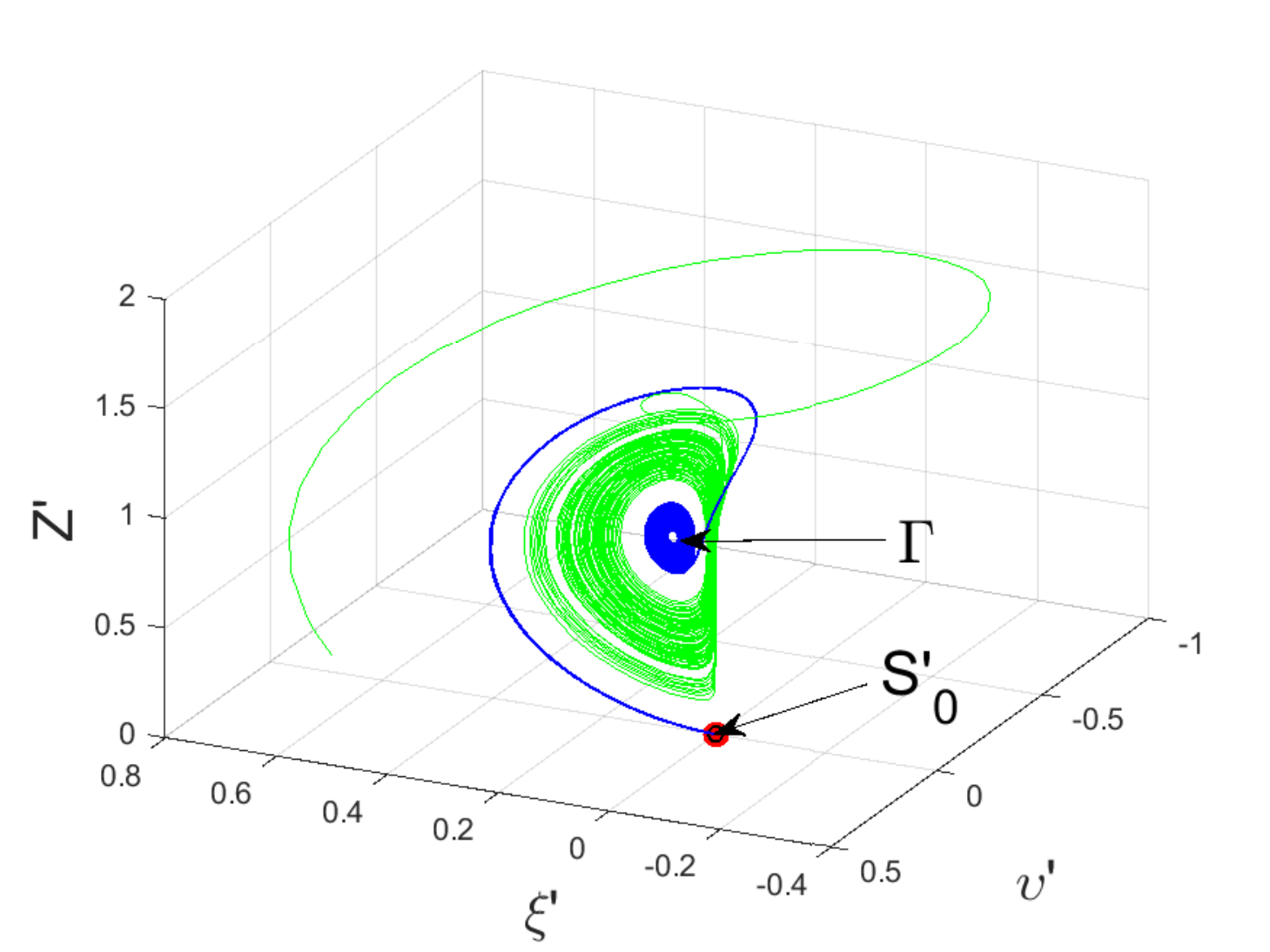}
  }~
  \subfloat[
  {\scriptsize }
  ] {
  \label{fig:spheretorus_projective}
  \includegraphics[width=0.45\textwidth]{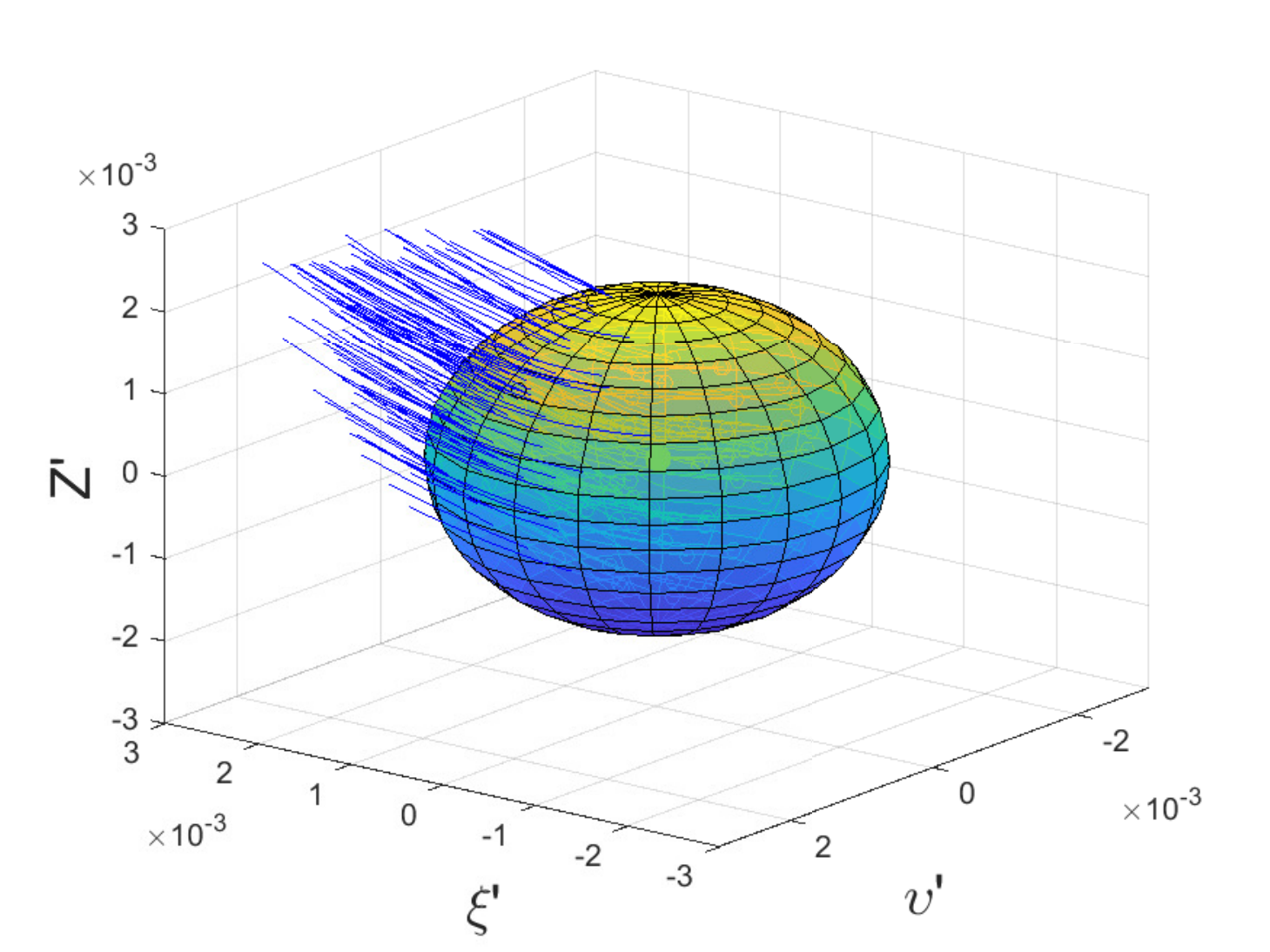}
  }
  \caption{(a) Localization of the chaotic set (green) in the projective space and $100$ unstable manifolds with starting points from a spherical of random points in a vicinity of $S'_{0}$,
  $(\sigma, r_{1}, r_{2}, e, b)=(10, 24, 0.002, -0.001, \frac{8}{3})$; \\
  (b) Magnification of the area around $S'_{0}$ to show the spherical of random points
  and the considered $100$ trajectories tending to the limit cycle $\Gamma$.}
  \label{fig:projective:torus}
 \end{figure*}
\end{document}